\documentclass[12pt]{article}
\usepackage{amsmath,amsxtra,amssymb,amsthm,amsfonts,eufrak,bm,mathtools}

\usepackage{tikz-cd}

\makeatletter
\newcommand*{\rom}[1]{\expandafter\@slowromancap\romannumeral #1@}
\makeatother
\usetikzlibrary{matrix,arrows,decorations.pathmorphing}

\newcommand{\bxi}{{\bm \xi}}

\newcommand{\bphi}{{\bm \phi}}
\newcommand{\bpsi}{{\bm \psi}}
\newcommand{\bet}{{\bm \eta}}
\newcommand{\brho}{{\bm \rho}}
\newcommand{\bsi}{{\bm \sigma}}
\newcommand{\bzeta}{{\bm \zeta}}
\newcommand{\bome}{{\bm \omega}}

\newcommand{\R}{{\mathbb R}}

\newcommand{\ten}{\otimes}

\newcommand{\pl}{\hspace{.1cm}}

\newcommand{\ran}{\rangle}
\newcommand{\lan}{\langle}
\newcommand{\al}{\alpha}
\newcommand{\si}{\sigma}

\newcommand{\la}{\lambda}
\newcommand{\eps}{\varepsilon}

\newcommand{\id}{\iota_{\infty,2}^n}

\newcommand{\M}{{\mathcal M}}

\newcommand{\N}{{\mathcal N}}

\newcommand{\norm}[2]{\left\Vert  #1  \right\Vert_{#2}}

\newcommand{\tr}{\operatorname{tr}}
\newtheorem{lemma}{Lemma}[section]
\newtheorem{prop}[lemma]{Proposition}
\newtheorem{theorem}[lemma]{Theorem}
\newtheorem{cor}[lemma]{Corollary}

\newtheorem{rem}[lemma]{Remark}

\newcommand{\re}{\begin{rem}\rm}
\newcommand{\mar}{\end{rem}}
\newtheorem{exam}[lemma]{Example}
\newcommand{\bra}[1]{\langle{#1}|}
\newcommand{\ket}[1]{|{#1}\rangle}

\newtheorem{definition}[lemma]{Definition}
\newcommand{\qd}{\end{proof}\vspace{0.5ex}}
\newcommand{\prf}{\begin{proof}[\bf Proof:]}

\newcommand{\xspace}{\hbox{\kern-2.5pt}}

\renewcommand{\R}{\mathbb{R}}

\renewcommand{\id}{\operatorname{id}}

\allowdisplaybreaks

\usepackage[ colorlinks = true,              linkcolor = blue, urlcolor  =
blue,              citecolor = red,              anchorcolor = green,
]{hyperref}

\begin{document}

\title{Recoverability for optimized quantum $f$-divergences}
\date{}

\author{Li Gao\thanks{Department of Mathematics,
Texas A\&M University, College Station, Texas 77840, USA; Email: ligao@math.tamu.edu} \and Mark M.~Wilde\thanks{Hearne Institute for Theoretical Physics, Department of Physics and Astronomy,
and Center for Computation and Technology, Louisiana State University, Baton Rouge, Louisiana 70803, USA; Email: mwilde@lsu.edu} \thanks{Stanford Institute for Theoretical Physics, Stanford University, Stanford, California 94305, USA}}

\maketitle

\begin{abstract}
The optimized quantum $f$-divergences form a family of distinguishability measures that includes  the quantum relative entropy and the sandwiched R\'enyi relative quasi-entropy as special cases. In this paper, we establish physically meaningful refinements of the data-processing inequality for the optimized $f$-divergence. In particular, the refinements state that the absolute difference between the optimized $f$-divergence and its channel-processed version is an upper bound on how well one can recover a quantum state acted upon by a quantum channel, whenever the recovery channel is taken to be a rotated Petz recovery channel. Not only do these results lead to physically meaningful refinements of the data-processing inequality for the sandwiched R\'enyi relative entropy, but they also have implications for perfect reversibility (i.e., quantum sufficiency) of the optimized $f$-divergences. Along the way, we improve upon previous physically meaningful refinements of the data-processing inequality for the standard $f$-divergence, as established in recent work of Carlen and Vershynina [arXiv:1710.02409, arXiv:1710.08080]. Finally, we extend the definition of the optimized $f$-divergence, its data-processing inequality, and all of our recoverability results to the general von Neumann algebraic setting, so that all of our results can be employed in physical settings beyond those confined to the most common finite-dimensional setting of interest in quantum information theory.
\end{abstract}

\tableofcontents

\section{Introduction}

The quantum relative entropy is a fundamental measure in quantum
information theory. It was first introduced by Umegaki \cite{Umegaki} as a noncommutative generalization of the classical relative entropy (the latter is also called
Kullback--Leibler divergence \cite{KL}). For two quantum states described by density operators $\rho$ and $\si$, the relative entropy of $\rho$ with respect to $\si$ is defined as
\[
D(\rho\|\si) \coloneqq \tr(\rho\log \rho-\rho\log \si)\pl,
\]
where $\tr$ denotes the matrix trace.  The relative entropy $D(\rho\|\si)$ measures how well the quantum state $\rho$ can be distinguished from $\si$
in an asymptotic setting of quantum hypothesis testing \cite{HP91,OgawaNagaoka00}.
One of its most important properties  is the data-processing inequality \cite{Lin75, Uhlmann77}: for all quantum channels $\Phi$ and states $\rho$ and $\si$, the following inequality holds
\begin{equation}
\label{eq:DP-q-rel-ent}
D(\rho\|\si)\ge D(\Phi(\rho)\|\Phi(\si)).
\end{equation}
As the quantum relative entropy is a distinguishability measure,
 the data-processing inequality asserts that two quantum states cannot become more distinguishable after applying the same quantum channel
to them. The data-processing inequality is a key principle underlying the widespread applications of quantum relative entropy in quantum information~\cite{Vedral02,Wilde17}.

The wide interest in relative entropy has sparked researchers to study other entropy-type measures that also satisfy the data-processing inequality. Important generalizations in classical information theory are the R\'enyi relative entropy \cite{renyi} and the more general notion of $f$-divergence \cite{Csi67, AS66, Mor63}. For two probability distributions
$\{p(x)\}_x$  and $\{q(x)\}_x$ and a convex function $f$, the classical $f$-divergence \cite{Csi67, AS66, Mor63} is defined as
\[
S_f(p\|q)\coloneqq \sum_{x}p(x)f\!\left(\frac{q(x)}{p(x)}\right)\pl,
\]
and it satisfies the data-processing inequality for classical channels.
In \cite{petz85,petz86}, Petz introduced a quantum version of the $f$-divergence and proved that the quantum $f$-divergence satisfies the data-processing inequality whenever the underlying function~$f$ is \emph{operator convex}. One notable example is the Petz--R\'enyi relative quasi-entropy \cite{petz85,petz86}, which corresponds to $f(t)=t^s$ for $s\in (-1,0)\cup(0,1)$, i.e., the power function. From this quantity, the Petz--R\'enyi relative entropy can be defined, and it has an operational interpretation in quantum hypothesis testing \cite{Nag06, Hay07}.

In recent years, the sandwiched R\'enyi relative entropy \cite{Muller13,WWY} was introduced as another quantum generalization of R\'enyi relative entropy and has found extensive application in establishing
strong converse results for communication tasks \cite{WWY, GW15, TWW17, CMW16, DW18, WTB17}. It also has a direct operational meaning in quantum hypothesis testing in terms of the strong converse exponent \cite{MO13}. While Petz's definition of quantum $f$-divergence from \cite{petz85,petz86} is often called the standard $f$-divergence, it was not clear how to express the sandwiched R\'enyi relative entropy in terms of a standard $f$-divergence. This problem was solved in~\cite{wilde} with the introduction of a different type of quantum $f$-divergence called the optimized $f$-divergence.
It was also proved in \cite{wilde} that the optimized $f$-divergence satisfies the data-processing inequality for an operator anti-monotone function $f$.

Over decades, the data-processing inequality of the quantum relative entropy has been  refined in various ways. Petz proved that the data-processing inequality in \eqref{eq:DP-q-rel-ent} is saturated, i.e., $D(\rho\|\si)=D(\Phi(\rho)\|\Phi(\si))$, if and only if there exists a quantum recovery channel~$R$ satisfying $(R\circ \Phi)(\rho)=\rho$ and $(R\circ\Phi)(\si)=\si$ \cite{petz86suff,petz88}. The latter condition is also called ``quantum sufficiency'' \cite{petz86suff} because it indicates that the pair $(\Phi(\rho),\Phi(\si))$ is just as good as the pair $(\rho,\si)$ in a distinguishability experiment.
Moreover, there is a canonical choice of the recovery channel $R$, now called the Petz recovery map, which is given by \begin{align}R_{\Phi,\si}(x)
\coloneqq
\si^{1/2}\Phi^\dag(\Phi(\si)^{-1/2}x \Phi(\si)^{-1/2})\si^{1/2}\label{petz},\end{align}
where $\Phi^\dag$ is the adjoint of $\Phi$ with respect to the Hilbert--Schmidt inner product.

More recently, much progress has been made on the  case of approximate recovery. The idea is that when the data-processing inequality is nearly saturated, then the states $(\rho,\si)$ can be approximately recovered from $(\Phi(\rho),\Phi(\si))$ by the action of some quantum channel~$R$. The first precise quantitative result of approximate recovery was obtained in \cite{FR} for the special case of $\Phi$ being a partial trace and $\si$ being a marginal of $\rho$ (this specialized setting is relevant for an information measure called conditional mutual information). The result of \cite{FR} has been generalized in \cite{Wilde15,STH16,universal,SBT17}. In particular, it was proved in \cite{universal} that the following inequality holds for a universal recovery map $R$:
\begin{equation}
D(\rho\|\si) \geq D(\Phi(\rho)\|\Phi(\si))
-\log F(\rho, (R \circ \Phi)(\rho)), \label{eq:recov-univ}
\end{equation}
while the equality $(R \circ \Phi)(\sigma) = \sigma$ holds also. In \eqref{eq:recov-univ} above, $F$ denotes the Uhlmann fidelity~\cite{Uhlmann76} (defined later in \eqref{eq:uhlmann-fidelity}) and
the recovery map $R$ is explicitly given as follows:
\begin{equation}
R \coloneqq \int_{\mathbb{\R}}R^\frac{t}{2}_{\Phi,\si}\ d\beta(t)\pl,
\qquad R^t_{\Phi,\si}(x)
\coloneqq
\si^{-it}R_{\Phi,\si}(\Phi(\si)^{it}x\Phi(\si)^{-it})\si^{it}  \pl,
\label{eq:recov-map-univ}
\end{equation}
where $R_{\Phi,\si}$ is the original Petz map in \eqref{petz}, $R^t_{\Phi,\si}(x)$ is called a rotated Petz map \cite{Wilde15}, and $R$ is the expectation of $R_t$ with respect to the following probability density function:
\[
d\beta(t)=\frac{\pi}{2}(\cosh(\pi t)+1)^{-1}dt.
\]
The recovery map $R$ in \eqref{eq:recov-map-univ} is said to be ``universal'' because it does not depend on the $\rho$ state; this property is useful in a variety of physical applications such as quantum error correction~\cite{universal}. Note that a slightly stronger inequality than the one in \eqref{eq:recov-univ} is available in~\cite{universal}.

Most recently, the main result of \cite{universal} has been extended to the von Neumann algebraic setting \cite{FHSW20}, and Refs.~\cite{CV18,Vershynina19} established an approximate recovery estimate for the original Petz map. The method of \cite{CV18,Vershynina19} is based on the integral representation of operator convex functions and further applies to the case of approximate recoverability for standard $f$-divergences, as well as to the case of Petz--R\'enyi relative entropies. A similar method has been employed to understand refinements of the data-processing inequality for the maximal $f$-divergences \cite{BC20}.

\section{Summary of results}

\label{sec:summary-results}

In this paper, we study approximate recoverability for optimized $f$-divergences and contribute the following findings:

\begin{enumerate}

\item  We prove that the difference of optimized $f$-divergences before and after the action of a quantum channel is an upper bound on the recoverability error for rotated Petz recovery maps (see Lemma~\ref{fd3}). Since the sandwiched R\'enyi relative quasi-entropy is a special kind of optimized $f$-divergence~\cite{wilde}, our result gives the first quantitative estimate for approximate recoverability with respect to the sandwiched R\'enyi relative (quasi-)entropies (see Theorem~\ref{fd4} and Corollary~\ref{renyi}).
The method that we employ here is inspired by \cite{CV17,CV18,Vershynina19}.

\item As a corollary, we find the following reversibility result:  if the optimized $f$-divergence is preserved under the action of a quantum channel, then every rotated Petz map is a perfect recovery map (see Corollary~\ref{cor:many-equivs-f-div}). This extends previous reversibility results found for the sandwiched R\'enyi relative entropy  \cite{Jenvcova17,HM17} (see also \cite{LRD17,chehade20,zhang20} for related conditions regarding the saturation of the data-processing inequality for the sandwiched R\'enyi relative entropy).

\item We also improve the results of \cite{CV17,CV18} for the quantum and Petz--R\'enyi relative entropies and further generalize these prior results to rotated Petz maps (see Theorems~\ref{re} and \ref{thm:quasi-renyi-recoverability}, Corollary~\ref{cor:petz-renyi-recoverability},  Theorems~\ref{rsi} and \ref{thm:other-petz-map-renyi-rec}, and Corollary~\ref{cor:petz-renyi-recoverability-2}). One advantage of these new bounds over the previous ones from \cite{CV17,CV18} is that the remainder term involves the Petz--R\'enyi relative entropy of order two, rather than the operator norm of the relative modular operator. As such, these bounds are non-trivial for the important class of bosonic Gaussian states \cite{S17}, whereas the previous bounds from \cite{CV17,CV18} do not apply for this class of states.

\item Motivated by the recent works on quantum $f$-divergences in general von Neumann algebras \cite{Hiai18,Hiai19},
we extend the definition of optimized $f$-divergence, its data-processing inequality, and our recoverability results to the general context of von Neumann algebras (see Definition~\ref{def:optimized-f-vNa}, Theorems~\ref{data} and \ref{thm:recoverability-vNa}, and Corollary~\ref{cor:reversibility-vNa}). Our results also provide a new way for understanding the sandwiched R\'enyi relative entropy in the von Neumann algebraic setting. Note that the sandwiched R\'enyi relative entropy was previously defined and analyzed in the von Neumann algebraic setting \cite{BVT18, Jenvcova18, Jenvcova2}. Later on, it was analyzed under a different approach \cite{GZZ19} and studied in the context of conformal field theory \cite{Lashkari19}.
\end{enumerate}

The rest of our paper is organized as follows. Section~\ref{sec:prelims} reviews the basic definitions of operator monotone and operator convex functions, quantum (optimized) $f$-divergences, and (rotated) Petz recovery maps. In Section~\ref{sec:finite}, we discuss our main recoverability results in the finite-dimensional setting, while focusing on quantum channels that act as restrictions to a subalgebra. This is the core case, and the  argument here avoids technicalities that occur in infinite dimensions. We prove that the recoverability error for a rotated Petz recovery map can be bounded from above by a difference of (optimized) $f$-divergences.
Section~\ref{sec:opt-f-div} is devoted to the optimized $f$-divergence in general von Neumann algebras. We prove the data-processing inequality and extend our recoverability results to a general quantum channel in this setting.

\section{Preliminaries}

\label{sec:prelims}

\subsection{Operator convex functions and operator monotone functions}

We briefly review the integral representation of operator monotone and operator convex functions. We refer to \cite{bhatia} 
for more information on this topic.

Let $B(H)$ denote the set of bounded operators acting on a  Hilbert space $H$. An operator $A\in B(H)$ is positive if $\bra{v} A \ket{v} \geq 0$ for all $\ket{v} \in H$. Let $B(H)^+$ denote the set of positive operators.
A function $f:(0,\infty)\to \R$ is operator monotone if the following inequality holds for all invertible positive operators $A,B\in B(H)$ satisfying~$A\le B$:
\[ f(A)\le f(B) \pl. \]
We say that $f$ is operator convex if the following inequality holds for all invertible positive operators $A,B$ and $\la \in [0,1]$:
\[f(\la A+(1-\la)B)\le \la f(A)+(1-\la)f(B) \pl.\]
We say that $f$ is operator anti-monotone (resp.~concave) if $-f$ is operator monotone (resp.~convex).
It is known that  $f:(0,\infty)\to \R$ is operator concave if it is operator monotone.
A function $f:(0,\infty)\to \R$ is operator convex if and only if for all invertible positive operators $A\in B(H)^+$ and Hilbert-space isometries $V:K\to H$, the following inequality holds
\[ V^*f(A)V\ge f(V^*AV)\pl.\]
This inequality is known as the operator Jensen inequality and also extends to positive  $A$ (see \cite[Appendix]{petz85}, as well as \cite{HP03}).

By the L\"owner Theorem (c.f. \cite[p.~144]{bhatia}), an operator monotone function $f :(0,\infty)\to \R$ admits the following integral representation:
\begin{align}
\label{om}
f(t)=a+bt+
\int_{0}^\infty\left(\frac{\la}{\la^2+1}-\frac{1}{\la+t}\right) d\nu(\la),
\end{align}
where $a\in \R$, $b\ge 0$, and $\nu$ is a positive measure on $[0,\infty)$ such that
$\int_{0}^\infty \frac{\la}{\la^2+1}d\nu(\la)<\infty$.

\begin{exam}{\rm
Below we list several important examples of  functions $f:(0,\infty)\to \mathbb{R}$ that are either operator monotone or operator anti-monotone.
\begin{enumerate}
\item $f(t)=(\la+t)^{-1}$ is operator anti-monotone and operator convex for $\lambda \geq 0$. This corresponds to $a=b=0$ and $\mu$ being the point measure at $\la$.

\item Let $0< r< 1$. The power function $t\mapsto t^r$ is operator monotone and operator concave by the following integral representation:
\[
t^r = \frac{\sin(r\pi) }{\pi}\int_0^\infty \left(\frac{ 1}{\la}- \frac{ 1}{\la+t}\right) \la^{r} \, d \la\pl,
\]
where $d\la$ is the Lebesgue measure on $\R$. On the other hand,
$t\mapsto t^{-r}$ is operator anti-monotone  because  it is a composition of $t\mapsto t^{-1}$ and $t\mapsto t^{r}$, with the former being operator anti-monotone and the latter operator monotone. It is thus also operator convex. The integral representation of $t^{-r}$ is
\[
t^{-r} = \frac{\sin(r\pi) }{\pi} \int_0^\infty \la^{-r}\frac{ 1}{\la+t} \, d \la \pl.
\]

\item The logarithm function $f(t)=\log t$ is operator monotone and operator concave. These statements are a consequence of the following integral representation of $\log t$:
\[\log t=-\int_0^\infty \left( \frac{ 1}{\la+t}-\frac{\la}{\la^2+1} \right) d \la \pl.\]
\end{enumerate}
}
\end{exam}

Following \cite{Vershynina19}, we say that an operator monotone or operator anti-monotone function~$f$ is regular if the measure $\nu(\la)$ in its integral representation is absolutely continuous with respect to the Lebesgue measure $d\la$ and for all $0<a<b<\infty$, there is a constant $C_{a,b}$ such that $d \la \le C_{a,b}\, d\nu(\la)$ on $(a,b)$. All examples b), c), and d) given above are regular operator monotone (or anti-monotone) functions.

\subsection{Standard and optimized $f$-divergences}

Let $\M$ be a finite-dimensional von Neumann algebra equipped with a faithful trace $\tau$. For the standard quantum information setup, one can take $\M=B(H)$ for a finite-dimensional Hilbert space $H$ and $\tr$  the matrix trace. For $ p \in [1, \infty)$, the $L_p$-norm for $a\in \M$ is defined as $\norm{a}{p} \coloneqq \tau(|a|^p)^{1/p}$. We use the same notation and definition for $p\in(0,1)$, when $\norm{a}{p}$ is a quasi-norm.
We identify $L_1(\M)\cong \M_*$ and $L_\infty(\M) \cong \M$.
A state $\rho$ is given by a density operator $\rho\in L_1(\M)$ with $\rho\ge 0$ and $\tau(\rho)=1$.
We denote the state space of $\M$ by $D(\M) \coloneqq \{\rho \in L_1(\M)\,|\, \rho\ge 0, \tau(\rho)=1 \}$  and the set of invertible density operators by $D_+(\M)$.
The $L_2$-space $L_2(\M)$ is a Hilbert space with the following trace inner product:
\[\lan x,y\ran\coloneqq \tau(x^*y)\pl.\]
Let $\ket{x}$ denote the vector in the GNS space $L_2(\mathcal{M})$ corresponding to the element $x$. The vector $\ket{1}$ corresponding to the identity operator is an analog of the (unnormalized) maximally entangled state.
The GNS representation $\pi:\M\to B(L_2(\M))$ is given by
\begin{align*} \pi(a)\ket{x}=\ket{ax}.
\end{align*}
We often omit $\pi$ and write $a\ket{x}\coloneqq\pi(a)\ket{x}$.
A state $\rho$ admits a vector representation by $\ket{\rho^{1/2}}$ (also called a purification) that satisfies
\[\rho(x)=\tau(\rho x)=\bra{\rho^{1/2}}x\ket{\rho^{1/2}}\pl.\]
For simplicity of notation, we sometimes write $\ket{\brho}=\ket{\rho^{1/2}}$.
Let $\rho$ and $\si$ be two states. Let  $s(\rho)$ and $s(\si)$ denote the support projections onto the supports of $\rho$ and $\si$, respectively, and let $\rho^{-1}$ denote the inverse of $\rho$ on its support. The relative modular operator is defined as
\[\Delta(\si,\rho)\ket{x} \coloneqq \ket{\si x\rho^{-1}}\pl,\]
which for faithful $\rho$ and $\si$ is always a positive and invertible operator on $L_2(\M)$.

Let $f:(0,\infty)\to \R$ be an operator anti-monotone function. Given two states $\rho$ and $\si$ with $s(\rho)\le s(\si)$,
the standard $f$-divergence is defined as \cite{petz85,petz86} (see also \cite{HMDB11})
\begin{align*}                                      Q_f(\rho\|\si)\coloneqq                                       \bra{\brho}f(\Delta(\si,\rho))\ket{\brho}\pl,
\end{align*}
where $f(\Delta(\si,\rho))$ makes use of the functional calculus, as applied to  $\Delta(\si,\rho)$. The optimized $f$-divergence is defined as \cite{wilde}
\begin{align}
\label{eq:opt-f-div-def}                                      \widetilde{Q}_f(\rho\|\si)
\coloneqq
\sup_{\omega\in D_+(\M) }\bra{\brho}f(\Delta(\si,\omega))\ket{\brho}.
 \end{align}

We review some important examples of relative entropies defined through standard  or optimized $f$-divergences.

\begin{enumerate}
\item Umegaki relative entropy \cite{Umegaki}:
\begin{equation}
\label{eq:umegaki-rel-ent-def}
D(\rho\|\si)\coloneqq-\bra{\brho}\log \Delta(\si,\rho)\ket{\brho}=\tau(\rho\log \rho-\rho\log \si)=Q_{-\log x}(\rho\|\si).
\end{equation}

\item Petz--R\'enyi relative entropy \cite{petz85,petz86}:
\begin{align*}D_\al(\rho\|\si)
& \coloneqq \frac{1}{\al-1}\log\tau(\rho^{\al} \si^{1-\al})
\\ & =
\frac{1}{\al-1}\log\bra{\brho} \Delta(\si,\rho)^{1-\al}\ket{\brho}=
\frac{1}{\al-1}\log Q_{x^{1-\al}}(\rho\|\si)\pl.\end{align*}
This quantity is defined for $\al\in (0,1)\cup (1,\infty)$. A special case of interest for some of the remainder terms in the entropy inequalities in Section~\ref{sec:finite} occurs when $\alpha = 2$ or $\alpha = -1$, for which $Q_{x^{-1}}(\rho\Vert\si) = \tau(\rho^{2}\sigma^{-1})$ or $Q_{x^{2}}(\rho\Vert\si) = \tau(\rho^{-1}\sigma^{2})$, respectively.
\smallskip

\item Sandwiched R\'enyi relative entropy \cite{Muller13,WWY}: for $\al >1$,
\[\widetilde{D}_\al(\rho\|\si) \coloneqq \frac{\al}{\al-1}\log \norm{\si^{\frac{1-\al}{2\al}}\rho \si^{\frac{1-\al}{2\al}}}{\al}
=
\frac{\al}{\al-1}\log \widetilde{Q}_{x^{\frac{1-\al}{\al}}}(\rho\|\si).
\]
where the last equality was identified in \cite{wilde}.
Also, for $0<\al<1$
\[
\widetilde{D}_\al(\rho\|\si)=
\frac{\al}{\al-1}\log \!\left(- \widetilde{Q}_{-x^{\frac{1-\al}{\al}}}(\rho\|\si)\right).
\]
 The sandwiched R\'enyi relative entropy  is defined for $\al\in (0,1)\cup (1,\infty)$.
\smallskip

\item Holevo fidelity \cite{Kholevo1972}:
\[
F_H(\rho,\si) \coloneqq \tau(\rho^{1/2}\si^{1/2})^2 \coloneqq \bra{\rho^{1/2}} \Delta(\si,\rho)^{1/2}\ket{\rho^{1/2}}^2=Q_{x^{1/2}}(\rho\|\si)^2\pl.\]
    Observe that $D_{1/2}(\rho\|\si)=-\log F_H(\rho,\si)$.
    \smallskip

\item Uhlmann fidelity \cite{Uhlmann76}:
\begin{equation}
F(\rho,\si) \coloneqq \tau(|
\sqrt{\rho}\sqrt{\si}|)^2=\norm{\rho^{1/2}\si \rho^{1/2}}{1/2}=\inf_{\omega \in \mathcal{D}(\mathcal{M}_+)} \tau(\rho^{1/2}\si \rho^{1/2}\omega^{-1})=\widetilde{Q}_{x}(\rho\|\si)^2.
\label{eq:uhlmann-fidelity}
\end{equation}
Observe that $\widetilde{D}_{1/2}(\rho\|\si)=-\log F(\rho,\si)$.
\end{enumerate}

\subsection{Petz map and rotated Petz map}

One of the key properties of a quantum divergence is its monotonicity under quantum channels, which is also called the data-processing inequality. Recall that a quantum channel $\Phi:L_1(\M_1)\to L_1(\M_2)$  is a completely positive, trace-preserving (CPTP) map. The data-processing inequality is as follows \cite{petz85}:
 \begin{align}
 \label{1}
 Q_f(\rho\|\si)\ge Q_f(\Phi(\rho)\|\Phi(\si))\pl,
 \end{align}
  holding for all quantum channels $\Phi$ and  states $\rho, \si$.
It was proved (see \cite{petz86,petz88} and \cite[Theorem 6.19]{hiai2021quantum} for the general case) that for a regular operator convex function $f$, the equality $Q_f(\rho\|\si)=  Q_f(\Phi(\rho)\|\Phi(\si))$ holds if and only if $R_{\Phi,\rho}(\Phi(\si))=\si$,
 where $R_{\Phi,\rho}:L_1(\M_2)\to L_1(\M_1)$ is the Petz recovery map, defined as
 \[R_{\Phi,\rho}(x) \coloneqq \rho^{1/2}\Phi^\dag( \Phi(\rho)^{-1/2}x \Phi(\rho)^{-1/2})\rho^{1/2}\pl.\]
The data-processing inequality for the optimized $f$-divergence in the finite-dimensional setting, i.e.,
 \begin{align*}
 \widetilde{Q}_f(\rho\|\si)\ge  \widetilde{Q}_f(\Phi(\rho)\|\Phi(\si))\pl,
 \end{align*}
was proved in \cite{wilde}.

Throughout this paper, we mostly discuss recoverability results for the special case when $\Phi$ is the restriction to a subalgebra, as was done in \cite{petz86}. For example, a partial-trace map $\id\ten \tr:B(H_A\ten H_B)\to B(H_A)$ is a restriction from a tensor-product system $AB$ to the subsystem $A$. By the Stinespring dilation theorem \cite{Stinespring55}, this is the core step in the data-processing inequality. Indeed, recoverability results for general quantum channels (CPTP maps) follow from the subalgebra case.\\

Let $\N\subset\M$ be a subalgebra, and let $E:\M\to \N$ be the unital, trace-preserving conditional expectation, defined through
\[
\tau(xy) = \tau(xE(y))\pl, \quad \forall\pl x\in \N,y\in \M.
\]
For $1\le p\le \infty$, the space $L_p(\N)$ is a subspace of $L_p(\M)$, and the conditional expectation
$E:L_p(\M)\to L_p(\N)$ extends to a projection. The quantum channel corresponding to restriction to a subalgebra is then given by $E:L_1(\M)\to L_1(\N)$.
For an invertible density operator $\rho\in L_1(\M)$, we write $\rho_\N:=E(\rho)$ for the reduced density operator of $\rho$ on~$\N$. The $\rho$-preserving conditional expectation $E_\rho:\M\to \N$ is a unital, completely positive (UCP) map:
\[  E_\rho(x) \coloneqq \rho_\N^{-\frac{1}{2}}E(\rho^{\frac{1}{2}}x\rho^{\frac{1}{2}})\rho_\N^{-\frac{1}{2}}.\]
The Petz recovery map $R_\rho:L_1(\N)\to L_1(\M)$ is the adjoint of $E_\rho$ and is a completely positive trace-preserving (CPTP) map:
\begin{equation}
\label{eq:Petz-rec-marginal}
R_\rho(x)=\rho^{\frac{1}{2}}(\rho_\N^{-\frac{1}{2}}x\rho_\N^{-\frac{1}{2}})\rho^{\frac{1}{2}}.
\end{equation}
Let us also define the rotated Petz map $R_\rho^t$ \cite{Wilde15} and the universal Petz map $R_\rho^u$ \cite{universal} as follows:
\begin{align}
R_\rho^t(x) & \coloneqq \rho^{\frac{1}{2}-it}\left(\rho_\N^{-\frac{1}{2}+it}x\rho_\N^{-\frac{1}{2}-it}\right)\rho^{\frac{1}{2}+it}\quad \pl, \pl \forall  t\in \R,
\label{eq:rotated-petz-map-def}
\\
R_\rho^u(x) & \coloneqq \int_{\R}R_{\rho}^{t/2}(x) \ d\beta(t),
\end{align}
where $d\beta(t)\coloneqq\frac{\pi}{2}(\cosh(\pi t)+1)^{-1}dt$ is a probability measure on $\R$. Both $R_\rho^t$ and $R_\rho^{u}$ are CPTP maps satisfying $R_\rho^t(\rho_\N)=R_\rho^u(\rho_\N)=\rho$ because $\rho^{it}\in \M$ (resp.~$\rho_\N^{-it}\in \N$) is a unitary operator and commutes with $\rho$ (resp.~$\rho_\N$).

\section{Recoverability for $f$-divergences via rotated Petz maps}

\label{sec:finite}

\subsection{Recoverability via a rotated Petz map $R_\rho^t$}

In this section, we discuss recoverability for the $f$-divergence in the finite-dimensional setting.
We start with an improvement of the argument in \cite{CV17}.

Let $\M$ be a finite-dimensional von Neumann algebra, and let $\N\subset \M$ be a subalgebra. Let $\rho,\si\in D_+(\M)$ be two faithful states for $\M$, and let $\rho_\N,\si_\N$ be their restrictions on $\N$, respectively. We use the following shorthand notation for the relative modular operators:
\[\Delta_\M\equiv \Delta(\si,\rho)\in B(L_2(\M))\pl, \qquad \Delta_\N \equiv \Delta(\si_\N,\rho_\N)\pl\in B(L_2(\N)).\]

Let $f:(0,\infty)\to \R$ be an operator anti-monotone function with the following integral representation:
\begin{align}
\label{om-second-time}
f(x)=a+bx+\int_{0}^\infty\left(\frac{1}{\la+x}-\frac{\la}{\la^2+1}\right)d\nu(\la),
\end{align}
where $d\nu(\la)$ is the corresponding measure on $\mathbb{R}$. For $\la\ge 0$, let
\begin{equation}
\label{eq:q-lambda-def}
Q_\la(\rho\|\si) \coloneqq \lan \rho^{1/2}| (\la+\Delta(\si,\rho))^{-1} |\rho^{1/2}\ran
\end{equation}
denote the standard $f$-divergence for $f(x)=(\la+x)^{-1}$.
It follows from the integral representation that
\begin{align}
Q_f(\rho\|\si)-Q_f(\rho_\N\|\si_\N)=\int_0^\infty \left[Q_\la(\rho\|\si)-Q_\la(\rho_\N\|\si_\N)\right]\ d\nu(\la)\pl. \label{dif}
\end{align}
By the data-processing inequality and inspection of \eqref{eq:q-lambda-def}, the following function
\[
F(\la) \coloneqq Q_\la(\rho\|\si)-Q_\la(\rho_\N\|\si_\N)\pl, \quad \la\in [0,\infty)\]
is a continuous non-negative function of all faithful states $\rho$ and $\si$.

We now recall \cite[Lemma~2.1]{CV17}:

\begin{lemma}[Lemma 2.1 of \cite{CV17}]\label{key}
Let $U:K\to H$ be a Hilbert-space isometry, and let $A$ be an invertible positive operator on $H$. Then for all $h\in K$, the following identity holds
\[
\lan h|U^*A^{-1}U| h\ran- \lan h|(U^*AU)^{-1}| h\ran= \lan v|A| v\ran\ge 0 ,
\]
where $|v\ran \coloneqq  A^{-1}U |h\ran - U(U^*AU)^{-1}|h\ran$.
\end{lemma}

Define the isometry $V_\rho:L_2(\N)\to L_2(\M)$ as
\[ V_\rho\ket{x}= \ket{x\rho_\N^{-\frac{1}{2}}\rho^{\frac{1}{2}}}\pl, \quad \forall x\in \N\pl.\]
The adjoint is $V_\rho^*(x)=E(x\rho^{1/2})\rho_\N^{-1/2}$.
Since $V_\rho^*\Delta_{\M}V_\rho=\Delta_{\N}$ (we require the assumption of faithfulness of $\rho$ here) and by Lemma~\ref{key} above, we find that
\begin{align}
F(\la)&=Q_\la(\rho\|\si)-Q_\la(\rho_\N\|\si_\N)
 \\&= \lan \rho^{1/2}| (\Delta_{\M}+\la)^{-1}|\rho^{1/2}\ran-\lan \rho_\N^{1/2}| (\Delta_{\N}+\la)^{-1}|\rho_\N^{1/2}\ran
 \\&=  \lan \rho_\N^{1/2}|V_\rho^* (\Delta_{\M}+\la)^{-1}V_\rho|\rho_\N^{1/2}\ran-\lan \rho_\N^{1/2}| (V_\rho^*(\Delta_{\M}+\la)V_\rho)^{-1}|\rho_\N^{1/2}\ran
 \\ & = \lan w_\la|(\Delta_{\M}+\la) |w_\la\ran \\
 & = \norm{\Delta_{\M}^{1/2}| w_\la\ran}{2}^2+\la \norm{|w_\la\ran }{2}^2,
 \label{eq:f_lam_norm_sum}
\end{align}
where
\begin{equation}
\label{eq:def-w_lam}
\ket{w_\la}\coloneqq (\Delta_{\M}+\la)^{-1}\ket{\rho^{1/2}}-V_\rho(\Delta_\N+\la)^{-1}\ket{\rho_\N^{\frac12}}
\end{equation}
is a vector in $L_2(\M)$.

\begin{lemma}
\label{3.2}
Let $t\in \R$ and set
\begin{align}
\label{eq:def-w_t}
\ket{w_t}:=\Delta_\M^{\frac{1}{2}+it}\ket{\rho^{1/2}}-V_\rho\Delta_\N^{\frac{1}{2}+it}\ket{\rho_\N^{1/2}}
=\ket{\si^{\frac{1}{2}+it}\rho^{-it}}- \ket{\si_\N^{\frac{1}{2}+it}\rho_\N^{-\frac{1}{2}-it}\rho^{\frac{1}{2}}}.
\end{align}
Then the following equality holds
\begin{equation}
\label{eq:w_lam-to-w_t}
\ket{w_t}=-\frac{\cosh(\pi t)}{\pi}\left(\int^{\infty}_{0}\la^{\frac{1}{2}+it}\ket{w_\la}\ d\la\right)
\end{equation}
and the following inequality holds
\begin{equation}
\label{eq:recovery-error-to-w_t}
\norm{\si-R_\rho^{t}(\si)}{1}\le 2\norm{\ket{w_t}}{2}.
\end{equation}
\end{lemma}

\begin{proof}
Recall the operator integral from \cite{Komatsu}, which states that the following integral formula holds for the imaginary power of a positive operator $A$:
\begin{align*}A^{\frac{1}{2}+it}&=\frac{\sin(\pi(\frac{1}{2}+it))}{\pi}\int^{\infty}_{0} \la^{\frac{1}{2}+it}\left(\frac{1}{\la}-\frac{1}{\la +A}\right)d\la\\
&=\frac{\cosh(\pi t)}{\pi}\int^{\infty}_{0} \la^{\frac{1}{2}+it}\left(\frac{1}{\la}-\frac{1}{\la +A}\right)d\la.
\end{align*}
Then
\begin{align*}
&\ket{w_t}\notag \\
& = \Delta_\M^{\frac{1}{2}+it}\ket{\rho^{1/2}}-V_\rho\Delta_\N^{\frac{1}{2}+it}\ket{\rho_\N^{1/2}}
\\ & =\frac{\cosh(\pi t)}{\pi}\left(\int^{\infty}_{0}\la^{1/2+it}\left(\frac{1}{\la}-\frac{1}{\Delta_\M+\la}\right)\ket{\rho^{1/2}}\, d\la
-V_\rho\int^{\infty}_{0}\la^{1/2+it}\left(\frac{1}{\la}-\frac{1}{\Delta_\N+\la}\right)\ket{\rho_\N^{1/2}}\, d\la\right)
\\ & = -\frac{\cosh(\pi t)}{\pi}\left(\int^{\infty}_{0}\la^{1/2+it}\left(\frac{1}{\Delta_\M+\la}\ket{\rho^{1/2}}
-V_\rho\frac{1}{\Delta_\N+\la}\ket{\rho_\N^{1/2}}\right)d\la\right)
\\ & = -\frac{\cosh(\pi t)}{\pi}\left(\int^{\infty}_{0}\la^{1/2+it}\ket{w_\la}\, d\la\right).
\end{align*}
This establishes \eqref{eq:w_lam-to-w_t}.

We now prove \eqref{eq:recovery-error-to-w_t}. Note that $(\si^{\frac{1}{2}+it})^*\si^{\frac{1}{2}+it}=\si$, and
\begin{align*} &(\si_\N^{\frac{1}{2}+it}\rho_\N^{-\frac{1}{2}-it}\rho^{\frac{1}{2}+it})^*(\si_\N^{\frac{1}{2}+it}\rho_\N^{-1/2-it}\rho^{\frac{1}{2}+it})= \rho^{\frac{1}{2}-it}\rho_\N^{-\frac{1}{2}+it}\si_\N\rho_\N^{-\frac{1}{2}-it}\rho^{\frac{1}{2}+it}=R_\rho^t(\si)\pl,
\end{align*}
where $R_\rho^t(\si)$ is defined in \eqref{eq:rotated-petz-map-def}.
Recall the following inequality from \cite[Lemma 2.2]{CV17}:
\begin{align}
\norm{x^*x-y^*y}{1}\le 2 \norm{x-y}{2}\pl,  \label{2}
\end{align}
which holds for $x$ and $y$ satisfying $\norm{x}{2}=\norm{y}{2}=1$.
Then \eqref{eq:recovery-error-to-w_t} follows because
\begin{align*}
\norm{\si-R_\rho^t(\si)}{1}&\le 2\norm{\si^{1/2+it}-\si_\N^{1/2+it}\rho_\N^{-1/2-it}\rho^{\frac{1}{2}+it}}{2}
\\  &=2\norm{\si^{1/2+it}\rho^{-it}-\si_\N^{1/2+it}\rho_\N^{-1/2-it}\rho^{\frac{1}{2}}}{2}=2\norm{\ket{w_t}}{2}.\qedhere
\end{align*}
\end{proof}

\bigskip

 For a regular operator anti-monotone function $f$, we have the following estimate of~$\norm{\ket{w_t}}{2}$:

\begin{lemma}
\label{fd}
Let $f:(0,\infty)\to \R$ be a regular operator anti-monotone function, and let $d\nu$ be the measure in its integral representation. Suppose $Q_{x^2}(\rho\|\si)=\bra{\rho^{1/2}}\Delta_\M^2\ket{\rho^{1/2}}<\infty$.
Suppose for some $S$ and $T$, satisfying $0\le S<T\le \infty$, that there exists $c(S,T)>0$ such that on the interval $(S,T)$,
\[d\la\le c(S,T)\, d\nu(\la).\]
Then, for $\ket{w_t}$ as defined in \eqref{eq:def-w_t}, the following inequality holds
\begin{multline*}
\norm{\ket{w_t}}{2}\le \\
\frac{\cosh(\pi t)}{\pi}\left(4S^{1/2}+[c(S,T)(T-S)]^{1/2}(Q_f(\rho\|\si)-Q_f(\rho_\N\|\si_\N))^{1/2}+4T^{-1/2}Q_{x^2}(\rho\|\si)^{1/2}
\right).
\end{multline*}
\end{lemma}

\begin{proof}
To estimate the norm of $\ket{w_t}$, we break $\ket{w_t}$ into three separate terms after applying~\eqref{eq:w_lam-to-w_t}:
\begin{align*}
-\ket{w_t}&= \frac{\cosh(\pi t)}{\pi}\left(\int^{\infty}_{0}\la^{1/2+it}\ket{w_\la}\, d\la\right)
\\ &= \frac{\cosh(\pi t)}{\pi}\left(\int^{S}_{0}\la^{1/2+it}\ket{w_\la}\, d\la+\int^{S}_{T}\la^{1/2+it}\ket{w_\la}\, d\la+\int^{\infty}_{T}\la^{1/2+it}\ket{w_\la}\, d\la\right)
\\ &
=: \frac{\cosh(\pi t)}{\pi}\left(\operatorname{I}+\operatorname{II}+\operatorname{III}\right),
\end{align*}
where $\ket{w_\la}$ is defined in \eqref{eq:def-w_lam}.
For the first term $\operatorname{I}$, we define the function
\[h_{S}^t(x) \coloneqq \int_0^S \la^{\frac{1}{2}+it}\left(\frac{1}{\la}-\frac{1}{\la+x}\right)d\la .
\]
Thus
\begin{align*}
\int_0^S \la^{\frac{1}{2}+it}\ket{w_\la}\, d\la= V_\rho h_S^t(\Delta_\N)\ket{\rho_\N^{1/2}}-h_S^t(\Delta_\M)\ket{\rho^{1/2}}.
\end{align*}
For $x\ge 0$,
\begin{align*}|h_{S}^t(x)|
& \le \int_0^S \left|\la^{\frac{1}{2}+it}\left(\frac{1}{\la}-\frac{1}{\la+x}\right)\right| d\la
\le \int_0^S \la^{\frac{1}{2}}\left(\frac{1}{\la}-\frac{1}{\la+x}\right) d\la
\\ & = 2x^{1/2}\arctan \!\left(\frac{\sqrt{S}}{\sqrt{x}}\right) =h_S^0(x).
\end{align*}
Note that $h_S^0$ is the function $h_S^t$ with $t=0$.
So we conclude that
\begin{align*}
\norm{h_S^t(\Delta_\M)\ket{\rho^{1/2}}}{2}^2
&=\bra{\rho^{1/2}}h^t_S(\Delta_\M)^*h^t_S(\Delta_\M)\ket{\rho^{1/2}}
\\ & \le \bra{\rho^{1/2}}4\Delta_\M\arctan^2 \!\left(\frac{\sqrt{S}}{\sqrt{\Delta_\M}}\right)\ket{\rho^{1/2}}
\\ & = \int_{0}^\infty 4s\arctan^2 \!\left(\frac{\sqrt{S}}{\sqrt{s}}\right)d\mu(s)
\\ & \le  \int_{0}^\infty 4s\left(\frac{S}{s}\right)d\mu(s)\le 4S.
\end{align*}
Here $d\mu(s)$ is the probability measure $d\mu(s)=d\bra{\rho^{1/2}}E_{[0,s]}(\Delta_\M)\ket{\rho^{1/2}}$ from the spectral decomposition, and the second inequality uses the fact that $\arctan(x)\le x$ for $x\ge 0$.
 Similarly,
\[ \norm{V_\rho h_S^t(\Delta_\N)\ket{\rho_\N^{1/2}}}{2}^2 \le\bra{\rho_\N^{1/2}}4\Delta_\N\arctan^2 \!\left(\frac{\sqrt{S}}{\sqrt{\Delta_\N}}\right)\ket{\rho_\N^{1/2}}\le 4S\pl. \]
Thus we have
\begin{equation}
\norm{\operatorname{I}}{2}\le \norm{h_S^t(\Delta_\M)\ket{\rho^{1/2}}}{2}+\norm{V_\rho h_S^t(\Delta_\N)\ket{\rho_\N^{1/2}}}{2}\le 2S^{1/2}+2S^{1/2}=4S^{1/2}. \label{eq:estimate-1-reg-op-mon-lemma}
\end{equation}

For the second term, consider that
\begin{align}
\norm{\operatorname{II}}{2}
&=
\norm{\int^{T}_{S}\la^{1/2+it}\ket{w_\la}d\la}{2} \\
& \le \int^{T}_{S}\la^{1/2}\norm{\ket{w_\la}}{2}d\la
\\
&\le \left(\int^{T}_{S} 1 \, d\la\right)^{1/2} \left(\int^{T}_{S}\la\norm{\ket{w_\la}}{2}^2d\la \right)^{1/2}
\\
&\le (T-S)^{1/2} \left(\int^{T}_{S}F(\la)\, d\la \right)^{1/2}
\\ &\le (T-S)^{1/2} \left(c(S,T)\int^{T}_{S}F(\la)\ d\nu(\la) \right)^{1/2}
\\ &\le \left[c(S,T)(T-S)(Q_f(\rho\|\si)-Q_f(\rho_\N\|\si_\N))\right]^{1/2}.
\label{eq:estimate-2-reg-op-mon-lemma}
\end{align}
The first inequality follows from the triangle inequality, the second from Cauchy--Schwarz, the third  from \eqref{eq:f_lam_norm_sum}, the fourth from the assumption of a regular operator anti-monotone function~$f$, and the fifth from \eqref{dif}.

For the third term, consider that
\begin{align*}
\operatorname{III}& =\int_{T}^{\infty}\la^{1/2+it}\ket{w_\la}\, d\la
\\
& = - \int_{T}^{\infty}\la^{1/2+it}\left(\frac{1}{\la}-\frac{1}{\la+\Delta_\M}\right) \, \ket{\rho^{1/2}}d\la
+ V_\rho\int_{T}^{\infty}\la^{1/2+it}\left(\frac{1}{\la}-\frac{1}{\la+\Delta_\N}\right) \, \ket{\rho_\N^{1/2}} d\la .
\end{align*}
Let us consider the integral
\begin{align*}
\int_{T}^{\infty}\la^{1/2+it}\left(\frac{1}{\la}-\frac{1}{\la+x}\right)d\la
=x^{1/2+it}\int_{\frac{T}{x}}^{\infty}\la^{1/2+it}\left(\frac{1}{\la}-\frac{1}{\la+1}\right)d\la.
\end{align*}
Note that the function $\la\mapsto \la^{1/2+it}\left(\frac{1}{\la}-\frac{1}{\la+1}\right)$ is bounded and integrable on $(0,\infty)$. We define the following continuous function:
\[
g_T^t(x):=\int_{T}^{\infty}\la^{1/2+it}\left(\frac{1}{\la}-\frac{1}{\la+x}\right)d\la.
\]
Then
\[
\int_{T}^{\infty}\la^{1/2+it}\left(\frac{1}{\la}-\frac{1}{\la+\Delta_\M}\right)d\la=g_T^t\!\left(\Delta_\M\right)\pl.
\]
For $x\ge 0$,
\begin{equation}
\label{eq:g_func-upp-bnd}
|g_T^t(x)|^2\le x \left(\int_{\frac{T}{x}}^{\infty}\la^{1/2}\left(\frac{1}{\la}-\frac{1}{\la+1}\right)d\la\right)^2=4x\arctan^2\!\left(\frac{\sqrt{x}}{\sqrt{T}}\right)\pl.
\end{equation}
Therefore
\begin{align*}
\norm{\int_{T}^{\infty}\la^{1/2+it}\left(\frac{1}{\la}-\frac{1}{\la+\Delta_\M}\right)d\la \ket{\rho^{1/2}}}{2}^2
& =
\bra{\rho^{1/2}}  |g_T^t(\Delta_\M)|^2   \ket{\rho^{1/2}}
\\
& \le \bra{\rho^{1/2}} 4\Delta_\M\arctan^2\!\left(\frac{\sqrt{\Delta_\M}}{\sqrt{T}}\right)   \ket{\rho^{1/2}}
\\ & \le \frac{4}{T}\bra{\rho^{1/2}} \Delta_\M^2  \ket{\rho^{1/2}}
\\ & = \frac{4}{T}Q_{x^2}(\rho\|\si).
\end{align*}
where we used \eqref{eq:g_func-upp-bnd} and the bound $\arctan(x)\le x$, holding for $x \geq 0$. Similarly,
\begin{align*}
\norm{V_\rho\int_{T}^{\infty}\la^{1/2+it}\left(\frac{1}{\la}-\frac{1}{\la+\Delta_\N}\right) \ket{\rho_\N^{1/2}}d\la }{2}^2
& \le
\frac{4}{T}\bra{\rho_\N^{1/2}} \Delta_\N^2  \ket{\rho_\N^{1/2}} \\
& = \frac{4}{T}Q_{x^2}(\rho_\N\|\si_\N) \\
& \le \frac{4}{T}Q_{x^2}(\rho\|\si).
\end{align*}
The final inequality follows from the data-processing inequality for the Petz--R\'enyi relative quasi-entropy $Q_{x^2}$.
So we conclude from the analysis above and the triangle inequality that
\begin{equation}
\label{eq:estimate-3-reg-op-mon-lemma}
\norm{\operatorname{III}}{2} \leq
4 T^{-1/2}Q_{x^2}(\rho\|\si)^{1/2}.
\end{equation}

 Putting the estimates from \eqref{eq:estimate-1-reg-op-mon-lemma}, \eqref{eq:estimate-2-reg-op-mon-lemma}, and \eqref{eq:estimate-3-reg-op-mon-lemma} together, we find that
\begin{align*}
\norm{\ket{w_t}}{2}
&\le \frac{\cosh(\pi t)}{\pi}\Big(\norm{\operatorname{I}}{2}+\norm{\operatorname{II}}{2}+\norm{\operatorname{III}}{2}\Big)
\\ & \le \frac{\cosh(\pi t)}{\pi}\Big(4S^{1/2}+c(S,T)^{1/2}(T-S)^{1/2}(Q_f(\rho\|\si)-Q_f(\rho_\N\|\si_\N))^{1/2}\\ &\qquad\qquad\qquad\qquad +4T^{-1/2}Q_{x^2}(\rho\|\si)^{1/2}\Big).
\end{align*}
This completes the proof.
\end{proof}

A direct consequence of Lemmas~\ref{3.2} and \ref{fd} is the following general bound on the recoverability error in terms of a standard $f$-divergence:
\begin{cor}
\label{cor:bnd-standard-f-div-1}
Considering the same hypotheses of Lemma~\ref{fd}, the following inequality holds
\begin{multline*}
\norm{\si-R_\rho^t(\si_\N)}{1}\le \frac{2 \cosh(\pi t)}{\pi} \times \\
\left(4S^{1/2}+[c(S,T)(T-S)]^{1/2}(Q_f(\rho\|\si)-Q_f(\rho_\N\|\si_\N))^{1/2}+4T^{-1/2}Q_{x^2}(\rho\|\si)^{1/2}
\right).
\end{multline*}
\end{cor}

\bigskip

For a particular choice of $f$, the estimate from Corollary~\ref{cor:bnd-standard-f-div-1} simplifies, depending on the measure~$\nu$. In the next section, we  consider some important examples.

\subsubsection{Recoverability for quantum relative entropy}

We begin with the quantum relative entropy $D(\rho\|\si)$, as defined in \eqref{eq:umegaki-rel-ent-def}.

\begin{theorem}
\label{re}
Let $\M$ be a finite-dimensional von Neumann algebra, and let $\N\subset \M$ be a subalgebra. Then for all faithful states $\rho$ and $\si$, the following inequalities hold
\begin{align}
D(\rho\|\si)-D(\rho_\N\|\si_\N) & \ge \left(\frac{\pi}{8} \right)^4 Q_{x^2}(\rho\|\si)^{-1}\norm{\si-R_\rho(\si_\N)}{1}^4 \label{eq1}, \\
D(\rho\|\si)-D(\rho_\N\|\si_\N) & \ge \left(\frac{\pi}{8 \cosh(\pi t)} \right)^4 Q_{x^2}(\rho\|\si)^{-1}\norm{\si-R_\rho^t(\si_\N)}{1}^4 \label{eq2} , \\
D(\rho\|\si)-D(\rho_\N\|\si_\N) & \ge \frac{1}{256} \, Q_{x^2}(\rho\|\si)^{-1}\norm{\si-R_\rho^u(\si_\N)}{1}^4.
\label{eq3}
\end{align}
Here $Q_{x^2}(\rho\|\si)=\lan\rho^{1/2}| \Delta(\si,\rho)^2|\rho^{1/2}\ran = \tau(\rho^{-1}\si^2)$.
\end{theorem}

\begin{proof}
Consider from \eqref{eq:umegaki-rel-ent-def} that $f(x)=-\log x$, for which we have the following integral representation:
\[
-\log x=\int_0^\infty \left(\frac{1}{\la+x}-\frac{\la}{1+\la^2}\right) d\la \pl,
\]
where $d\la$ is the Lebesgue measure. Thus $c(S,T)=1$ for $0\le S\le T\le \infty$. Choose $S=0$ and
\[
T= 4 \, Q_{x^2}(\rho\|\si)^{1/2}(D(\rho\|\si)-D(\rho_\N\|\si_\N))^{-1/2}\pl.
\]
Applying Corollary~\ref{cor:bnd-standard-f-div-1}, we obtain the following:
\begin{align}
\norm{\si-R_\rho^t(\si_\N)}{1} & \le
\nonumber
 \frac{2\cosh(\pi t)}{\pi}\left[T^{1/2}(D(\rho\|\si)-D(\rho_\N\|\si_\N))^{1/2}+4T^{-1/2}Q_{x^2}(\rho\|\si)^{1/2}\right]
\nonumber \\
& = \frac{8\cosh(\pi t)}{\pi}(D(\rho\|\si)-D(\rho_\N\|\si_\N))^{1/4}Q_{x^2}(\rho\|\si)^{1/4} \label{est}.
\end{align}

 Eq.~\eqref{eq2} follows from rewriting \eqref{est}, and \eqref{eq3} follows from integrating \eqref{est} with respect to the measure $d\beta(t)=\frac{\pi}{2}\left(\cosh(\pi t)+1\right)^{-1}dt$, from convexity of the trace norm, the fact that $\int_{\mathbb{R}} \frac{\cosh(\pi t/2)}{  (\cosh(\pi t)+1)} dt =1$, and by the integral expression  $R_\rho^u=\int_{\mathbb{R}}R_\rho^{\frac{t}{2}}\,d\beta(t)$. Eq.~\eqref{eq1} is a special case of \eqref{eq2} at $t=0$.
\end{proof}

\begin{rem}{\rm As mentioned in Section~\ref{sec:summary-results}, an advantage of the bounds in Theorem~\ref{re} over previous bounds from \cite{CV17,CV18} is that the remainder term features the quantity $Q_{x^2}(\rho\|\si)$ rather than the operator norm of the relative modular operator. As such, these bounds are non-trivial for the important class of bosonic Gaussian states \cite{S17}, whereas the previous bounds from \cite{CV17,CV18} are trivial for this class of states. Moreover, explicit formulas are available for evaluating the Petz-- and sandwiched-R\'enyi relative entropies of bosonic Gaussian states (see \cite{SLW18} and references therein). This remark applies not only to Theorem~\ref{re}, but also to Theorem~\ref{thm:quasi-renyi-recoverability}, Corollary~\ref{cor:petz-renyi-recoverability}, Theorems~\ref{rsi} and \ref{thm:other-petz-map-renyi-rec}, and Corollary~\ref{cor:petz-renyi-recoverability-2}.}
\end{rem}

\subsubsection{Recoverability for Petz--R\'enyi relative (quasi-)entropy}

For $\al\in (0,1)\cup (1,2)$, the Petz--R\'enyi relative entropy  is given by
\begin{equation}
\label{eq:petz-renyi-def-recall}
D_\al(\rho\|\si)=\frac{1}{\al-1}\log \tau(\rho^{\al}\si^{1-\al})=\frac{1}{\al-1}\log Q_{x^{1-\al}}(\rho\|\si)\pl.
\end{equation}
In some of the statements that follow, we also adopt a different parameterization by setting \begin{equation}
s \coloneqq 1-\alpha \in (-1,0)\cup(0,1),
\label{eq:s-for-quasi-renyi}
\end{equation}
 so that
\begin{equation}
\label{eq:petz-renyi-def-recall-1}
D_\al(\rho\|\si)=D_{1-s}(\rho\|\si) =
-\frac{1}{s}\log \tau(\rho^{1-s}\si^{s})= - \frac{1}{s}\log Q_{x^{s}}(\rho\|\si)\pl,
\end{equation}
and we also adopt the abbreviation
\begin{equation}
Q_{s}(\rho\|\si) \equiv Q_{x^{s}}(\rho\|\si) =
\tau(\rho^{1-s}\si^{s}).
\label{eq:quasi-renyi-with-s}
\end{equation}
We begin by focusing on the Petz--R\'enyi relative quasi-entropy in Theorem~\ref{thm:quasi-renyi-recoverability}, and then we extend the result to the Petz--R\'enyi relative entropy in Corollary~\ref{cor:petz-renyi-recoverability}.

\begin{theorem}
\label{thm:quasi-renyi-recoverability}
Let $\M$ be a finite-dimensional von Neumann algebra, and let $\N\subset \M$ be a subalgebra.
Then for all faithful states $\rho$ and $\si$, the following inequalities hold for the Petz--R\'enyi relative quasi-entropy $Q_{s}(\rho\|\si) = \tau(\rho^{1-s}\si^{s})$ for $s \in (-1,0)\cup (0,1)$:
\begin{align}
|Q_{s}(\rho\|\si)-Q_{s}(\rho_\N\|\si_\N)| & \ge K(s ,Q_{x^2}(\rho\|\si))\left(\frac{\pi}{4+2|s|}\norm{\si-R_\rho(\si_\N)}{1}\right)^{4+2|s|},
\label{petz-renyi-quasi-recov-no-t}
\\
|Q_{s}(\rho\|\si)-Q_{s}(\rho_\N\|\si_\N)| & \ge K(s ,Q_{x^2}(\rho\|\si)) \left(\frac{\pi}{(4+2|s|)\cosh{\pi t}}\norm{\si-R_\rho^t(\si_\N)}{1}\right)^{4+2|s|},  \label{petz-renyi-quasi-recov-with-t}
\\
|Q_{s}(\rho\|\si)-Q_{s}(\rho_\N\|\si_\N)| & \ge K(s ,Q_{x^2}(\rho\|\si))\left(\frac{1}{(2+|s|)}\norm{\si-R_\rho^u(\si_\N)}{1}\right)^{4+2|s|}.
\label{petz-renyi-quasi-recov-universal-1}
\end{align}
where $Q_{x^2}(\rho\|\si) = \tau(\rho^{-1}\si^{2})$ and
\begin{align}
\label{constant}
K(s ,Q_{x^2}(\rho\|\si))
\coloneqq
\begin{dcases}
\frac{\sin(\pi|s|)}{\pi}
\left(\frac{(|s|+1)^2}{16 Q_{x^2}(\rho\|\si)}\right)^{|s|+1} & \text{ for } {-1} < s < 0
\\
\frac{\sin(\pi s)}{\pi Q_{x^2}(\rho\|\si)}
\frac{s^{2s}}{16^{s+1}}& \text{ for } 0<s<1
\end{dcases}.
\end{align}
\end{theorem}

\begin{proof}
For $0<s<1$, the function $f(x)=x^s$ is operator monotone and operator concave. An integral representation for it is
\[x^s=\frac{\sin (\pi s)}{\pi}\int_0^\infty \la^s\left(\frac{1}{\la}-\frac{1}{\la+x}\right)d\la\pl.\]
So for $0<S<T<0$, the constant $c(S,T)\le \frac{\pi}{\sin (\pi s)}S^{-s}$.
Then by applying Corollary~\ref{cor:bnd-standard-f-div-1}, we find that
\begin{multline*}
\norm{\si-R_\rho^t(\si_\N)}{1} \le
\frac{2\cosh(\pi t)}{\pi}\times \\\left(4S^{1/2}+\left(\frac{\pi}{\sin (\pi s)}\right)^{1/2}S^{-s/2}(T-S)^{1/2}|Q_{s}(\rho\|\si)-Q_{s}(\rho_\N\|\si_\N)|^{1/2} +4T^{-1/2}Q_{x^2}(\rho\|\si)^{1/2}\right)
\end{multline*}
Define the following constants:
\[a=4\pl,\quad b=\left(\frac{\pi}{\sin (\pi s)}\right)^{\frac{1}{2}}|Q_{s}(\rho\|\si) - Q_{s}(\rho_\N\|\si_\N)|^{1/2}\pl, \quad c=4Q_{x^2}(\rho\|\si)^{1/2}.\]
We then have
\begin{align*}
\norm{\si-R_\rho^t(\si_\N)}{1}
& \le \frac{2\cosh(\pi t)}{\pi}(aS^{1/2}+bS^{-s/2}(T-S)^{\frac{1}{2}}+cT^{-1/2})\\
& \le \frac{2\cosh(\pi t)}{\pi}(aS^{1/2}+bS^{-s/2}(T-S)^{\frac{1}{2}}+c(T-S)^{-\frac{1}{2}}).
\end{align*}
Minimizing the following function over $0<S<T<\infty$:
\[F(S,T)=aS^{1/2}+bS^{-s/2}(T-S)^{\frac{1}{2}}+c(T-S)^{-\frac{1}{2}},
\]
we obtain the following inequality at the choices $S = \left(\frac{cbs^2}{a^2}\right)^{\frac{2}{s+2}}$ and $T-S=\left(\frac{cbs^2}{a^2}\right)^{\frac{s}{s+2}}\left(\frac{c}{b}\right)$: \begin{align}
\nonumber
&\norm{\si-R_\rho^t(\si_\N)}{1}\\
&\nonumber \le
\frac{2\cosh(\pi t)}{\pi}(s+2)(a^s s^{-s}bc)^{\frac{1}{s+2}}
\\ &= \frac{2\cosh(\pi t)}{\pi}(s+2) \Big(16^{s+1}s^{-2s}\frac{\pi}{\sin(\pi s)}|Q_{s}(\rho\|\si) - Q_{s}(\rho_\N\|\si_\N)|Q_{x^2}(\rho\|\si)   \Big)^{\frac{1}{2s+4}}.
\label{eq:Petz-renyi-remainder-1}
\end{align}

For $-1< s<0$, the function $f(x)=x^s$ is operator anti-monotone and operator convex.
An integral representation of $x^s$ for $s \in (-1,0)$ is
\[x^s= \frac{\sin (\pi |s|)}{\pi}\int_0^\infty \frac{\la^s}{\la+x}\ d\la\pl.\]
Then we can choose $S=0$ and the constant $c(0,T)\le \frac{\pi}{\sin (\pi |s|)}T^{|s|}$.
By Corollary~\ref{cor:bnd-standard-f-div-1}, we find that
\begin{multline*}
\norm{\si-R_\rho^t(\si_\N)}{1}
 \le
\frac{2\cosh(\pi t)}{\pi}\times \\\left(\left(\frac{\pi}{\sin (\pi |s|)}T^{|s|}\right)^{1/2} T^{1/2}(Q_{s}(\rho\|\si)-Q_{s}(\rho_\N\|\si_\N))^{1/2}
 +4T^{-1/2}Q_{x^2}(\rho\|\si)^{1/2}\right).
\end{multline*}
Define the following constants:
\[ b=\left(\frac{\pi}{\sin (\pi |s|)}\right)^{\frac{1}{2}}(Q_{s}(\rho\|\si)-Q_{s}(\rho_\N\|\si_\N))^{1/2}\pl, \qquad c=4Q_{x^2}(\rho\|\si)^{1/2}.\]
We want to minimize the following function over $0<T<\infty$:
\[
G(T) = bT^{\frac{|s|+1}{2}}+cT^{-1/2}\pl.
\]
Choosing $T=\left(\frac{c}{b(|s|+1)}\right)^{\frac{2}{|s|+2}}$, we find that
\begin{align}
\nonumber
&\norm{\si-R_\rho^t(\si_\N)}{1}\\
& \nonumber \le
\frac{2\cosh(\pi t)}{\pi}
(|s|+2)
(|s|+1)^{-\frac{|s|+1}{|s|+2}}
c^{\frac{|s|+1}{|s|+2}}b^{\frac{1}{|s|+2}}\\
& =
\frac{2\cosh(\pi t)}{\pi}
(|s|+2)
(|s|+1)^{-\frac{|s|+1}{|s|+2}}\times\nonumber \\
& \qquad \left(\left(\frac{\pi}{\sin (\pi |s|)}\right)(Q_{s}(\rho\|\si)-Q_{s}(\rho_\N\|\si_\N))\right)^{\frac{1}{2(|s|+2)}}
(4\sqrt{Q_{x^2}(\rho\|\si)})^{\frac{|s|+1}{|s|+2}}
\label{eq:Petz-renyi-remainder-2}.
\end{align}

Putting together the conclusions in \eqref{eq:Petz-renyi-remainder-1} and \eqref{eq:Petz-renyi-remainder-2}, we arrive at \eqref{petz-renyi-quasi-recov-with-t}. Eq.~\eqref{petz-renyi-quasi-recov-no-t}  is a special case of \eqref{petz-renyi-quasi-recov-with-t}.
The proof of \eqref{petz-renyi-quasi-recov-universal-1} is similar to the proof in Theorem~\ref{re}.
\end{proof}

\bigskip
Note that the estimate above fails for $s=-1$  because the measure in the integral representation of $x^{-1}$ is a point mass at $\lambda = 0$.

\begin{exam}{\rm  For $s=1/2$, we have the Holevo fidelity
\[ F_H(\rho,\si)=Q_{x^{1/2}}(\rho\|\si)^2=\tau(\rho^{1/2}\si^{1/2})^2\pl. \]
Then
\begin{equation}
\label{eq:holevo-fid-rem}
\sqrt{F_H(\rho_\N,\si_\N)}-\sqrt{F_H}(\rho,\si) \geq \frac{1}{128\pi Q_{x^2}(\rho\|\si)}\left(\frac{\pi}{5}\norm{\si-R_\rho(\si_\N)}{1}\right)^{5} .
\end{equation}
The inequality in \eqref{eq:holevo-fid-rem} can be compared with the main result of \cite{wilde18}. The prefactor $[Q_{x^2}(\rho\|\si)]^{-1}$ is an improvement, but the fifth power on the trace distance is not.
}
\end{exam}

The estimate in Theorem~\ref{thm:quasi-renyi-recoverability} leads to a strengthened data-processing inequality for the Petz--R\'enyi relative entropy, as defined in \eqref{eq:petz-renyi-def-recall},
by following the same argument used to arrive at \cite[Theorem 6.1]{CV18}, along with an additional argument:

\begin{cor}
\label{cor:petz-renyi-recoverability}
Let $\M$ be a finite-dimensional von Neumann algebra, and let $\N\subset \M$ be a subalgebra. Let $\rho$ and $\si$ be two faithful states. For $\al\in (0,1)$ and $t\in \mathbb{R}$,
\begin{multline*}
D_\al(\rho\|\si)-D_\al(\rho_\N\|\si_\N)\\ \ge \frac{1}{1-\al}\log\!\left(1+K(1-\al ,Q_{x^2}(\rho\|\si))\left(\frac{\pi}{2(3-\al)\cosh{\pi t}}\norm{\si-R_\rho^t(\si_\N)}{1}\right)^{2(3-\al)}\right),
\end{multline*}
and for $\al \in (1,2)$ and $t\in \mathbb{R}$,
\begin{multline*}
D_\al(\rho\|\si)-D_\al(\rho_\N\|\si_\N)\\ \ge
\frac{1}{\al-1} \log \!\left(1+\frac{ K(1-\al , Q_{x^2}(\rho\|\si)) }{Q_{x^{-1}}(\rho\|\si)^{\alpha-1}}\left( \frac{\pi}{2(\al + 1)\cosh{\pi t}}
\norm{\si-R_\rho^t(\si_\N)}{1}\right)^{2(\alpha+1)}\right),
\end{multline*}
where $Q_{x^2}(\rho\|\si))  = \tau(\rho^{-1}\si^{2})$, $Q_{x^{-1}}(\rho\|\si))  = \tau(\rho^{2}\si^{-1})$, and the constant $K(1-\al ,Q_{x^2}(\rho\|\si))$ is given by \eqref{constant}.
\end{cor}

\begin{proof}
For $0<\al<1$, we find that
\begin{align*}
& D_\al(\rho\|\si)-D_\al(\rho_\N\|\si_\N)\\
& = \frac{1}{1-\al}
\log \frac{ Q_{x^{1-\al}}(\rho_\N\|\si_\N)}{Q_{x^{1-\al}}(\rho\|\si)}
\\
& = \frac{1}{1-\al}  \log \left(1+
\frac{ Q_{x^{1-\al}} (\rho_\N\|\si_\N)-Q_{x^{1-\al}}(\rho\|\si)}{Q_{x^{1-\al}}(\rho\|\si)}\right)
\\
& \ge  \frac{1}{1-\al}  \log \left(1+\frac{ K(1-\al  ,Q_{x^2}(\rho\|\si)) }{Q_{x^{1-\al}}(\rho\|\si)} \left( \frac{\pi}{2(3-\al)\cosh{\pi t}}\norm{\si-R_\rho^t(\si_\N)}{1}\right)^{2(3-\al)}\right),
\\ & \ge
\frac{1}{1-\al}  \log\left(1+
K(1-\al  ,Q_{x^2}(\rho\|\si)) \left(\frac{\pi}{2(3-\al)\cosh{\pi t}} \norm{\si-R_\rho^t(\si_\N)}{1}\right)^{2(3-\al)}\right).
\end{align*}
The first inequality follows from \eqref{petz-renyi-quasi-recov-with-t}, and the second
 follows because  $Q_{x^{1-\al}}(\rho\|\si)\le 1$ for $\al \in (0,1)$.

For $\al \in (1,2)$, consider that
\begin{align*}
&D_\al(\rho\|\si)- D_\al(\rho_\N\|\si_\N)\\
& = \frac{1}{\al-1}\log \frac{ Q_{x^{1-\al}}(\rho\|\si)}{Q_{x^{1-\al}}(\rho_\N\|\si_\N)}
\\& = \frac{1}{\al-1}\log \!\left(1+\frac{ Q_{x^{1-\al}}(\rho\|\si) - Q_{x^{1-\al}}(\rho_\N\|\si_\N)}{Q_{x^{1-\al}}(\rho_\N\|\si_\N)}\right)
\\
& \ge  \frac{1}{\al-1}\log \!\left(1+\frac{ K(1-\al  ,Q_{x^2}(\rho\|\si)) }{Q_{x^{1-\al}}(\rho_\N\|\si_\N)}\left(\frac{\pi}{2(\al + 1)\cosh{\pi t}}
\norm{\si-R_\rho^t(\si_\N)}{1}\right)^{2(\al + 1)}\right)
\\
& \ge  \frac{1}{\al-1}\log \!\left(1+\frac{ K(1-\al  ,Q_{x^2}(\rho\|\si)) }{Q_{x^{-1}}(\rho\|\si)^{\alpha-1}}\left( \frac{\pi}{2(\al + 1)\cosh{\pi t}}
\norm{\si-R_\rho^t(\si_\N)}{1}\right)^{2(\al + 1)}\right).
\end{align*}
The first inequality follows from \eqref{petz-renyi-quasi-recov-with-t}, and the second
 follows because $Q_{x^{1-\al}}(\rho_\N\|\si_\N) \leq Q_{x^{1-\al}}(\rho\|\si) \leq Q_{x^{-1}}(\rho\|\si)^{\al -1}$, the latter following from data processing and the fact that the Petz--R\'enyi relative entropies are monotone non-decreasing with respect to $\al$.
\end{proof}

\subsection{Recoverability via another rotated Petz map $R_\si^t$}

We now modify the argument from the previous section to obtain a recoverability statement involving the other rotated Petz map $R_\si^t$. This time we use the following integral representation, holding for $t\in \R$:
\begin{align}
\label{neg}
x^{-\frac{1}{2}-it}
& = -\frac{\sin(\pi(-\frac{1}{2}-it))}{\pi}\int^{\infty}_{0} \la^{-\frac{1}{2}-it}({\la +x})^{-1}d\la
\\ &=\frac{\cosh(\pi t)}{\pi}\int^{\infty}_{0} \la^{-\frac{1}{2}-it}({\la +x})^{-1}d\la\ .
\label{eq:integral-rep-x-min-1-2-it}
\end{align}

\begin{lemma}
\label{3.10}
Let $t\in \R$ and
\begin{align*}\ket{v_t}:=\ket{\rho^{1/2}}-\Delta_{\M}^{\frac{1}{2}+it}V_\rho\Delta_{\N}^{-\frac{1}{2}-it}\ket{\rho_\N^{1/2}}
=\ket{\rho^{1/2}}- \ket{\si^{\frac{1}{2}+it}\si_\N^{-\frac{1}{2}-it}\rho_\N^{it+\frac{1}{2}}\rho^{-it}}.
\end{align*}
Then the following equality holds
\begin{equation}
\ket{v_t}=\frac{\cosh(\pi t)}{\pi}\Delta_{\M}^{\frac{1}{2}+it}\int_0^\infty \la^{-\frac{1}{2}-it}\ket{w_\la}\, d\la,
\label{eq:v_t-rewrite}
\end{equation}
where $\ket{w_\la}$ is defined in \eqref{eq:def-w_lam},
and the following inequality holds
\begin{equation}
\label{eq:v_t-inequality-to-rec-err}
\norm{\rho-R_\si^{-t}(\rho_\N)}{1}\le 2\norm{\ket{v_t}}{2}.
\end{equation}
\end{lemma}

\begin{proof}
Using the integral representation in \eqref{eq:integral-rep-x-min-1-2-it} for $\Delta_{\M}$ and $\Delta_{\N}$, we find that
\begin{align}
& \Delta_{\M}^{-\frac{1}{2}-it} \ket{v_t} \notag \\
&= \Delta_{\M}^{-\frac{1}{2}-it}\ket{\rho^{1/2}}-V_\rho\Delta_{\N}^{-\frac{1}{2}-it}\ket{\rho_\N^{1/2}}
\label{eq:v_t_mod_1}
\\
& =  \frac{\cosh(\pi t)}{\pi}\left(\int_0^\infty \la^{-\frac{1}{2}-it}(\Delta_{\M}+\la)^{-1}\, d\la \ket{\rho^{1/2}}-V_\rho\left(\int_0^\infty \la^{-\frac{1}{2}-it}(\Delta_{\N}+\la)^{-1}\, d\la\right) \ket{\rho_\N^{1/2}}\right)
\\
& =  \frac{\cosh(\pi t)}{\pi}\int_0^\infty \la^{-\frac{1}{2}-it}\ket{w_\la}\, d\la ,
\label{eq:v_t_mod_3}
\end{align}
where $\ket{w_\la}$ is defined in \eqref{eq:def-w_lam}.
Applying $\Delta_{\M}^{1/2+it}$ leads to \eqref{eq:v_t-rewrite}:
\begin{align*}
\Delta_{\M}^{1/2+it}\left(\Delta_{\M}^{-\frac{1}{2}-it}\ket{\rho^{1/2}}-V_\rho\Delta_{\N}^{-\frac{1}{2}-it}\ket{\rho_\N^{1/2}}\right)
&=\ket{\rho^{1/2}}-\Delta_{\M}^{\frac{1}{2}+it}V_\rho\Delta_{\N}^{-\frac{1}{2}-it}\ket{\rho_\N^{1/2}}
\\& =\ket{\rho^{1/2}}- \ket{\si^{\frac{1}{2}+it}\si_\N^{-\frac{1}{2}-it}\rho_\N^{\frac{1}{2}+it}\rho^{-it}}\pl.
\end{align*}

The inequality in \eqref{eq:v_t-inequality-to-rec-err} follows from \eqref{2} and the following identity:
\begin{align*}
\si^{\frac{1}{2}+it}\si_\N^{-\frac{1}{2}-it}\rho_\N^{\frac{1}{2}+it}\rho^{-it}(\si^{\frac{1}{2}+it}\si_\N^{-\frac{1}{2}-it}\rho_\N^{\frac{1}{2}+it}\rho^{-it})^*
& =
\si^{\frac{1}{2}+it}\si_\N^{-\frac{1}{2}-it}\rho_\N
\si_\N^{-\frac{1}{2}+it}\si^{\frac{1}{2}-it}\\
& = R_\si^{-t}(\rho_{\N}),
\end{align*}
where $R_\si^{-t}(\rho_{\N})$ is defined through \eqref{eq:rotated-petz-map-def}.
\end{proof}

\bigskip
 We have the following estimate of $\norm{\ket{v_t}}{2}$:

\begin{lemma}
\label{fd2}
Let $f:(0,\infty)\to \R$ be a regular  operator anti-monotone function, and let $d\nu$ be the measure in its integral representation. Let $Q_{x^{-1}}(\rho\|\si)=\bra{\rho^{1/2}}\Delta_\M^{-1}\ket{\rho^{1/2}}=\tau(\rho^{2}\si^{-1})$.
Suppose for $0< S<T< \infty$ that there exists $c(S,T)>0$ such that on the interval $(S,T)$
\[d\la\le c(S,T)\, d\nu(\la).\]
Then
\begin{multline*}
\norm{\ket{v_t}}{2}\le  \frac{\cosh(\pi t)}{\pi} \times \\
\left(4(Q_{x^{-1}}(\rho\|\si)S)^{1/2}+[c(S,T)\ln(T/S)]^{1/2}(Q_f(\rho\|\si)-Q_f(\rho_\N\|\si_\N))^{1/2}+4T^{-1/2}
\right).
\end{multline*}
\end{lemma}

\begin{proof}
By applying \eqref{eq:v_t-rewrite}, consider that
\begin{align*}
\norm{\ket{v_t}}{2} & =  \frac{\cosh(\pi t)}{\pi}\norm{\Delta_{\M}^{1/2}\int_0^\infty \la^{-\frac{1}{2}-it}\ket{w_\la}\, d\la}{2}
\\ & \le  \frac{\cosh(\pi t)}{\pi}\Bigg(\norm{\Delta_{\M}^{1/2}\int_0^S \la^{-\frac{1}{2}-it}\ket{w_\la}\, d\la}{2}+\norm{\Delta_{\M}^{1/2}\int_S^T \la^{-\frac{1}{2}-it}\ket{w_\la}\, d\la}{2}\\
& \qquad\qquad\qquad\qquad +\norm{\Delta_{\M}^{1/2}\int_T^\infty \la^{-\frac{1}{2}-it}\ket{w_\la}\, d\la}{2}\Bigg)
\\ & =:  \frac{\cosh(\pi t)}{\pi}(\norm{\operatorname{I}}{2}+\norm{\operatorname{II}}{2}+\norm{\operatorname{III}}{2}).
\end{align*}
For the terms above, we show the following estimates:
\begin{align*}
\norm{\operatorname{I}}{2}& = \norm{\Delta_{\M}^{\frac{1}{2}}\left(\int_0^S \la^{-\frac{1}{2}-it}\ket{w_\la} \,d\la\right)}{2}\le 4 S^{1/2}[Q_{x^{-1}}(\rho\|\si)]^{1/2},
\\  \norm{\operatorname{II}}{2}& = \norm{\Delta_{\M}^{\frac{1}{2}}\left(\int_S^T \la^{-\frac{1}{2}-it}\ket{w_\la} \,d\la\right)}{2}\le \left(c(S,T)\ln \!\left(\frac{T}{S}\right) (Q_f(\rho\|\si)-Q_f(\rho_\N\|\si_\N))\right)^{1/2},
\\  \norm{\operatorname{III}}{2}& = \norm{\Delta_{\M}^{\frac{1}{2}}\left(\int_T^\infty \la^{-\frac{1}{2}-it}\ket{w_\la} \, d\la\right)}{2}\le 4T^{-1/2}.
\end{align*}

For the first term, we define the following function:
\[h_S^t(x) \coloneqq \int_0^S \la^{-\frac{1}{2}-it}\frac{1}{\la+x}\ d\la \pl.\]
Thus
\begin{align*}
\int_0^S \la^{-\frac{1}{2} - it}\ket{w_\la}\, d\la= h_S^t(\Delta_\M)\ket{\rho^{1/2}}- V_\rho h_S^t(\Delta_\N)\ket{\rho_\N^{1/2}},
\end{align*}
leading to
\begin{align*}
 \Delta_{\M}^{\frac{1}{2}}\left(\int_0^S \la^{-\frac{1}{2}-it}\ket{w_\la} \, d\la\right)
=\Delta_{\M}^{\frac{1}{2}} h_S^t(\Delta_\M)\ket{\rho^{1/2}}- \Delta_{\M}^{\frac{1}{2}} V_\rho h_S^t(\Delta_\N)\ket{\rho_\N^{1/2}}.
\end{align*}
Note that for $x\ge 0$,
\begin{align*}|h_{S}^t(x)|& \le \int_0^S \la^{-\frac{1}{2}}\frac{1}{\la+x}\  d\la= 2x^{-1/2}\arctan \!\left(\frac{\sqrt{S}}{\sqrt{x}}\right) =: h_S(x),
\end{align*}
so that
\begin{align*}
|h_{S}^t(x)|^2& \le 4x^{-1}\arctan^2 \!\left(\frac{\sqrt{S}}{\sqrt{x}}\right).
\end{align*}
Here $h_S$ is the function $h_S^t$ with $t=0$. Using the spectral theorem for the probability measure $d\mu(s)=d \bra{\rho^{1/2}} E_{[0,s]}(\Delta_\M)\ket{\rho^{1/2}}$, we find that
\begin{align*}
& \norm{\Delta_{\M}^{\frac{1}{2}} h_S^t(\Delta_\M)\ket{\rho^{1/2}}}{2}^2\notag \\
& =\bra{\rho^{1/2}}\Delta_\M h^t_S(\Delta_\M )^*h^t_S(\Delta_\M )\ket{\rho^{1/2}}
\\
& \le \bra{\rho^{1/2}}\Delta_\M h^0_S(\Delta_\M )^*h^0_S(\Delta_\M )\ket{\rho^{1/2}}
=\int_0^\infty s\, h^2_S(s) \, d\mu(s)
\\
& =\int_0^\infty 4 \arctan^{2}\left(\frac{\sqrt{S}}{\sqrt{s}}\right) d\mu(s)
\le  4S\int_0^\infty  \frac{1}{s}\, d\mu(s)
=  4S\bra{\rho^{1/2}} \Delta_\M^{-1}\ket{\rho^{1/2}}
\\
& =4S\,  Q_{x^{-1}}(\rho\|\si).
\end{align*}
where we use again the inequality $\arctan(x)\le x$, holding for $x\geq 0$. Similarly,
\begin{align*}
\norm{\Delta_{\M}^{\frac{1}{2}} V_\rho h_S^t(\Delta_\N)\ket{\rho_\N^{1/2}}}{2}^2 & =
\bra{\rho_\N^{1/2}}h_S^t(\Delta_\N)^*V_\rho^* \Delta_{\M} V_\rho h_S^t(\Delta_\N)\ket{\rho^{1/2}}
\\ & \le \bra{\rho_\N^{1/2}}h_S^t(\Delta_\N)^*\Delta_{\N} h_S^t(\Delta_\N)\ket{\rho_\N^{1/2}}
\\ & = \bra{\rho_\N^{1/2}}\Delta_{\N} h_S(\Delta_\N)^2\ket{\rho_\N^{1/2}}
\\
& \le  4S \bra{\rho_\N^{1/2}} \Delta_\N^{-1}\ket{\rho_\N^{1/2}}
\\
& \le 4S \bra{\rho^{1/2}} \Delta_\M^{-1}\ket{\rho^{1/2}}
= 4S\,  Q_{x^{-1}}(\rho\|\si),
\end{align*}
where we used $V_\rho^* \Delta_{\M} V_\rho= \Delta_{\N}$ and the fact that $x\mapsto x^{-1}$ is an operator anti-monotone and operator convex function. Therefore
\begin{align*}
\norm{\operatorname{I}}{2} & \le  \norm{\Delta_{\M}^{\frac{1}{2}}\int_0^S \la^{-\frac{1}{2}-it}(\Delta_{\M}+\la)\ket{\rho^{1/2}}}{2}+\norm{\Delta_{\M}^{\frac{1}{2}}\int_0^S \la^{-\frac{1}{2}-it}V_\rho(\Delta_{\N}+\la)\ket{\rho_\N^{1/2}}}{2}
\\ &\le 2 S^{1/2}[Q_{x^{-1}}(\rho\|\si)]^{1/2}
+2 S^{1/2}[Q_{x^{-1}}(\rho\|\si)]^{1/2} \notag \\
& \le 4 S^{1/2}[Q_{x^{-1}}(\rho\|\si)]^{1/2}.
\end{align*}

For the second term, consider that
\begin{align*}
\norm{\operatorname{II}}{2}&=\norm{\Delta_\M^{1/2} \int^{T}_{S}\la^{-1/2-it}\ket{w_\la}d\la}{2} \\
& \le \int^{T}_{S}\la^{-1/2}\norm{\Delta_\M^{1/2}\ket{w_\la}}{2}d\la
\\ &\le \left(\int^{T}_{S} \la^{-1} d\la\right)^{1/2} \left(\int^{T}_{S}\norm{\Delta_\M^{1/2}\ket{w_\la}}{2}^2d\la \right)^{1/2}
\\ &\le (\ln (T/S))^{1/2} \left(\int^{T}_{S}F(\la)\, d\la \right)^{1/2}
\\ &\le (\ln (T/S))^{1/2} \left(c(S,T)\int^{T}_{S}F(\la)\, d\nu(\la) \right)^{1/2}
\\ &\leq  \Big(c(S,T)\ln (T/S)(Q_f(\rho\|\si)-Q_f(\rho_\N\|\si_\N))\Big)^{1/2}.
\end{align*}

For the third term, consider that
\begin{align*}
&\Delta_\M^{1/2}\int_{T}^{\infty}\la^{-1/2-it}\ket{w_\la}\,d\la
\\ & = \Delta_\M^{1/2}\int_{T}^{\infty}\la^{-1/2-it}\frac{1}{\la+\Delta_\M}\, d\la \ket{\rho^{1/2}}
- \Delta_\M^{1/2}V_\rho\int_{T}^{\infty}\la^{-1/2-it}\frac{1}{\la+\Delta_\N}\,d\la\ket{\rho_\N^{1/2}}.
\end{align*}
Let us consider the integral
\begin{align*}
\int_{T}^{\infty}\la^{-1/2-it}\frac{1}{\la+x}\ d\la
=x^{-1/2-it}\int_{\frac{T}{x}}^{\infty}\la^{-1/2-it}\frac{1}{\la+1}\ d\la.
\end{align*}
Note that the function $\la\mapsto \la^{-1/2-it}\frac{1}{\la+1}$ is bounded and integrable on $(0,\infty)$. We define the continuous function
\[g_T^t(x):=\int_{T}^{\infty}\la^{-1/2-it}\frac{1}{\la+x} \ d\la. \]
Then
\[\int_{T}^{\infty}\la^{-1/2-it}\frac{1}{\la+\Delta_\M}\ d\la=g_T^t({\Delta_\M})\pl.\]
For $x\ge 0$,
\[|g_T^t(x)|^2\le x^{-1} \left(\int_{\frac{T}{x}}^{\infty}\la^{-1/2}\frac{1}{\la+1}\, d\la\right)^2 = 4x^{-1}
\arctan^2\!\left(\frac{\sqrt{x}}{\sqrt{T}}\right)\pl.\]
Therefore
\begin{align*}
\norm{\Delta_\M^{1/2}\int_{T}^{\infty}\la^{-1/2-it}\left(\frac{1}{\la+\Delta_\M}\right) d\la \ket{\rho^{1/2}}}{2}^2
& = \bra{\rho^{1/2}}  g_T^t(\Delta_\M)^*\Delta_\M g_T^t(\Delta_\M)   \ket{\rho^{1/2}}
\\ & \le \bra{\rho^{1/2}} 4\arctan^2\!\left(\frac{\sqrt{\Delta_\M}}{\sqrt{T}}\right)   \ket{\rho^{1/2}}
\\ & \le \frac{4}{T}\bra{\rho^{1/2}} \Delta_\M  \ket{\rho^{1/2}}
\\ & = \frac{4}{T},
\end{align*}
where we used again the inequality $\arctan(x)\le x$, holding for $x\geq 0$. Similarly,
\begin{align*}
\norm{\Delta_\M^{1/2}V_\rho\int_{T}^{\infty}\la^{-1/2-it}\frac{1}{\la+\Delta_\N}d\la \ket{\rho_\N^{1/2}}}{2}^2 & \le \frac{4}{T}\bra{\rho_\N^{1/2}} \Delta_\N  \ket{\rho_\N^{1/2}} = \frac{4}{T}.
\end{align*}
Putting the estimates together, we conclude the proof.
\end{proof}
\bigskip

A direct consequence of Lemmas~\ref{3.10} and \ref{fd2}, as well as the symmetry of the function $\cosh(\pi t)$ about $t=0$, is the following general bound on the recoverability error in terms of a standard $f$-divergence:
\begin{cor}
\label{cor:bnd-standard-f-div-2}
Considering the same hypotheses of Lemma~\ref{fd2}, the following inequality holds
\begin{multline*}
\norm{\rho-R_\si^t(\rho_\N)}{1}\le
\frac{2 \cosh(\pi t)}{\pi} \\
\left(4S^{1/2}Q_{x^{-1}}(\rho\|\si)^{1/2}+[c(S,T)\ln(T/S)]^{1/2}(Q_f(\rho\|\si)-Q_f(\rho_\N\|\si_\N))^{1/2}+4T^{-1/2}
\right).
\end{multline*}
\end{cor}

\subsubsection{Recoverability for quantum relative entropy}

We have the following estimate for the quantum relative entropy $D(\rho\|\si)$, as defined in~\eqref{eq:umegaki-rel-ent-def}:

\begin{theorem}
\label{rsi}
Let $\M$ be a finite von Neumann algebra, and let $\N\subset \M$ be a subalgebra.
Let $\rho$ and $\si$ be faithful density operators of $\M$, and let $\rho_\N$ and $\si_\N$  be the respective reduced density operators on $\N$.
Denote $Q_{x^{-1}}(\rho\|\si)=\tau(\rho^2\si^{-1})$. Then for all $\varepsilon\in(0,1/2)$ and~$t\in \mathbb{R}$,
\begin{align}
D(\rho\|\si)-D(\rho_\N\|\si_\N)
& \ge \left(K(Q_{x^{-1}}(\rho\|\si),\varepsilon)\frac{\pi}{2} \norm{\rho-R_\si(\rho_\N)}{1}\right)^{\frac{1}{1/2-\varepsilon}} \label{eq1b},\\
D(\rho\|\si)-D(\rho_\N\|\si_\N) & \ge \left(K(Q_{x^{-1}}(\rho\|\si),\varepsilon)\frac{\pi}{2\cosh(\pi t)} \norm{\rho-R_\si^t(\rho_\N)}{1}\right)^{\frac{1}{1/2-\varepsilon}},
\\
D(\rho\|\si)-D(\rho_\N\|\si_\N) & \ge \Big(K(Q_{x^{-1}}(\rho\|\si),\varepsilon) \norm{\rho-R_\si^u(\rho_\N)}{1}\Big)^{\frac{1}{1/2-\varepsilon}}, \label{eq2b}
\end{align}
where the constant
\begin{equation}
K(Q_{x^{-1}}(\rho\|\si),\varepsilon) \coloneqq \Big(4\sqrt{Q_{x^{-1}}(\rho\|\si)} + 4 + (\varepsilon e)^{-1/2}\Big)^{-1}.
\label{eq:K-constant-rel-ent-orig}
\end{equation}
\end{theorem}

\begin{proof}
Consider from \eqref{eq:umegaki-rel-ent-def} that $f(x)=-\log x$, for which we have the following integral representation:
\[-\log x=\int_0^\infty \left(\frac{1}{\la+x}-\frac{\la}{\la^2+1}\right)d\la \pl,\]
where $d\la$ is the Lebesgue measure. Thus $c(S,T)=1$ for all $0\le S\le T\le \infty$.
Then, by applying Corollary~\ref{cor:bnd-standard-f-div-2}, we find that
\begin{multline}
\norm{\rho-R_\si^t(\rho_\N)}{1}
\le \frac{2\cosh(\pi t)}{\pi} \times
\\ \left(4Q_{x^{-1}}(\rho\|\si)^{1/2}S^{1/2}+(\ln(T/S) (D(\rho\|\si)-D(\rho_\N\|\si_\N)))^{1/2}+4T^{-1/2}\right).
\label{eq:rel-ent-1st-step-remainder-term}
\end{multline}
Define the following constants:
\[a \coloneqq 4Q_{x^{-1}}(\rho\|\si)^{1/2}\pl, \qquad b \coloneqq (D(\rho\|\si)-D(\rho_\N\|\si_\N))^{1/2}\pl, \qquad c \coloneqq 4.\]
We want to minimize the following function over $0<S<T<\infty$,
\[ F(S,T)=aS^{1/2}+b\sqrt{\ln(T/S)}+cT^{-1/2}\pl.\]
Set $\delta \coloneqq  \min\{ D(\rho\|\si)-D(\rho_\N\|\si_\N),1\}$.
Then a rough choice is $S=T^{-1}=\delta$, and we find that
\begin{align*}
\norm{\rho-R_\si^t(\rho_\N)}{1}
& \le \frac{2\cosh(\pi t)}{\pi}\left(4Q_{x^{-1}}(\rho\|\si)^{1/2}\delta^{1/2}+(2\delta|\ln \delta|)^{1/2}+4\delta^{1/2}\right)
\\ & = \frac{2\cosh(\pi t)}{\pi}\left(4\sqrt{Q_{x^{-1}}(\rho\|\si)}+4+\sqrt{2|\ln \delta|}\right)\delta^{1/2}
\\ & \le  \frac{2\cosh(\pi t)}{\pi} \left(4\sqrt{Q_{x^{-1}}(\rho\|\si)}+4 + (\varepsilon e)^{-1/2}\right)\delta^{1/2-\varepsilon}.
\end{align*}
In the case that $\delta = 1$,  the first inequality is trivial, following because $\norm{\rho-R_\si^t(\rho_\N)}{1} \leq 2$, $4Q_{x^{-1}}(\rho\|\si)^{1/2}\delta^{1/2}+(2\delta|\ln \delta|)^{1/2} \geq 0$, and $\frac{2\cosh(\pi t)}{\pi}4\delta^{1/2} \geq 2$ for all $t\in \mathbb{R}$ in this case. Otherwise, the first inequality is a consequence of \eqref{eq:rel-ent-1st-step-remainder-term}.
The last inequality is a consequence of  the inequalities $\delta^{\varepsilon} < 1$ and  $\delta^\varepsilon\sqrt{2|\ln \delta|}\le (\varepsilon e)^{-1/2}$, holding for $0<\delta<1$ and $\varepsilon > 0$.

The rest of the proof is similar to the proof of Theorem~\ref{re}.
\end{proof}

\begin{rem}{\rm It is a consequence of the result in \cite{universal} that the following inequality holds for the universal recovery map $R_\si^u$:
\begin{equation}
D(\rho\|\si)-D(\rho_\N\|\si_\N)\ge 4\norm{\rho - R_\si^u(\rho_\N)}{1}^2\pl.
\label{eq:universal-consq-JRSWW}
\end{equation}
For $R_\si^u$, the inequality in \eqref{eq:universal-consq-JRSWW} is stronger than our estimate in \eqref{eq2b}. Nevertheless,  Theorem~\ref{rsi} above provides an error estimate with the rotated Petz map $R_\si^t$ for each $t$.
}
\end{rem}

\subsubsection{Recoverability for Petz--R\'enyi relative (quasi-)entropy}

We now turn to the Petz--R\'enyi relative quasi-entropy, as defined in \eqref{eq:petz-renyi-def-recall}--\eqref{eq:quasi-renyi-with-s}.

\begin{theorem}
\label{thm:other-petz-map-renyi-rec}
Let $s\in (-1,0)\cup(0,1)$. Denote $Q_{x^{-1}}(\rho\|\si)=\tau(\rho^2\si^{-1})$. Then the following inequalities hold for all $t\in \mathbb{R}$ and $\varepsilon \in (0,(1-|s|)/2)$:
\begin{align}
|Q_{s}(\rho\|\si)-Q_{s}(\rho_\N\|\si_\N)|
& \ge \left(K(s ,Q_{x^{-1}}(\rho\|\si),\varepsilon)\frac{\pi}{2}
\norm{\rho-R_\si(\rho_\N)}{1}\right)^{\frac{1}{\frac{1-|s|}{2}-\varepsilon }} \label{eq1c}, \\
|Q_{s}(\rho\|\si)-Q_{s}(\rho_\N\|\si_\N)|
& \ge  \left(K(s ,Q_{x^{-1}}(\rho\|\si),\varepsilon) \frac{\pi}{2\cosh{\pi t}}\norm{\rho-R_\si^t(\rho_\N)}{1}\right)^{\frac{1}{\frac{1-|s|}{2}-\varepsilon }}  \label{eq2c} , \\
|Q_{s}(\rho\|\si)-Q_{s}(\rho_\N\|\si_\N)|
& \ge \big(K(s ,Q_{x^{-1}}(\rho\|\si),\varepsilon )\norm{\rho-R_\si^u(\rho_\N)}{1}\big)^{\frac{1}{\frac{1-|s|}{2}-\varepsilon }}
\label{eq3c},
\end{align}
where the constant
\begin{equation}
\label{eq:K-const-other-recov}
K(s ,Q_{x^{-1}}(\rho\|\si),\varepsilon ) \coloneqq \Big(
4Q_{x^{-1}}(\rho\|\si)^{1/2}
+\left(\frac{\pi}{e\varepsilon\sin (\pi |s|)}\right)^{1/2}+4\Big)^{-1}.
\end{equation}
\end{theorem}

\begin{proof}
For $0<s<1$, the function $f(x)=x^s$ is operator monotone and operator concave. The integral representation of $x^s$ is
\[x^s=\frac{\sin (\pi s)}{\pi}\int_0^\infty \la^s\left(\frac{1}{\la}-\frac{1}{\la+x}\right)d\la\pl.
\]
Corollary~\ref{cor:bnd-standard-f-div-2} holds for $c(S,T)\le \frac{\pi}{\sin (\pi s)}S^{- s}$, and we find that
\begin{multline}
\label{eq:renyi-rel-ent-recovery-1st-step}
\norm{\rho-R_\si^t(\rho_\N)}{1}  \le
\frac{2\cosh(\pi t)}{\pi}\Bigg(4S^{1/2}Q_{x^{-1}}(\rho\|\si)^{1/2}\\
+\sqrt{S^{- s}\frac{\pi}{\sin (\pi  s)}\ln(T/S) }\left|Q_{s}(\rho\|\si)-Q_{s}(\rho_\N\|\si_\N)\right|^{1/2}+4T^{-1/2}\Bigg).
\end{multline}
Define the following constants:
\[a \coloneqq 4Q_{x^{-1}}(\rho\|\si)^{1/2}\pl, \qquad b \coloneqq \left(\frac{\pi}{\sin (\pi  s)}\left|Q_{s}(\rho\|\si)-Q_{s}(\rho_\N\|\si_\N)\right|\right)^{1/2}\pl,\qquad c\coloneqq 4,\]
and the function
\[F(S,T)=aS^{1/2}+b\sqrt{S^{- s}\ln(T/S)}+cT^{-1/2}\pl.\]
Setting $\delta \coloneqq \min\{ |Q_{s}(\rho\|\si)-Q_{s}(\rho_\N\|\si_\N)|,  1\}$ and $S=T^{-1}=\delta$,
we find that
\begin{align}
&\norm{\rho-R_\si^t(\rho_\N)}{1}  \notag \\
& \label{eq:delta-arg-1st-time-1} \le
\frac{2\cosh(\pi t)}{\pi}\left(4\delta^{1/2}Q_{x^{-1}}(\rho\|\si)^{1/2}
+\delta^{(1- s)/2}\sqrt{\frac{2\pi}{\sin (\pi  s)}|\ln \delta|}+4\delta^{1/2}\right)
\\ & =
\frac{2\cosh(\pi t)}{\pi}\left(4Q_{x^{-1}}(\rho\|\si)^{1/2} \delta^{ s/2}
+\sqrt{\frac{2\pi}{\sin (\pi  s)}|\ln \delta|}+4\delta^{ s/2} \right)\delta^{(1- s)/2}
\\ &  \le
\frac{2\cosh(\pi t)}{\pi}\left(4Q_{x^{-1}}(\rho\|\si)^{1/2}
+\left(\frac{\pi}{e\varepsilon\sin (\pi  s)}\right)^{1/2}+4\right)\delta^{(1- s)/2-\varepsilon}.
\label{eq:delta-arg-1st-time-last}
\end{align}
In the case that $\delta = 1$,  the first inequality is trivial, following from the facts that $\norm{\rho-R_\si^t(\rho_\N)}{1} \leq 2$,
\[
4\delta^{1/2}Q_{x^{-1}}(\rho\|\si)^{1/2}
+\delta^{(1- s)/2} \sqrt{\frac{2\pi}{\sin (\pi  s)}|\ln \delta|} \geq 0,
\]
and $\frac{2\cosh(\pi t)}{\pi}4\delta^{1/2} \geq 2$ for all $t\in \mathbb{R}$ in this case. Otherwise, the first inequality is a consequence of \eqref{eq:renyi-rel-ent-recovery-1st-step}.
The last inequality is a consequence of  the inequalities $\delta^{ s/2+\varepsilon}\leq 1 $ and
$\delta^\varepsilon\sqrt{2|\ln \delta|}\le (\varepsilon e)^{-1/2}$, holding for $0<\delta<1$.

For $-1<  s<0$, the function $f(x)=x^ s$ is operator anti-monotone and operator convex. The integral representation of $x^ s$ in this case is
\[x^ s=\frac{\sin (\pi | s|)}{\pi}\int_0^\infty \frac{\la^ s}{\la+x}\, d\la\pl.\]
Then the constant $c(S,T)\le \frac{\pi}{\sin (\pi | s|)}T^{|s|}$. The following inequality holds as a consequence of  Corollary~\ref{cor:bnd-standard-f-div-2}:
\begin{multline*}
\norm{\rho-R_\si^t(\rho_\N)}{1}  \le \frac{2\cosh(\pi t)}{\pi}
\Bigg(4S^{1/2}Q_{x^{-1}}(\rho\|\si)^{1/2} \\
+\sqrt{\frac{\pi}{\sin (\pi | s|)}T^{|s|}\ln (T/S)}(Q_{s}(\rho\|\si)-Q_{s}(\rho_\N\|\si_\N))^ {1/2}+4T^{-1/2}\Bigg).
\end{multline*}
The rest of the analysis is the same as that given for the case $0<s<1$, by taking $S=T^{-1}=\delta$.
\end{proof}


\bigskip
Following the same method of proof given for
Corollary~\ref{cor:petz-renyi-recoverability}, we arrive at the following corollary:

\begin{cor}
\label{cor:petz-renyi-recoverability-2}
Let $\M$ be a finite-dimensional von Neumann algebra, and let $\N\subset \M$ be a subalgebra. Let $\rho$ and $\si$ be two faithful states. For $\al\in (0,1)$, $\varepsilon\in(0,\alpha/2)$, and  $t\in \mathbb{R}$,
\begin{multline*}
D_\al(\rho\|\si)-D_\al(\rho_\N\|\si_\N)\\
\ge \frac{1}{1-\al}\log\!\left(1+
\left(
K(1-\al ,Q_{x^{-1}}(\rho\|\si),\varepsilon)
\frac{\pi}{2\cosh{\pi t}}\norm{\rho-R_\si^t(\rho_\N)}{1}\right)^{\frac{1}{\al/2-\varepsilon}}\right),
\end{multline*}
and for $\al \in(1, 2)$, $\varepsilon\in(0,(2-\alpha)/2)$, and $t\in \mathbb{R}$,
\begin{multline*}
D_\al(\rho\|\si)-D_\al(\rho_\N\|\si_\N)\\ \ge
\frac{1}{\al-1} \log \!\left(1+\frac{
1 }{Q_{x^{-1}}(\rho\|\si)^{\alpha-1}}\left(  \frac{K(1-\al , Q_{x^{-1}}(\rho\|\si),\varepsilon)\pi}{2(\al + 1)\cosh{\pi t}}
\norm{\rho-R_\si^t(\rho_\N)}{1}\right)^{\frac{1}{(2-\al)/2-\varepsilon}}\right),
\end{multline*}
where the constant $K(1-\al ,Q_{x^{-1}}(\rho\|\si),\varepsilon)$ is given by \eqref{eq:K-const-other-recov}.
\end{cor}

\subsection{Recoverability for optimized $f$-divergence}

We now discuss recoverability for the optimized $f$-divergence.
Let $\M$ be a finite-dimensional von Neumann algebra with trace $\tau$, and let $\N\subset\M$ be a subalgebra. Let $\rho,\si\in \M$ be two faithful states, and let $E(\rho)=\rho_\N$ and $E(\si)=\si_\N$ be the respective reduced density operators on $\N$.
Let $f$ be an operator anti-monotone function. Recall from \eqref{eq:opt-f-div-def} that the optimized $f$-divergences are defined as follows:
\begin{align*}
\widetilde{Q}_f(\rho\|\si) & =\sup_{\omega\in D_+(\M)}\bra{\rho^{1/2}} f(\Delta_\M(\si,\omega))\ket{\rho^{1/2}},\\
\widetilde{Q}_f(\rho_\N\|\si_\N) & = \sup_{\omega_\N\in D_+(\N)}\bra{\rho_\N^{1/2}} f(\Delta_\N(\si_\N,\omega_\N))\ket{\rho_\N^{1/2}}\pl,
\end{align*}
where the supremum is with respect to all invertible density operators $\omega\in D_+(\M) $ (resp. $\omega_\N\in D_+(\N)$). Let $V_\rho:L_2(\N)\to L_2(\M)$ denote the following isometry:
\[ V_\rho(a\ket{\rho_\N^{1/2}})=a\ket{\rho^{1/2}}\pl, \quad \forall a\in \N, \]
with a similar definition for $V_\si$.

\begin{lemma}\label{omega}
Let $\rho,\si,\omega\in D_+(\M)$.
The following equality holds
\[ V_\rho^*\Delta_\M(\si, R_\rho(\omega_\N))V_\rho= \Delta_\N(\si_\N,\omega_\N),
\]
where $R_\rho$ is the Petz recovery map from \eqref{eq:Petz-rec-marginal}. As a consequence, the following inequality holds for all operator anti-monotone functions $f:(0,\infty)\to \R$:
\[
\bra{\rho_\N^{1/2}} f(\Delta_\N(\si_\N,\omega_\N))\ket{\rho_\N^{1/2}} \le \bra{\rho^{1/2}} f(\Delta_\M(\si,R_\rho(\omega_\N)))\ket{\rho^{1/2}}.
\]
\end{lemma}

\begin{proof}
Recall the recovery map $R_\rho(\omega_\N)=\rho^{1/2}\rho_\N^{-1/2}\omega_\N \rho_\N^{-1/2}\rho^{1/2} $.
By definition of $\Delta$ and $V_\rho$, we have for any $a\in \N$,
\begin{align*}
V_\rho^*\Delta(\si,R_\rho(\omega_\N))V_\rho \ket{a\rho_\N^{1/2}}
=& V_\rho^*\Delta(\si,R_\rho(\omega_\N))\ket{a\rho^{1/2}}
= V_\rho^*\ket{\si a\rho^{1/2}R_\rho(\omega_\N)^{-1}}
\\=&\ket{E(\si a\rho_\N^{1/2}\omega_\N^{-1} \rho_\N^{1/2})\rho_\N^{-1/2}}
\\=& \ket{\si_\N a\rho_\N^{1/2}\omega_\N^{-1}}
\\=& \Delta(\si_\N,\omega_\N)\ket{ a\rho_\N^{1/2}}
\end{align*}
where in the second last equality we used the fact that $\rho_\N^{1/2}\omega_\N^{-1} \rho_\N^{1/2}\in \N$.
This verifies the claimed equality.

Now using operator convexity and operator anti-monotonicity of $f$, we find that
\begin{align*}
\bra{\rho_\N^{1/2}} f(\Delta_\N(\si_\N,\omega_\N))\ket{\rho_\N^{1/2}}&= \bra{\rho_\N^{1/2}} f(V_\rho^*\Delta_\M(\si,R_\rho(\omega_\N))V_\rho)\ket{\rho_\N^{1/2}}\\&\le
\bra{\rho_\N^{1/2}} V_\rho^* f(\Delta_\M(\si,R_\rho(\omega_\N)))V_\rho\ket{\rho_\N^{1/2}}
\\&=
\bra{\rho^{1/2}} f(\Delta_\M(\si,R_\rho(\omega_\N)))\ket{\rho^{1/2}}.\qedhere
\end{align*}
\end{proof}
\bigskip
For all $\varepsilon>0$, we can choose $\omega_\N\in D_+(\N)$ such that
\begin{align*}\bra{\rho_\N^{1/2}} f(\Delta_\N(\si_\N,\omega_\N))\ket{\rho_\N^{1/2}}\ge \widetilde{Q}_f(\rho_\N\|\si_\N)-\varepsilon.
\end{align*}
Note that by Lemma \ref{omega}, \[\bra{\rho^{1/2}}\Delta_\M(\si,R_\rho(\omega_\N))\ket{\rho^{1/2}}= \bra{\rho_\N^{1/2}}V_\rho^*\Delta_\M(\si,R_\rho(\omega_\N))V_\rho\ket{\rho_\N^{1/2}}
=\bra{\rho_\N^{1/2}} \Delta_\N(\si_\N,\omega_\N)\ket{\rho_\N^{1/2}}\pl.\]
Then
\begin{align*}
&\widetilde{Q}_f(\rho\|\si)-\widetilde{Q}_f(\rho_\N\|\si_\N)\\
& =
\sup_{\omega\in D_+(\M)}\bra{\rho^{1/2}} f(\Delta_\M(\si,\omega))\ket{\rho^{1/2}}-\sup_{\omega\in D_+(\N)}\bra{\rho_\N^{1/2}} f(\Delta_\N(\si_\N,\omega))\ket{\rho_\N^{1/2}}
\\
& \ge
\sup_{\omega\in D_+(\M)}\bra{\rho^{1/2}} f(\Delta_\M(\si,\omega))\ket{\rho^{1/2}}-\bra{\rho_\N^{1/2}} f(\Delta_\N(\si_\N,\omega_\N))\ket{\rho_\N^{1/2}}-\varepsilon
\\
& \ge
\bra{\rho^{1/2}} f(\Delta_\M(\si,R_\rho(\omega_\N)))\ket{\rho^{1/2}}-\bra{\rho_\N^{1/2}} f(\Delta_\N(\si_\N,\omega_\N))\ket{\rho_\N^{1/2}}-\varepsilon
\\
& = b\bra{\rho^{1/2}} \Delta_\M(\si,R_\rho(\omega_\N))\ket{\rho^{1/2}}-b \bra{\rho_\N^{1/2}} \Delta_\N(\si_\N,\omega_\N)\ket{\rho_\N^{1/2}}\\
\pl &+\int_0^\infty \left(\bra{\rho^{1/2}} (\Delta_\M(\si,R_\rho(\omega_\N))+\la)^{-1}\ket{\rho^{1/2}}-\bra{\rho_\N^{1/2}} (\Delta_\N(\si_\N,\omega_\N)+\la)^{-1}\ket{\rho_\N^{1/2} } \right)d\nu(\la)-\varepsilon\pl.
\\
& = \int_0^\infty \left(\bra{\rho^{1/2}} (\Delta_\M(\si,R_\rho(\omega_\N))+\la)^{-1}\ket{\rho^{1/2}}-\bra{\rho_\N^{1/2}} (\Delta_\N(\si_\N,\omega_\N)+\la)^{-1}\ket{\rho_\N^{1/2} } \right)d\nu(\la)-\varepsilon\pl.
\end{align*}
where $b$ is the parameter and $d\nu$ is the measure in the integral representation \eqref{om} of $f$. Denote
\begin{align*}
& F(\la) \\
& :=\bra{\rho^{1/2}} (\Delta_\M(\si,R_\rho(\omega_\N))+\la)^{-1}\ket{\rho^{1/2}}-\bra{\rho_\N^{1/2}} (\Delta_\N(\si_\N,\omega_\N)+\la)^{-1}\ket{\rho_\N^{1/2} } ,\\
& = \lan \rho_\N^{1/2}|V_\rho^* (\Delta_{\M}(\si,R_\rho(\omega_\N))+\la)^{-1}V_\rho|\rho_\N^{1/2}\ran-\lan \rho_\N^{1/2}| (V_\rho^*(\Delta_{\M}(\si,R_\rho(\omega_\N))+\la)V_\rho)^{-1}|\rho_\N^{1/2}\ran \\
& =\bra{u_\la}\Delta_\M(\si,R_\rho(\omega_\N))+\la \ket{u_\la}\ge 0
\end{align*}
where
\[
\ket{u_\la}:=(\Delta_\M(\si,R_\rho(\omega_\N))+\la)^{-1}\ket{\rho^{1/2}}-V_\rho (\Delta_\N(\si_\N,\omega_\N)+\la)^{-1}\ket{\rho_\N^{1/2}}\pl,
\]
and the last line follows from Lemma~\ref{key}.
Thus, we find that
\begin{equation}
\widetilde{Q}_f(\rho\|\si)-\widetilde{Q}_f(\rho_\N\|\si_\N) \geq \int_0^\infty F(\la) \, d\nu(\la)-\varepsilon.
\label{eq:lower-bnd-optimized-f-diff-big-F}
\end{equation}

\begin{lemma}
\label{3.19}
Let $t\in \mathbb{R}$
and
\[
\ket{u_t}\coloneqq  \frac{\cosh(\pi t)}{\pi} \Delta_\M(\si,R_\rho(\omega_\N))^{1/2+it} \int_{0}^\infty \la^{-1/2-it} \ket{u_\la}\, d\la\pl.\]
Then the following inequality holds
\[
\norm{\rho-R_\si^{-t}(\rho)}{1}\le 2\norm{\ket{u_t}}{2} .
\]
\end{lemma}

\begin{proof}
Using the integral representation in \eqref{eq:integral-rep-x-min-1-2-it}, i.e.,
\[
x^{-1/2-it} = \frac{\cosh(\pi t )}{\pi} \int_{0}^\infty \la^{-1/2-it} (\la+x)^{-1}d\la,
\]
we find, by a similar argument to that given for \eqref{eq:v_t_mod_1}--\eqref{eq:v_t_mod_3}, that
\begin{multline*}
\Delta_\M(\si,R_\rho(\omega_\N))^{-1/2-it}\ket{\rho^{1/2}}-V_\rho \Delta_\N(\si_\N,\omega_\N)^{-1/2-it}\ket{\rho_\N^{1/2}} \\
 = \frac{\cosh(\pi t)}{\pi} \int_{0}^\infty \la^{-1/2-it} \ket{u_\la}d\la\pl.
\end{multline*}
Applying $\Delta_\M(\si,R_\rho(\omega_\N))^{1/2+it}$,
we find that
\begin{multline*}
\ket{u_t} = \ket{\rho^{1/2}}-\Delta_\M(\si,R_\rho(\omega_\N))^{1/2+it}V_\rho \Delta_\N(\si_\N,\omega_\N)^{-1/2-it}\ket{\rho_\N^{1/2}}\\
 = \ket{\rho^{1/2}}-\ket{\si^{1/2+it}\si_\N^{-1/2-it}\rho_\N^{1/2}\omega_\N^{1/2+it} \rho_\N^{-1/2}\rho^{1/2}R_\rho(\omega_\N)^{-1/2-it}}.
\end{multline*}
For the second term above, we have the following collapse:
\begin{align*}&\si^{\frac12+it}\si_\N^{-\frac12-it}\rho_\N^{\frac12}\omega_\N^{\frac12+it} \rho_\N^{-\frac12}\rho^{\frac12}R_\rho(\omega_\N)^{-\frac12-it}\Big(\si^{\frac12+it}\si_\N^{-\frac12-it}\rho_\N^{\frac12}\omega_\N^{\frac12+it} \rho_\N^{-\frac12}\rho^{\frac12}R_\rho(\omega_\N)^{-\frac12-it}\Big)^*
\\
& = \si^{\frac12+it}\si_\N^{-\frac12-it}\rho_\N^{\frac12}\omega_\N^{\frac12+it} \rho_\N^{-\frac12}\rho^{\frac12}R_\rho(\omega_\N)^{-\frac12-it}R_\rho(\omega_\N)^{-\frac12+it}\rho^{\frac12}\rho_\N^{-\frac12}\omega_\N^{\frac12-it}\rho_\N^{\frac12}\si_\N^{-\frac12+it}\si^{\frac12-it}
\\
& = \si^{\frac12+it}\si_\N^{-\frac12-it}\rho_\N^{\frac12}\omega_\N^{\frac12+it} \rho_\N^{-\frac12}\rho^{\frac12}R_\rho(\omega_\N)^{-1}\rho^{\frac12}\rho_\N^{-\frac12}\omega_\N^{\frac12-it}\rho_\N^{\frac12}\si_\N^{-\frac12+it}\si^{\frac12-it}
\\
& = \si^{\frac12+it}\si_\N^{-\frac12-it}\rho_\N^{\frac12}\omega_\N^{\frac12+it} \omega_\N^{-1}\omega_\N^{\frac12-it}\rho_\N^{\frac12}\si_\N^{-\frac12+it}\si^{\frac12-it}
\\
& =\si^{\frac12+it}\si_\N^{-\frac12-it}\rho_\N\si_\N^{-\frac12+it}\si^{\frac12-it}
\\
& = R_\si^{-t}(\rho),
\end{align*}
where $R_\si^{-t}$ is defined through \eqref{eq:rotated-petz-map-def}.
For the third equality above, we used the following:
\[
R_\rho(\omega_\N)=\rho^{1/2}\rho_\N^{-1/2}\omega_\N\rho_\N^{-1/2}\rho^{1/2} \qquad \Longleftrightarrow \qquad \rho_\N^{1/2}\rho^{-1/2}R_\rho(\omega_\N)\rho^{-1/2}\rho_\N^{1/2}=\omega_\N.
\]
After applying \eqref{2}, we  find that
\begin{align*}
\norm{\rho-R_\si^{-t}(\rho)}{1}& \le  2\norm{\ket{\rho^{1/2}}-\Delta_\M^{1/2+it}(\si,R_\rho(\omega_\N))^{1/2}V_\rho \Delta_\N(\si_\N,\omega_\N)^{-1/2-it}\ket{\rho_\N^{1/2}}}{2}
\\
& =  \frac{2\cosh(\pi t )}{\pi}\norm{\Delta_\M^{1/2}(\si,R_\rho(\omega_\N))\int_{0}^\infty \la^{-1/2-it} \ket{u_\la}}{2}=2\norm{\ket{u_t}}{2}.\qedhere
\end{align*}
\end{proof}
\bigskip

\begin{lemma}
\label{fd3}
Let $f:(0,\infty)\to \mathbb{R}$ be a regular operator anti-monotone function, and let $d\nu$ be the measure in its integral representation. Suppose for some $S$ and $T$, satisfying $0<S<T<\infty$, that $d\la\le c(S,T)\, d\nu(\la)$ for $c(S,T)>0$. Then
\begin{multline*}
\norm{\rho-R_\si^t(\rho_\N)}{1} \le \frac{2\cosh(\pi t )}{\pi}\Bigg(4S^{1/2}\widetilde{Q}_{x^{-1}}(\rho\|\si)^{1/2}\\
+\left(c(S,T)\ln\!\left(T/S\right)\right)^{1/2} (\widetilde{Q}_f(\rho\|\si)-\widetilde{Q}_f(\rho_\N\|\si_\N))^{1/2} + 4T^{-1/2} \Bigg) ,
\end{multline*}
where $\widetilde{Q}_{x^{-1}}(\rho\|\si)=\norm{\rho^{1/2}\si^{-1}\rho^{1/2}}{\infty}=\inf \{\la >0\pl |\pl \rho\le \la\si \}$.
\end{lemma}

\begin{proof}
The following argument is similar to the case of the non-optimized $Q_f$, as presented in the proof of Lemma~\ref{fd2}. We employ the shorthand $\Delta_\M \equiv \Delta_\M(\si,R_\rho(\omega_\N))$. Applying Lemma~\ref{3.19} and the triangle inequality, we find that
\begin{align*}
\norm{\rho-R_\si^{-t}(\rho)}{1} & \le \frac{2\cosh(\pi t )}{\pi}\Bigg(\norm{\Delta_\M^{1/2+it}\int_{0}^S \la^{-1/2-it} \ket{u_\la}\, d\la}{2}+\norm{\Delta_\M^{1/2+it}\int_{S}^T \la^{-1/2-it} \ket{u_\la}\, d\la}{2}\\
&\qquad\qquad\qquad\qquad +\norm{\Delta_\M^{1/2+it}\int_{T}^\infty \la^{-1/2-it} \ket{u_\la}\, d\la}{2}\Bigg)\\
& = \frac{2\cosh(\pi t )}{\pi}\Bigg(\norm{\Delta_\M^{1/2}\int_{0}^S \la^{-1/2-it} \ket{u_\la}\, d\la}{2}+\norm{\Delta_\M^{1/2}\int_{S}^T \la^{-1/2-it} \ket{u_\la}\, d\la}{2}\\
&\qquad\qquad\qquad\qquad +\norm{\Delta_\M^{1/2}\int_{T}^\infty \la^{-1/2-it} \ket{u_\la}\, d\la}{2}\Bigg)\\
& =\frac{2\cosh(\pi t )}{\pi}\left(\norm{\operatorname{I}}{2}+\norm{\operatorname{II}}{2}+\norm{\operatorname{III}}{2}
\right).
\end{align*}
For each term, we argue similarly as in the proof of Lemma~\ref{fd2}, but implicitly using Lemma~\ref{3.19} and \eqref{eq:lower-bnd-optimized-f-diff-big-F}:
\begin{align*}
\norm{\operatorname{I}}{2}& =\norm{\Delta_\M^{1/2}\int_{0}^S \la^{-1/2-it} \ket{u_\la}\, d\la}{2}\le 4 S^{1/2} \bra{\rho^{1/2}}\Delta_\M^{-1}\ket{\rho^{1/2}}^{1/2}\\
& \le 4 S^{1/2} \widetilde{Q}_{x^{-1}}(\rho\|\si)^{1/2},\\
\norm{\operatorname{II}}{2}& =\norm{\int_{S}^T \la^{-1/2-it} \Delta_\M^{1/2}\ket{u_\la}\, d\la}{2}\le \left(\int_{S}^T \la^{-1}d\la\right)^{1/2} \left(\int_S^T\norm{\Delta_\M^{1/2}\ket{u_\la}}{2}^2 \, d\la\right)^{1/2}
\\&\le \sqrt{c(S,T)\ln(T/S)(\widetilde{Q}_f(\rho\|\si)-\widetilde{Q}_f(\rho_\N\|\si_\N)+\varepsilon)} \pl,
\\ \norm{\operatorname{III}}{2}& =\norm{\int_{T}^\infty \la^{-1/2-it} \Delta_\M^{1/2}\ket{u_\la}\, d\la}{2}\le 4T^{-1/2}.
\end{align*}
Note that here
 \begin{align*}\widetilde{Q}_{x^{-1}}(\rho\|\si)=\sup_{\omega\in D_+(\M)}\bra{\rho^{1/2}}\Delta(\si,\omega)^{-1}\ket{\rho^{1/2}}=\sup_{\omega\in D_+(\M)}\tau(\rho^{1/2}\si^{-1}\rho^{1/2}\omega)=\norm{\rho^{1/2}\si^{-1}\rho^{1/2}}{\infty}\pl.
 \end{align*}
 is related to the max-relative entropy $D_\infty(\rho\|\si)=\log \inf \{ \la \pl | \pl \rho\le \la\si \pl\}$ \cite{Datta09}.
Since $\varepsilon$ is arbitrary and the upper bound is symmetric in $t$, we arrive at the statement of the lemma.
\end{proof}

\subsubsection{Recoverability for sandwiched R\'enyi relative (quasi-)entropy}

We now turn to the sandwiched R\'enyi relative (quasi-)entropy and identify physically meaningful refinements of its data-processing inequality. Let $\al\in [1/2,1)\cup (1,\infty]$ and set $\al':= \al/(\al-1)$, so that $\frac{1}{\al}+\frac{1}{\al'}=1$. The sandwiched R\'enyi relative entropy is given by
\[
\widetilde{D}_\al(\rho\|\si)=\al'\log\norm{\rho^{1/2}\si^{-\frac{1}{\al'}}\rho^{1/2}}{\al}\pl.
\]

\begin{theorem}
\label{fd4}
Let $\al\in (1/2,1)\cup(1,\infty)$, $\al'=\frac{\al}{\al-1}$, and $\varepsilon \in \left(0,\tfrac{1-1/|\al'|}{2}\right)$. Let $\widetilde{Q}_\al(\rho\|\si) \coloneqq \norm{\rho^{1/2}\si^{-\frac{1}{\al'}}\rho^{1/2}}{\al}$ denote the sandwiched $\al$-R\'enyi relative quasi-entropy. Let $\widetilde{Q}_\infty(\rho\|\si) \coloneqq \norm{\rho^{1/2}\si^{-1}\rho^{1/2}}{\infty}$. Then
\begin{align}
&|\widetilde{Q}_\al(\rho\|\si)-\widetilde{Q}_\al(\rho_\N\|\si_\N)|\ge \left(K(\al ,\widetilde{Q}_\infty(\rho\|\si),\varepsilon)\frac{\pi}{2}\norm{\rho-R_\si(\rho_\N)}{1}\right)^{\frac{1}{\frac{1-1/|\al'|}{2}-\varepsilon}} \label{eq1d},\\
&|\widetilde{Q}_\al(\rho\|\si)-\widetilde{Q}_\al(\rho_\N\|\si_\N)|\ge  \left(K(\al ,\widetilde{Q}_\infty(\rho\|\si),\varepsilon)\frac{\pi}{2\cosh{\pi t}}\norm{\rho-R_\si^t(\rho_\N)}{1}\right)^{\frac{1}{\frac{1-1/|\al'|}{2}-\varepsilon}} \label{eq2d} ,
\\
&|\widetilde{Q}_\al(\rho\|\si)-\widetilde{Q}_\al(\rho_\N\|\si_\N)|\ge \left(K(\al ,\widetilde{Q}_\infty(\rho\|\si),\varepsilon)\, \frac{1}{2}\norm{\rho-R_\si^u(\rho_\N)}{1}\right)^{\frac{1}{\frac{1-1/|\al'|}{2}-\varepsilon}}
\label{eq3d}.
\end{align}
where the constant
\begin{align}
K(\widetilde{Q}_\infty(\rho\|\si),\al,\varepsilon)
\coloneqq \Big(
 4\widetilde{Q}_{\infty}(\rho\|\si)^{1/2} +\left(\frac{\pi}{e\varepsilon\sin (\pi|\frac{1}{\al'}|)}\right)^{1/2}+4\Big)^{-1}.
\label{con}
\end{align}
\end{theorem}

\begin{proof}
For $1<\al,\al'<\infty$, the function $x^{-\frac{1}{\al'}}$ is operator convex and operator anti-monotone. We have
\begin{align*} \widetilde{Q}_{x^{-\frac{1}{\al'}}}(\rho\|\si)& =\sup_{\omega\in D_+(\M)} \bra{\rho^{1/2}} \Delta(\si,\omega)^{-1/\al'}\ket{\rho^{1/2}}
\\ & =\sup_{\omega\in D_+(\M)} \tau(\rho^{1/2}\si^{-1/\al'}\rho^{1/2}\omega^{1/\al'})
=\norm{\rho^{1/2}\si^{-\frac{1}{\al'}}\rho^{1/2}}{\al}.
\end{align*}
Thus, for $1<\al\le \infty$,
\[\widetilde{D}_\al(\rho\|\si)=\al'\log \widetilde{Q}_{x^{-1/\al'}}(\rho\|\si)\pl.\]
Writing $0<\beta \coloneqq 1/\al'<1$,
 the integral representation of $x^{-\beta}$ is as follows:
\[x^{-\beta}=\frac{\sin (\pi \beta)}{\pi}\int_0^\infty \la^{-\beta}\frac{1}{\la+x}\,d\la\pl.\]
The constant $c(S,T)\le \frac{\pi}{\sin (\pi \beta)}T^{\beta}$.
Then by Lemma~\ref{fd3}, we have
\begin{multline*}
\norm{\rho-R_\si^t(\rho_\N)}{1} \le
\frac{2\cosh(\pi t)}{\pi}\Bigg(4S^{1/2}\widetilde{Q}_{x^{-1}}(\rho\|\si)^{1/2}\\
+\sqrt{\frac{\pi}{\sin (\pi \beta)}T^\beta\ln (T/S) }(\widetilde{Q}_{x^{-\beta}}(\rho\|\si)-\widetilde{Q}_{x^{-\beta}}(\rho_\N\|\si_\N))^{1/2}+4T^{-1/2}\Bigg)\pl.
\end{multline*}
Choose $S=T^{-1}=\delta$ and $\delta\coloneqq \min\{ |\widetilde{Q}_{x^{-\beta}}(\rho\|\si)-\widetilde{Q}_{x^{-\beta}}(\rho_\N\|\si_\N)|,  1\}$. Thus
\begin{align*}
\norm{\rho-R_\si^t(\rho_\N)}{1} & \le
\frac{2\cosh(\pi t)}{\pi}\left(4\delta^{1/2}\widetilde{Q}_{x^{-1}}(\rho\|\si)^{1/2} +\sqrt{\frac{2\pi}{\sin (\pi \beta)}|\ln\delta| }\delta^{\frac{1-\beta}{2}}+4\delta^{1/2}\right)
\\
& \le
\frac{2\cosh(\pi t)}{\pi}\left(4\widetilde{Q}_{x^{-1}}(\rho\|\si)^{1/2} +\sqrt{\frac{\pi}{e\varepsilon\sin (\pi \beta)}}+4\right)\delta^{\frac{1-\beta}{2}-\varepsilon}.
\end{align*}
The reasoning for these steps is similar to that given for \eqref{eq:delta-arg-1st-time-1}--\eqref{eq:delta-arg-1st-time-last}.

For $1/2\le \al<1$, which implies that $  \al'\le -1$, the function $x^{-\frac{1}{\al'}}$ is operator monotone and operator concave because $-\frac{1}{\al'}\in(0,1)$. We have
\begin{align*}\widetilde{Q}_{-x^{-\frac{1}{\al'}}}(\rho\|\si)& =\sup_{\omega\in D_+(\M)} -\bra{\rho^{1/2}} \Delta(\si,\omega)^{-1/\al'}\ket{\rho^{1/2}}
\\ & =-\inf_{\omega\in D_+(\M)} \tau(\rho^{1/2} \si^{-1/\al'}\rho^{1/2}\omega^{1/\al'})
= - \norm{\rho^{1/2} \si^{-1/\al'}\rho^{1/2}}{\al}.
\end{align*}
Then for $ 1/2\le \al<1$,
\[\widetilde{D}_\al(\rho\|\si)=\al'\log \big(-\widetilde{Q}_{-x^{-1/\al'}}(\rho\|\si)\big)\pl.\]
Let $\gamma\coloneqq -1/\al'$. For $0<\gamma<1$, the integral representation is
\[x^\gamma=\frac{\sin (\pi \gamma)}{\pi}\int_0^\infty \la^\gamma\left(\frac{1}{\la}-\frac{1}{\la+x}\right)d\la\pl.\]
Then the constant $c(S,T)\le \frac{\pi}{\sin (\pi \gamma)}S^{-\gamma}$. By Lemma~\ref{fd3}, we have
\begin{multline*}
\norm{\rho-R_\si^t(\rho_\N)}{1}\le
\frac{2\cosh(\pi t)}{\pi}\Bigg( 4S^{1/2}\widetilde{Q}_{x^{-1}}(\rho\|\si)^{1/2} \\ +\sqrt{\frac{\pi}{\sin (\pi\gamma)}S^{-\gamma}\ln(T/S)}|\widetilde{Q}_{-x^\gamma}(\rho\|\si)-\widetilde{Q}_{-x^\gamma}(\rho_\N\|\si_\N)|^{1/2}+4T^{-1/2} \Bigg).
\end{multline*}
Set $S=T^{-1}=\delta$ and $\delta \coloneqq \min\{ |\widetilde{Q}_{-x^{\gamma}}(\rho\|\si)-\widetilde{Q}_{-x^{\gamma}}(\rho_\N\|\si_\N)|,  1\}$. Then
\begin{align*}
\norm{\rho-R_\si^t(\rho_\N)}{1}
& \le
\frac{2\cosh(\pi t)}{\pi}\left( 4\delta^{1/2}\widetilde{Q}_{x^{-1}}(\rho\|\si)^{1/2} +\sqrt{\frac{2\pi}{\sin (\pi\gamma)}|\ln\delta|}\delta^{\frac{1-\gamma}{2}}+4\delta^{1/2} \right)
\\
& \le
\frac{2\cosh(\pi t)}{\pi}\left( 4\widetilde{Q}_{x^{-1}}(\rho\|\si)^{1/2} +\sqrt{\frac{\pi}{e\varepsilon\sin (\pi\gamma)}}+4\right)\delta^{\frac{1-\gamma}{2}-\varepsilon}.
\end{align*}
The reasoning for these steps is similar to that given for \eqref{eq:delta-arg-1st-time-1}--\eqref{eq:delta-arg-1st-time-last}.
\end{proof}

\bigskip
We then find the following for the sandwiched R\'enyi relative entropy:
\begin{cor}
\label{renyi}
Let $\M$ be a finite-dimensional von Neumann algebra, and let $\N\subset \M$ be a subalgebra. Let $\rho$ and $\si$ be two faithful states. Let $\al\in (1/2,1)\cup (1,\infty)$ and $\al'=\al/(\al-1)$, so that $1/\al+1/\al'=1$. Set $t\in \mathbb{R}$ and $\varepsilon\in (0,\frac{1-1/|\al'|}{2})$. For $\alpha \in (1/2,1)$, the following inequality holds
\begin{multline*}
\widetilde{D}_\al(\rho\|\si)-\widetilde{D}_\al(\rho_\N\|\si_\N) \ge \\
|\al'|\log\left(1+
\left(K(\al ,\varepsilon ,\widetilde{Q}_{\infty}(\rho\|\si))\frac{\pi}{2\cosh{\pi t}} \norm{\rho-R_\si^t(\rho_\N)}{1}\right)^{\frac{1}{\frac{1-1/|\al'|}{2}-\varepsilon}}\right),
\end{multline*}
and for $\alpha >1$, the following inequality holds
\begin{multline*}
\widetilde{D}_\al(\rho\|\si)-\widetilde{D}_\al(\rho_\N\|\si_\N) \ge \\
\al'\log \!\left(1+\frac{ 1 }{\widetilde{Q}_\infty(\rho\|\si)^{\frac{1}{\alpha'}}}\left(K(\al ,\varepsilon ,\widetilde{Q}_{\infty}(\rho\|\si)) \frac{\pi}{2\cosh{\pi t}}
\norm{\rho-R_\si^t(\rho_\N)}{1}\right)^{\frac{1}{\frac{1-1/|\al'|}{2}-\varepsilon}}\right),
\end{multline*}
where the constant $K(\al ,\varepsilon,\widetilde{Q}_{\infty}(\rho\|\si))$ is given in \eqref{con}.
\end{cor}

\begin{proof}
For $1/2<\al<1$ and $\al'\le -1$, we find that
\begin{align*}
& \widetilde{D}_\al(\rho\|\si)-\widetilde{D}_\al(\rho_\N\|\si_\N)\\
& = |\al'|\log \frac{ \widetilde{Q}_\al(\rho_\N\|\si_\N)}{\widetilde{Q}_\al(\rho\|\si)}
\\
& = |\al'|\log \left(1+\frac{ \widetilde{Q}_\al(\rho_\N\|\si_\N)-\widetilde{Q}_\al(\rho\|\si)}{\widetilde{Q}_\al(\rho\|\si)}\right)
\\
& \ge  |\al'|\log \left(1+\frac{ 1 }{\widetilde{Q}_\al(\rho\|\si)}\left(K(\al ,\varepsilon ,\widetilde{Q}_{\infty}(\rho\|\si)) \frac{\pi}{2\cosh{\pi t}}\norm{\rho-R_\si^t(\rho_\N)}{1}\right)^{\frac{1}{\frac{1-1/|\al'|}{2}-\varepsilon}}\right),
\\ & \ge
|\al'|\log\left(1+
\left(K(\al ,\varepsilon ,\widetilde{Q}_{\infty}(\rho\|\si))\frac{\pi}{2\cosh{\pi t}} \norm{\rho-R_\si^t(\rho_\N)}{1}\right)^{\frac{1}{\frac{1-1/|\al'|}{2}-\varepsilon}}\right).
\end{align*}
The first inequality follows from \eqref{eq2d}, and the second
 follows because  $\widetilde{Q}_{\al}(\rho\|\si)\le 1$ for $\al \in (1/2,1)$.

For $\al>1$, consider that
\begin{align*}
&\widetilde{D}_\al(\rho\|\si)-\widetilde{D}_\al(\rho_\N\|\si_\N)\\
& = \al'\log \frac{ \widetilde{Q}_\al(\rho\|\si)}{\widetilde{Q}_\al(\rho_\N\|\si_\N)}
\\& = \al'\log \!\left(1+\frac{ \widetilde{Q}_\al(\rho\|\si)-\widetilde{Q}_\al(\rho_\N\|\si_\N)}{\widetilde{Q}_\al(\rho_\N\|\si_\N)}\right)
\\ & \ge  \al'\log \!\left(1+\frac{ 1 }{\widetilde{Q}_\al(\rho_\N\|\si_\N)}\left(K(\al ,\varepsilon ,\widetilde{Q}_{\infty}(\rho\|\si))\frac{\pi}{2\cosh{\pi t}}\norm{\rho-R_\si^t(\rho_\N)}{1}\right)^{\frac{1}{\frac{1-1/|\al'|}{2}-\varepsilon}}\right)
\\ & \ge  \al'\log \!\left(1+\frac{ 1 }{\widetilde{Q}_\infty(\rho\|\si)^{\frac{1}{\alpha'}}}\left(K(\al ,\varepsilon ,\widetilde{Q}_{\infty}(\rho\|\si)) \frac{\pi}{2\cosh{\pi t}}
\norm{\rho-R_\si^t(\rho_\N)}{1}\right)^{\frac{1}{\frac{1-1/|\al'|}{2}-\varepsilon}}\right).
\end{align*}
The first inequality follows from \eqref{eq2d}.
The second inequality follows from the inequalities $\widetilde{Q}_\al(\rho_\N\|\si_\N)\le \widetilde{Q}_\al(\rho\|\si) \le \widetilde{Q}_\infty(\rho\|\si)^{\frac{1}{\al'}}$. The first is a consequence of the data-processing inequality and the second a consequence of the monotonicity of the sandwiched R\'enyi relative entropies with respect to $\alpha$ (for the latter, see \cite[Theorem~7]{Muller13} and \cite[Lemma~8]{BVT18}).
\end{proof}

\begin{rem}{\rm For $\al=1$, $\widetilde{D}_\al$ coincides with the standard relative entropy $D$, for which  results are given in Theorem~\ref{rsi}. For the two boundary cases $\al=1/2$ and $\al=\infty$, the  recoverability result in Corollary~\ref{renyi} does not hold. The $\al=1/2$ case corresponds to the root fidelity
\[ \sqrt{F}(\rho,\si)=\norm{\rho^{1/2}\si\rho^{1/2}}{1/2}=-\widetilde{Q}_{-x}(\rho\|\si),\]
and $\al=\infty$ to
\[
\widetilde{Q}_\infty(\rho\|\si)=\widetilde{Q}_{x^{-1}}(\rho\|\si)=\inf \{\la \pl |\pl \rho\le \la \si \}\pl.
\]
Our method fails for these two cases because
both operator anti-monotone functions $f(x)=-x$ and $g(x)=x^{-1}$ have trivial measure $d\nu$ in their integral representations.
Indeed, for both cases, it was already observed in \cite[Remarks 5.15 \& 5.16]{HM17} that there are examples for which the data-processing inequality for fidelity is saturated, i.e., $F(\rho,\si)=F(\rho_\N,\si_\N)$ (resp. $\widetilde{Q}_\infty(\rho\|\si)=\widetilde{Q}_\infty(\rho_\N\|\si_\N) $), but it is not for the relative entropy $D(\rho\|\si)>D(\rho_\N\|\si_\N)$, which implies that the existence of any exact recovery map is impossible. This extends the results in \cite{Jenvcova17}.
}
\end{rem}
The following are reversibility results as consequence of recoverability estimates. Note that the faithfulness assumption can be weaken to $s(\rho)\le s(\si)$ as in Corollary \ref{cor:reversibility-vNa} for the von Neumann algebra case.
\begin{cor}
\label{cor:many-equivs-f-div}
Let $\rho$ and $\si$ be faithful quantum states.
The following are equivalent:
\begin{enumerate}
\item[i)] $D(\rho\|\si)=D(\rho_\N\|\si_\N)$.
\item[ii)]$\widetilde{D}_\al(\rho\|\si) = \widetilde{D}_\al(\rho_\N\|\si_\N)$ for some $\al\in(1/2,1)\cup (1,\infty)$ where $\widetilde{D}_\al$ is the $\al$-sandwiched R\'enyi relative entropy.
\item[iii)]$\widetilde{Q}_f(\rho\|\si)=\widetilde{Q}_f(\rho_\N\|\si_\N)$ for some regular operator anti-monotone function $f$.
\item[iv)]$\widetilde{Q}_f(\rho\|\si)=\widetilde{Q}_f(\rho_\N\|\si_\N)$ for all operator anti-monotone functions $f$.
\item[v)]$R_\rho^{t}(\si_\N)=\si$ for all $t \in \mathbb{R}$.
\item[vi)]$R_\si^{t}(\rho_\N)=\rho$ for all $t \in \mathbb{R}$.
\item[vii)]there exists some CPTP map $\Phi:L_1(\N)\to L_1(\M)$ such that $\Phi(\rho_\N)=\rho$ and $\Phi(\si_\N)=\si$.
\end{enumerate}
\end{cor}

\begin{proof}
The implications v) $\Rightarrow$ vii) and vi) $\Rightarrow$ vii) are trivial. vii) $\Rightarrow$ i)-iv) follows from the data-processing inequality. Note that for faithful $\rho$ and $\si$, $Q_{x^2}(\rho\Vert\si),Q_{x^{-1}}(\rho\Vert\si),\widetilde{Q}_{\infty}(\rho\Vert\si)<\infty$ are finite.

 i) $\Rightarrow$ v) follows from Theorem~\ref{re}. i) $\Rightarrow$ vi) uses Theorem~\ref{rsi}.
ii) $\Rightarrow$ vi) uses Corollary~\ref{renyi}. iii) $\Rightarrow$ vi) follows from Theorem~\ref{fd4}. iv) $\Rightarrow$ iii) is trivial.
\end{proof}

\begin{rem}{\rm It follows from \cite{petz86} and \cite{universal} that the same equivalences hold for the standard $f$-divergence $Q_f$ and the Petz--R\'enyi relative entropy $D_\al$. Corollary~\ref{cor:many-equivs-f-div} above shows that the preservation of a ``regular'' optimized $f$-divergence is also equivalent to the existence of a recovery map.
}
\end{rem}


\section{Optimized $f$-divergence in von Neumann algebras}

\label{sec:opt-f-div}

\subsection{Definition of optimized $f$-divergence}

In this section, we define the optimized $f$-divergence for states of a general von Neumann algebra. We also prove the data-processing inequality for the optimized $f$-divergence. We refer to Appendix~\ref{app:vNa-basics} for a review of the basics of von Neumann algebras and the notations used in this section. We first define the optimized $f$-divergence for two states $\rho$ and $\si$ with the support assumption $s(\rho)\le s(\si)$.

\begin{definition}
\label{def:optimized-f-vNa}
Let $\M\subset B(H)$ be a von Neumann algebra. Let $\rho,\si$ be two normal states such that $s(\rho)\le s(\si)$, and let $\brho,\bsi\in H$ be their corresponding vector representations. For an operator anti-monotone function $f:(0,\infty)\to \R$, we define the optimized $f$-divergence $\widetilde{Q}_f(\rho\|\si)$ as follows:
\begin{align}
\label{eq:opt-f-div-vNa-def}
\widetilde{Q}_f(\rho\|\si)= \sup_{\bome\, :\, \norm{\bome}{2}=1,\, \brho\in [\M\bome] }\bra{\brho}f(\Delta(\bsi,\bome))\ket{\brho} \end{align}
where the supremum runs over all unit vectors $\bome\in H$ such that $\brho\in [\M\bome]$ and $\Delta(\bsi,\bome)$ is the relative modular operator. This definition of $\widetilde{Q}_f$ only depends on the states $\rho$ and $\si$, and is independent of the choice of the algebra representation $\M\subset B(H)$ and the vector representations $\ket{\brho}$ and $\ket{\bsi}$ for $\rho$ and $\si$ respectively.
\end{definition}

 If $f$ is a continuous function on $[0,\infty)$, we do not need the restriction $\brho\in [\M\bome]$ and
can take the supremum over all $\bome$ satisfying $\norm{\bome}{2}=1$. Otherwise, we have to require $\brho\in [\M'\bsi]$ and $\brho\in [\M\bome]$, since $\Delta(\bsi,\bome)$ is supported on $s_{\M}(\sigma)s_{\M'}(\omega')$, where
$\omega'$ is the state of the commutant $\M'$ implemented by the vector ${\bm \omega}$ and $s_{\M}(\sigma)$ (resp.~$s_{\M'}(\omega'))$ is the support projection of $\sigma$ (resp.~$\omega')$ on $\M$ (resp.~$\M')$.  The relative modular operator connects to the spatial derivative as follows:
\[\Delta(\bsi,\bome)=\Delta(\si/\bome)=\Delta(\si/\omega'),
\]
where $\omega'\in \M'_*$ is the state on $\M'$ implemented by the vector $\bome\in H$. Note that $\Delta(\si/\omega')$ and $\Delta(\omega'
/\si)$ have the same support and
\[\Delta(\si/\omega')=\Delta(\omega'
/\si)^{-1}\]
on their support. Then we have the following equivalent definition for the optimized $f$-divergence:
\begin{align}\label{eqde}
\widetilde{Q}_f(\rho\|\si)= \sup_{\bome\,:\,\norm{\bome}{2}=1, \brho\in [\M\bome] }\bra{\brho}\tilde{f}(\Delta(\omega'/\si))\ket{\brho}, \end{align}
where $\tilde{f}(x)=f(x^{-1})$ is operator monotone.
This latter definition via the spatial derivative is closer to the definition of the sandwiched R\'enyi relative entropy from \cite{BVT18}, which used Araki--Masuda $L_p$ spaces \cite{AM82}.

We now verify that the definition of $\widetilde{Q}_f$ in \eqref{eq:opt-f-div-vNa-def} is independent of vector representations. Note that the representation $\pi$ in the following need not be faithful.

\begin{prop}
\label{ind2}
Let $\pi:\M\to B(H_1)$ be a $*$-representation, and let $\brho_1,\bsi_1\in H_1$ be the unit vectors implementing $\rho$ and $\si$, respectively, via $\pi$. Then
\[\widetilde{Q}_f(\rho\|\si)=\sup_{\bome_1 \, :\, \norm{\bome_1}{2}=1, \brho_1\in [\pi(\M)\bome_1]}\bra{\brho_1}f(\Delta(\bsi_1, \bome_1))\ket{\brho_1}.\]
\end{prop}

\begin{proof}
We follow the idea of \cite[Lemma 3]{BVT18} and use the equivalent definition from \eqref{eqde} with the spatial derivative. Consider that
\[\widetilde{Q}_f(\rho\|\si)=\sup_{\bome\,:\,\norm{\bome}{2}=1, \brho\in [\M\bome] }\bra{\brho}\tilde{f}(\Delta(\bome/\si))\ket{\brho}.\]
Define $V_\rho : H\to H_1$ as the partial isometry such that, for $\bet\in [\M\rho]^\perp$,
\[  V(a\brho+\bet) = \pi(a)\brho_1   \pl, \pl a\in \M\]
Since $\pi(a)V_\rho=V_\rho a$, we have $R_\si(V_\rho\brho)=V_\rho R_\si(\brho)$ (see \eqref{R} for the definition of operator $R_\si(\brho)$). Let $V\equiv V_\rho$. Then for all $\bxi\in [\M\bome]s(\si)H$ and $\bome_1\in H_1$, we find that
\begin{align*}
\bra{\bxi}V^*\Delta(\bome_1/\si)V\ket{\bxi}
& =\bra{\bome_1}R_{\si}(V\bxi)R_{\si}(V\bxi)^*\ket{\bome_1}\\
& =\bra{\bome_1}VR_{\si}(\bxi)R_{\si}(\bxi)^*V^*\ket{\bome_1}\\
& =\bra{V^*\bome_1}R_{\si}(\bxi)R_{\si}(\bxi)^*\ket{V^*\bome_1}\\
& = \bra{\bxi}\Delta(V^*\bome_1/\si)\ket{\bxi}.
\end{align*}
Moreover $s'(V^*\bome_1)=[\M V^*\bome_1]=[V^*\pi(\M)\bome_1]=V^*s'(\bome_1)V$ and hence
\[V^*\Delta(\bome_1/\si)V=\Delta(V^*\bome_1/\si),
\] with the same support for all $\bome_1\in H_1$ with $\rho\in [\pi(\M)\bome_1]$.
Since $\tilde{f}$ is operator concave and operator monotone
\begin{align*}\bra{\brho_1}\tilde{f}(\Delta(\bome_1/\si))\ket{\brho_1}
& =\bra{\brho}V^*\tilde{f}(\Delta(\bome_1/\si))V\ket{\brho}
\\& \le  \bra{\brho}\tilde{f}(\Delta(V^*\bome_1/\si))\ket{\brho}
\le \bra{\brho}\tilde{f}(\Delta(\overline{V^*\bome_1}/\si))\ket{\brho}\pl,
\end{align*}
where $\overline{V^*\bome_1}$ is the normalization of $V^*\bome_1$. Here we view $V$ as an isometry by restricting on the support $V^*V=[\pi_1(\M)\brho_1]$.
Therefore
\[\sup_{\bome_1 \, :\, \norm{\bome_1}{2}=1,\brho_1\in [\pi(\M)\bome_1] }\bra{\brho_1}\tilde{f}(\Delta(\bome_1/\si))\ket{\brho_1}
\le \sup_{\bome \, :\, \norm{\bome}{2}=1 , \brho\in [\M\bome]}\bra{\brho}\tilde{f}(\Delta(\bome/\si))\ket{\brho}\pl.\]
The converse direction follows by the symmetric role of the representations $\pi_1(\M)\subset B(H_1)$ and $\M\subset B(H)$.
\end{proof}

By the independence above, we can then carry the definition to the standard form $(\M,L_2(\M),J,L_2(\M)^+)$ using Haagerup $L_2$-spaces. Let $h_\rho\in L_1(\M)$ be the density operator corresponding to $\rho$. Then
\[
\widetilde{Q}_f(\rho\|\si)
=\sup_{\omega} \bra{h_\rho^{1/2}}f(\Delta(\si,\omega)\ket{h_\rho^{1/2}},
\]
where the supremum runs over all states $\omega$ such that $s(\omega)\ge s(\rho)$.
The next proposition shows that the definition above coincides with the finite-dimensional definition in \cite{wilde}, and one can further restrict to $\omega\gg\rho$; i.e., there exists $\la>0$ such that $\rho\le \la\omega$.

\begin{prop}
\label{dense}
Let $f:(0,\infty)\to \mathbb{R}$ be an operator anti-monotone and $\nu$ be the measure in the integral representation of $f$ as in \eqref{om}.
Suppose $\nu$ does not contain a point mass at $\la=0$. Then
\begin{align}
       \widetilde{Q}_f(\rho\|\si)&=\sup_{\omega\in D(\M)}\lim_{\varepsilon\to 0^+} \bra{h_\rho^{1/2}}f(\Delta(\si,\omega+\varepsilon \phi))\ket{h_\rho^{1/2}}  \label{equiv1}\\
&=\sup_{\omega\in D(\M),\, \omega\gg\rho} \bra{h_\rho^{1/2}}f(\Delta(\si,\omega))\ket{h_\rho^{1/2}},
\label{equiv}
\end{align}
where in \eqref{equiv1}, $\phi$ can be any normal state with $s(\rho)\le s(\phi)$.
\end{prop}
\begin{proof}For the first expression, we note that $\Delta(\si,\omega+\varepsilon\phi)^{1/2}\to  \Delta(\si,\omega)^{1/2}$ strongly in the resolvent sense by \cite[Proposition 4.9]{petzbook}.
This implies (by the integral representation of $f$) that
\[ \lim_{\varepsilon\to 0^+}\lan h_\rho^{1/2}|f(\Delta(\si,\omega+\varepsilon\phi))|h_\rho^{1/2}\ran=\lan h_\rho^{1/2}|f(\Delta(\si,\omega))|h_\rho^{1/2}\ran\pl.\]
For the second expression, we can choose
$\omega_\varepsilon=\varepsilon\rho+(1-\varepsilon)\omega$. By the same reasoning,
\[ \lim_{\varepsilon\to 0^+}\lan h_\rho^{1/2}|f(\Delta(\si,\omega_\varepsilon))|h_\rho^{1/2}\ran=\lan h_\rho^{1/2}|f(\Delta(\si,\omega))|h_\rho^{1/2}\ran\pl.\]
Note that $\rho\le \varepsilon^{-1}\omega_\varepsilon$. Then we have \[ \sup_{\omega\in D(\M),\omega\gg\rho} \bra{h_\rho^{1/2}}f(\Delta(\si,\omega))\ket{h_\rho^{1/2}}\ge \widetilde{Q}_f(\rho\|\si)\pl.\]
The inverse inequality is obvious.
\end{proof}
Following the same idea above, the optimized divergence for general two states $\rho$ and $\si$ can be defined as follows
\begin{align} \widetilde{Q}_f(\rho\|\si):=\lim_{\varepsilon\to 0^+} \widetilde{Q}_f(\rho\|\si+\varepsilon\rho)=\sup_{\varepsilon>0} \pl \widetilde{Q}_f(\rho\|\si+\varepsilon\rho)\pl. \label{limit}\end{align}
The above limit is increasing as $\varepsilon\to 0^+$ for all $\omega$ because $\Delta(\si+\varepsilon\rho,\omega)=\Delta(\si,\omega)+\varepsilon\Delta(\rho,\omega)$ and $f$ is operator anti-monotone.  For $\rho$ and $\si$ with $s(\rho)\le s(\si)$, this recovers the Definition \ref{def:optimized-f-vNa}
\begin{align}
\lim_{\varepsilon\to 0^+} \nonumber \widetilde{Q}_f(\rho\|\si+\varepsilon\rho)=&\sup_{\varepsilon>0}\sup_{\omega\in D(\M),\, \omega\gg\rho} \bra{h_\rho^{1/2}}f(\Delta(\si+\varepsilon\rho,\omega))\ket{h_\rho^{1/2}} \\=&\sup_{\omega\in D(\M),\, \omega\gg\rho} \sup_{\varepsilon>0}\pl \bra{h_\rho^{1/2}}f(\Delta(\si+\varepsilon\rho,\omega))\ket{h_\rho^{1/2}}
\nonumber \\=&\pl \widetilde{Q}_f(\rho\|\si)
\pl,\label{eq:limit}
\end{align}
For the last step above, by \cite[Proposition 4.9]{petzbook} we have that for each $\omega\gg\rho$ and $t>0$,
\[ \bra{h_\rho^{1/2}}(\Delta(\si+\varepsilon\rho,\omega)+t)^{-1}\ket{h_\rho^{1/2}}\to \bra{h_\rho^{1/2}}(\Delta(\si,\omega)+t)^{-1}\ket{h_\rho^{1/2}}\]
Using integral representation \eqref{om} of $f$ and monotone convergence theorem over $\varepsilon\to 0^+$, we have
\[ \bra{h_\rho^{1/2}}f(\Delta(\si+\varepsilon\rho,\omega))\ket{h_\rho^{1/2}}\to \bra{h_\rho^{1/2}}f(\Delta(\si,\omega))\ket{h_\rho^{1/2}},\]
which verifies \eqref{eq:limit}.

As the optimized $f$-divergence for general $\rho$ and $\si$ is defined through approximation, for most of the following discussion it suffices to consider $\widetilde{Q}_f(\rho\|\si)$ with support assumption.




\subsection{Comparison to standard $f$-divergence}

In this section, we first review the definition of $f$-divergence introduced by Petz in \cite{petz85,petz86}, which we call the standard $f$-divergence. Let $\M\subset B(H)$ be a von Neumann algebra, and let $\rho,\si$ be two normal states implemented by $\brho,\bsi\in H$, respectively. For an operator convex function $f:(0,\infty)\to \R$, the standard $f$-divergence is defined as follows:
\[Q_f(\rho\|\si) \coloneqq \bra{\brho}f(\Delta(\bsi, \brho))\ket{\brho}\pl, \pl\text{if}\pl s(\rho)\le s(\si)\]
which is also independent of the particular vector representation, as in Lemma~\ref{ind2}. Because the standard $f$-divergence $Q_f$ for general $\rho$ and $\si$ also admits approximation as in \eqref{limit} (see \cite{Hiai18}),
it is clear from definitions that
 \[\widetilde{Q}_f(\rho\|\si)\ge Q_f(\rho\|\si)\pl.\]

\begin{exam}{\rm  The sandwiched R\'enyi relative entropy was defined in \cite{BVT18} as
$\widetilde{D}_\al(\rho\|\si) \coloneqq \al'\log \widetilde{Q}_{\al}(\rho\|\si)$, where $\al'\coloneqq \al / (\al-1)$ and
\[\widetilde{Q}_{\al}(\rho\|\si):=\begin{cases}
                                \displaystyle \sup_{\bome\,:\,\|\bome\|=1}\norm{\Delta(\bome/\si)^{\frac{1}{2\al'}}\ket{\brho}}{H}^2 & \mbox{if } 1<\al\le \infty \\
                                \displaystyle \inf_{\bome\,:\,\|\bome\|=1\pl, \brho\in[\M\brho]}\norm{\Delta(\bome/\si)^{\frac{1}{2\al'}}\ket{\brho}}{H}^2 & \mbox{if } \frac{1}{2}\le\al<1.
                              \end{cases}\]
Note that \begin{align*}\norm{\Delta(\bome/\si)^{\frac{1}{2\al'}}\ket{\brho}}{H}^2& =\lan \brho| \Delta(\bome/\si)^{\frac{1}{\al'}}|\brho\ran
=\lan \brho| \Delta(\bsi,\bome)^{-\frac{1}{\al'}}|\brho\ran\pl.
\end{align*}
Thus we have \[\widetilde{Q}_{\al}(\rho\|\si)=\begin{cases}
                               \widetilde{Q}_{x^{-\frac{1}{\al'}}}(\rho\|\si) & \mbox{if } 1<\al\le \infty \\
                               -\widetilde{Q}_{-x^{-\frac{1}{\al'}}}(\rho\|\si) & \mbox{if } \frac{1}{2}\le \al<1
                             \end{cases}\pl.\]
}\end{exam}

\begin{exam}{\label{5.5}\rm For $f(x)=-\log x$, it was shown in \cite{wilde}, by invoking the Klein inequality, that for  $\M=B(H)$, the following equality holds
\[\widetilde{Q}_{-\log x}(\rho\|\si)= D(\rho\|\si)\pl.\]
For the general case, we immediately have that
\begin{align*}\widetilde{Q}_{-\log x}(\rho\|\si)& \ge Q_{-\log x}(\rho\|\si)
=D(\rho\|\si).
\end{align*}
On the other hand, since $t\mapsto \al'\log t$ is concave for $\al > 1$ (and hence $\al' >1$), we find that
\begin{align*}
\widetilde{Q}_{-\log x}(\rho\|\si)&=\sup_{\omega}\bra{\brho}-\log \Delta(\si,\omega)\ket{\brho}  = \sup_{\omega}\bra{\brho}\al'\log \Delta(\si,\omega)^{-\frac{1}{\al'}}\ket{\brho}
\\ &\le \al'\log \sup_{\omega}\bra{\brho} \Delta(\si,\omega)^{-\frac{1}{\al'}}\ket{\brho}
\le \widetilde{D}_\al(\rho\|\si)\pl.
\end{align*}
Moreover, it was proved in \cite[Theorem 13]{BVT18} that if $\rho\le c\si$ for some $c>0$, then
\[\lim_{\al\to 1^+}\widetilde{D}_{\al}(\rho\|\si)=D(\rho\|\si)\pl.\]
For general case, we have
\begin{align*}\widetilde{Q}_{-\log x}(\rho\|\si)=&\lim_{\eps\to 0^+} \widetilde{Q}_{-\log x}(\rho\|\si+\eps\rho)= \lim_{\eps\to 0^+}D(\rho\|\si+\eps\rho)=D(\rho\|\si)\pl.\end{align*}
 Here the second limit follows from the fact that $D(\rho\|\si+\eps\rho)$ is monotone non-decreasing and lower semi-continuity of $D$.
}
\end{exam}
Recall that we denote by $D_\al$ the Petz-R\'enyi relative entropy. For $\al=1$, we write $\widetilde{D}_{1}(\rho\Vert\si)=D_{1}(\rho\Vert\si):=D(\rho\Vert\si)$ as the standard relative entropy. The following lemma enables us to approximate relative entropy $D_\al(\rho\Vert\si)$ and $\widetilde{D}_\al(\rho\Vert\si)$ by $\rho,\si$ with $s(\rho)=s(\si)$.
\begin{lemma}\label{approximation}Let $\rho,\si\in D(\M)$ be two normal states with $s(\rho)\le s(\si)$. For $0<\varepsilon<1$, denote $\rho_\varepsilon=(1-\varepsilon)\rho+\varepsilon \si $.
Then \begin{enumerate}\item[i)]for any $0< \al<2$, $\displaystyle\lim_{\varepsilon\to 0}D_\al(\rho_\varepsilon\Vert\si)=D_\al(\rho\Vert\si)$; \item[ii)] for any $1/2\le \al\le \infty$,
$\displaystyle \lim_{\varepsilon\to 0 }\widetilde{D}_\al(\rho_\varepsilon\Vert\si)=\widetilde{D}_\al(\rho\Vert\si)$.
\end{enumerate}
\end{lemma}
\begin{proof}
Write $\id:\M\to\M$ as the identity map and define the normal UCP map
\[\Psi_\si:\M\to \M, \Psi_\si(x)=\si(x)1\pl.\]
It is clear that the adjoint $\Psi_\si^\dag(\rho)=\si$ for any state $\rho\in D(\M)$. Take the normal UCP map $\Psi_\varepsilon=(1-\varepsilon) \id+\varepsilon\Psi$. Then $\Psi_\varepsilon^\dag(\rho)=\rho\circ \Psi_\varepsilon=\rho_\varepsilon$ and $\Psi_\varepsilon^\dag(\si)=\si$. For i), using data processing inequality of $D_\al$
\begin{align*}\limsup_{\varepsilon\to 0}D_\al(\rho_\varepsilon\Vert\si)=&\limsup_{\varepsilon\to 0}D_\al(\Psi_\varepsilon^\dag(\rho)\Vert\Psi_\varepsilon^\dag(\si))\le
\limsup_{\varepsilon\to 0}D_\al(\rho_\varepsilon\Vert\si)\\ \le &D_\al(\rho\Vert\si)\le\liminf_{\varepsilon\to 0}D_\al(\rho_\varepsilon\Vert\si)\pl, \end{align*}
where the last inequality uses the lower semi-continuity \cite[Theorem 4.1]{Hiai18}. The argument for ii) and $\al>1$ is similar by using the data processing inequality and the lower semi-continuity of $\widetilde{D}_\al$ \cite[Proposition 3.7 \& Theorem 3.11]{Jenvcova18}. For $1/2\le\al<1$,
the lower semi-continuity can be replaced by \begin{align*}&\liminf_{\varepsilon\to \infty}\widetilde{D}_\al(\rho_\varepsilon\Vert\si)\ge \liminf_{\varepsilon\to \infty} \widetilde{D}_\al((1-\varepsilon)\rho\Vert\si)= \liminf_{\varepsilon\to \infty} \widetilde{D}_\al(\rho\Vert\si)+\al' \log(1-\varepsilon)=\widetilde{D}_\al(\rho\Vert\si)\pl. \qedhere\end{align*}
\end{proof}

\subsection{Data-processing inequality for optimized $f$-divergence}

We  now establish the data-processing inequality for the optimized $f$-divergence $\widetilde{Q}_f$. We start with the key case of restricting to a subalgebra.

\begin{lemma}\label{incl2}
Let $\M\subset B(H)$ be a von Neumann algebra, and let $\N\subset \M$ be a subalgebra. Let $\brho,\bsi\in H$ be two unit vectors, and let $\rho_\M,\si_\M$ (resp.~$\rho_\N,\si_\N$) be the corresponding normal states on $\M$ (resp.~$\N$). Then for an operator anti-monotone function $f:(0,\infty)\to \mathbb{R}$, the following inequality holds
\begin{align}
\widetilde{Q}_f(\rho_\M\|\si_\M)\ge \widetilde{Q}_f(\rho_\N\|\si_\N).
\end{align}
\end{lemma}

\begin{proof}
For two vectors $\bsi,\bome\in H$, we write $\Delta^{\M}(\bsi,\bome)$ (resp. $\Delta^{\N}(\bsi,\bome)$) as the relative modular operator with respect to $\M$ (resp. $\N$). Let  $S_{\bsi,\bome}^\M$ and $S_{\bsi,\bome}^\N$ be the corresponding anti-linear operators such that
\[ (S_{\bsi,\bome}^\M)^*\bar{S}_{\bsi,\bome}^\M=\Delta^\M(\bsi,\bome)\pl,
\qquad
(S_{\bsi,\bome}^\N)^*\bar{S}_{\bsi,\bome}^\N=\Delta^\N(\bsi,\bome)\pl.\]
Recall the support projections are given by
\begin{align*} s_{\M}(\bome)=[\M ' \bome]\pl,
\qquad s_{\N'}(\bome)=[\N \bome]\pl,
\end{align*}
By the definition of the $S$ operators, we find that
\begin{align*}
&S_{\bsi,\bome}^\M s_{\N'}(\bome)=s_{\M}(\bome)S_{\bsi,\bome}^\N\pl , \qquad  \Delta^{\N}(\bsi,\bome)\ge s_{\N'}(\bome)\Delta^{\M}(\bsi,\bome)s_{\N'}(\bome).
\end{align*}Then for
all $\bome$ such that $\norm{\bome}{2}=1$ and $\brho\in [\N \bome]=s_{\N'}(\bome)$, we find that
\begin{align*}
\bra{\brho}f(\Delta^{\N}(\bsi,\bome))\ket{\brho}
&\le \bra{\brho}f(s_{\N'}(\bome)\Delta^{\M}(\bsi,\bome)s_{\N'}(\bome))\ket{\brho}
\\&\le \bra{\brho}s_{\N'}(\bome)f(\Delta^{\M}(\bsi,\bome))s_{\N'}(\bome)\ket{\brho}
\\&= \bra{\brho}f(\Delta^{\M}(\bsi,\bome))\ket{\brho}.
\end{align*}
Here we view the projection $s_{\N'}(\bome)$ as an isometry on its support.
Noting that $\brho\in [\N \bome]\subset [\M \bome]$, then
\begin{align*}
\widetilde{Q}_f(\rho_\N\|\si_\N)& =\sup_{\bome \, :\, \norm{\bome}{2}=1, \brho\in[\N\bome]}\lan \brho| f(\Delta^\N(\si,\bome))| \brho\ran
\\ & \le\sup_{\bome \, :\, \norm{\bome}{2}=1, \brho\in[\N\bome]}\lan \brho| f(\Delta^\M(\si,\bome))| \brho\ran
\\ & \le\sup_{\bome \, :\, \norm{\bome}{2}=1, \brho\in[\M\bome]}\lan \brho| f(\Delta^\M(\si,\bome))| \brho\ran= \widetilde{Q}_f(\rho_\M\|\si_\M).
\end{align*}
This concludes the proof.
\end{proof}

\begin{lemma}\label{rest2}
Let $\M$ be a von Neumann algebra, and let $e\in \M$ be a projection.
Let $\rho,\si\in D(\M)$ be two normal states with support $s(\rho)\le s(\si)\le e$. Let $\si_e,\rho_e$ denote the corresponding normal states on $e\M e$. Then for all operator anti-monotone functions $f:(0,\infty)\to \R$, the following equality holds
\begin{align}
\widetilde{Q}_f(\rho\|\si)=\widetilde{Q}_f(\rho_e\|\si_e).
\end{align}

\end{lemma}
\begin{proof}We use the standard form $(\M,L_2(\M),J,L_2(\M)_+)$ from Appendix~\ref{a3}.
The standard form of $e\M e$ is $(e\M e,eL_2(\M)e,J,eL_2(\M)_+e)$. Let $V: eL_2(\M)e \hookrightarrow L_2(\M)$ be the isometry that is the adjoint of the projection $P:L_2(\M)\to eL_2(\M)e$ with $P(x)=exe$.
Let $h_\rho^{1/2}$ and $h_\si^{1/2}$ be the vectors in $L_2(\M)_+$ corresponding to $\rho$ and $\si$, respectively. Since $s(\rho)\le s(\si)\le e$, we have that $h_\rho^{1/2}=eh_\rho^{1/2}=eh_\rho^{1/2} e$ and similarly for $h_\si^{1/2}$.
Let $\omega\in D(\M)$ be a normal state, and let $h_\omega^{1/2}\in L_2(\M)_+$ be the corresponding unit vector. Let $\omega_e\in (e\M e)_+$ be the restriction of $\omega$ on $e\M e$. Note that $\omega_e$ is a sub-state corresponding to $eh_\omega e\in eL_1(\M )e\cong L_1(e\M e)$. By Proposition~\ref{dense}, it suffices to consider $\omega$ such that $\omega_e\neq 0$. Otherwise we can always replace $\omega$ by $\omega_\varepsilon=(1-\varepsilon)\omega+\varepsilon \rho$.

Recall that $\Delta_\M(\si,\omega)^{-1}=J\Delta_\M(\omega, \si)J$ and for $x\in \M$,
$\Delta(\omega, \si)^{1/2}JP\ket{ h_\si^{1/2} x}=\ket{h_\omega^{1/2}e x e}$. Then we find that
\begin{align*}
\bra{ h_\si^{1/2} x} P\Delta_\M(\si,\omega)^{-1}P \ket{ h_\si^{1/2} x}
& =\lan h_\omega^{1/2}e x e \ket{ h_\omega^{1/2}e x e}\\
& = \tr( ex^*eh_\omega e xe)=
\bra{ h_\si^{1/2} exe}\Delta_{e\M e}(\si,\omega_e)^{-1}\ket{ h_\si^{1/2} exe}.
\end{align*}
This implies that
\[
P\Delta_\M(\si,\omega)^{-1}P= \Delta_{e\M e}(\si,\omega_e)^{-1}.
\]
For $f:(0,\infty)\to \R$ operator anti-monotone, $\tilde{f}(x)=f(x^{-1})$ is operator monotone and operator concave. Since $h_\rho^{1/2}\in eL_2(\M)e=PL_2(\M)$,
\begin{align} \bra{h_\rho^{1/2}}f(\Delta_\M(\si,\omega))\ket{h_\rho^{1/2}}
& =\bra{h_\rho^{1/2}}P\tilde{f}(\Delta_\M(\si,\omega)^{-1})P\ket{h_\rho^{1/2}}
\nonumber\\ & \le  \bra{h_\rho^{1/2}}\tilde{f}(P\Delta_\M(\si,\omega)^{-1}P)\ket{h_\rho^{1/2}}
\nonumber\\ & =  \bra{h_\rho^{1/2}}\tilde{f}(\Delta_{e\M e}(\si,\omega_e)^{-1})\ket{h_\rho^{1/2}}
\nonumber\\ & \le  \bra{h_\rho^{1/2}}\tilde{f}(\Delta_{e\M e}(\si,\overline{\omega}_e)^{-1})\ket{h_\rho^{1/2}} \notag \\
& \le\widetilde{Q}_f(\rho_e\|\si_e)\label{ineq},
\end{align}
where $\overline{\omega}_e=\frac{\omega_e}{\omega_e(1)}$ is the normalized state of $\omega_e$.
By taking all $\omega\gg \rho$,
\[ \widetilde{Q}_f(\rho\|\si)\le \widetilde{Q}_f(\rho_e\|\si_e)\pl.\]
The reverse inequality follows from Lemma \ref{incl2} because $e\M e\subset \M$ as a (non-unital) subalgebra.
\end{proof}
\begin{rem}{\label{4.7} \rm The lemma above is an extension of isometric invariance \cite[Proposition 4]{wilde} in finite-dimensional case. It implies that it suffices to consider optimized $f$-divergence on $\si$-finite von Neumann algebras. Indeed, we can always restrict to $e\M e$ for $e=s(\rho+\si)$ because  $\widetilde{Q}_f(\rho\|\si)=\widetilde{Q}_f(\rho_e\|\si_e)$. Based on that, one can further deduce the following variant of Proposition~\ref{dense}:
 \[\widetilde{Q}_f(\rho\|\si)=\sup_{\omega\in D_+(\M)} \bra{h_\rho^{1/2}}f(\Delta(\si,\omega))\ket{h_\rho^{1/2}},\]
where $D_+(\M)$ is the set of all faithful normal states.
}
\end{rem}

\begin{theorem}[Data-processing inequality]
\label{data}
Let $\Phi:\N\to\M$ be a normal completely positive unital map, and let $\rho,\si\in D(\M)$ be two normal states. For $f:(0,\infty)\to \R$ operator anti-monotone, the following data-processing inequality holds
\[
\widetilde{Q}_f(\rho\|\si)\ge \widetilde{Q}_f(\rho\circ\Phi\|\si\circ \Phi).
\]
\end{theorem}

\begin{proof}Let $\M\subset B(H)$, and let $\brho,\bsi\in H$ be the vectors implementing $\rho,\si$, respectively.
Let $\Phi(\cdot)=V^*\pi(\cdot) V$ be the Stinespring dilation of $\Phi$, where $\pi:\N\to B(K)$ is a normal $*$-homomorphism and $V:H\to K$ is an isometry \cite{Stinespring55}. Let $\rho_1=\rho\circ\Phi$ and $\si_1=\si\circ\Phi$ denote states on $\N$. Then $\brho_1=V\brho$ and $\bsi_1=V\bsi$ are vector representations of $\rho_1$ and $\si_1$, respectively, via $\pi$ because
\begin{align*}&\rho\circ\Phi(x)=\rho(V^*\pi(x)V)=\bra{\brho}V^*\pi(x)V\ket{\brho}\pl, \\
&\si\circ\Phi(x)=\si(V^*\pi(x)V)=\bra{\bsi}V^*\pi(x)V\ket{\bsi}\pl.\end{align*}
Take the projection $e=VV^*\in B(H)$. Let  $\mathcal{L}\subset B(K)$ denote the von Neumann subalgebra in $B(K)$ generated by $V\M V^*$ and $\pi(\N)$. Note that $V:H\to eK$ is a surjective isometry and define the map $T:B(eK)\to B(H)$ as
\[ x\mapsto V^*xV\pl.\]
The map $T$ is a $*$-isomorphism that sends $e\mathcal{L} e$ to $\M$. Thus
we have the following factorization of $\Phi$:
\begin{align}\label{chain}
 \N\overset{\pi}{\longrightarrow} \pi(\N)\hookrightarrow \mathcal{L}\to e\mathcal{L} e \overset{T}{\longrightarrow}\M\pl.\end{align}
Let us introduce the shorthand $\widetilde{Q}^\M_f(\brho \| \bsi)\equiv\widetilde{Q}_f(\rho_\M\|\si_\M)$, where $\rho_\M,\si_\M$ are the states on~$\M$ implemented by the vectors $\brho,\bsi$.
Using this notation, we have
\begin{align*} \widetilde{Q}_f(\rho\circ\Phi\|\si\circ\Phi)
&=\widetilde{Q}_f^{\pi(\N)}(\brho_1\|\bsi_1)
\le  \widetilde{Q}_f^{\mathcal{L}}(\brho_1\|\bsi_1)
=  \widetilde{Q}_f^{e\mathcal{L} e}(\brho_1\|\bsi_1)
\\ &=  \widetilde{Q}_f^{\M}(\brho\|\bsi) =\widetilde{Q}_f(\rho\|\si).
\end{align*}
Here the first equality follows from the independence in Lemma \ref{ind2}. The inequality follows from the inclusion $\pi(\N)\subset \mathcal{L}$ and Lemma \ref{incl2}. The second equality follows because $\brho_1,\bsi_1\in eK$ and by applying Lemma \ref{rest2}. The last step is a $*$-isomorphism.
\end{proof}
It is clear from the argument above that the actual inequality in data processing is the inclusion $\pi(\N)\subset \mathcal{L}$.

\subsection{Recoverability results}

\label{section5}

In this section, we discuss recoverability results in the setting of general von Neumann algebras. We first review the generalized conditional expectation introduced in \cite{AC}, which is the (dual of) Petz map of the inclusion $\N\subset \M$ in the Heisenberg picture.

Let $\M$ be a von Neumann algebra, and let $\N\subset \M$ be a subalgebra.
We denote by $(\M,L_2(\M),J,L_2(\M)^+)$ (resp. $(\N,L_2(\N),J_0,L_2(\N)^+)$) the standard form of $\M$ (resp. $\N$) using Haagerup $L_2$-spaces.
Given a normal state $\rho\in D(\M)$ and its restriction $\rho_\N$ in $D(\N)$, we denote by $h_\rho$ (resp. $h_{\rho_\N}$)  the density operator of $\rho$ (resp. $\rho_\N$) in $L_1(\M)$ (resp. $L_1(\N)$). Thus $h_\rho^{1/2}\in L_2(\M)$  (resp. $h_{\rho_\N}^{1/2}\in L_2(\N)$) is a vector representation of $\rho$ (resp. $\rho_\N$).
Define the partial isometry $V_\rho: L_2(\N)\to L_2(\M)$ as
\[ V_\rho(a h_{\rho_\N}^{1/2}+\xi)=a h_{\rho}^{1/2}\pl, \forall a\in \N, \xi \in [\N h_{\rho_\N}^{1/2}]^\perp\pl. \]
Indeed,
\[ \norm{V_\rho(a h_{\rho_\N}^{1/2})}{L_2(\M)}^2=\norm{a h_{\rho}^{1/2}}{2}^2 = \tr(a^*a h_\rho)=\rho(a^*a)=\norm{a h_{\rho_\N}^{1/2}}{L_2(\M)}^2 . \]
The $\rho$-preserving generalized conditional expectation $E_\rho:\M\to \N$ is defined as follows:
\[ E_\rho(x)\coloneqq J_0V_\rho Jx JV_\rho J_0\pl.\]
Observe that $E_\rho:\M\to \N $ is a normal completely positive sub-unital map. Moreover $E_\rho(s(\rho))=s_\N(\rho)$  and $E_\rho(1-s(\rho))=0$ where $s(\rho)$ (resp. $s_\N(\rho)$) is the support of $\rho$ (resp. $\rho_\N$).
It was proved by Petz \cite{petz88} that if $D(\rho\|\si) < \infty$, then the equality $D(\rho\|\si)=D(\rho_\N\|\si_\N)$ is equivalent to the following conditions:
\begin{enumerate}
\item[i)]$E_\rho=E_\si$;
\item[ii)]$\rho_\N\circ E_\si=\rho$;
\item[iii)]$\si_\N\circ E_\rho=\si$.
\end{enumerate}
In this sense $E_\rho$ (or equivalently $E_\si$) is a recovery for the inclusion $\N\subset\M$.

In general, consider a normal completely positive unital map $\Phi:\N\to \M$. Let $\rho\in D(\M)$ be  a state, and set $\rho_0=\rho\circ \Phi\in D(\N)$. The Petz map $R:=R_{\Phi,\rho}:\M\to \N$ is the unique normal completely positive sub-unital map such that
\[ R(s(\rho))=s(\rho_0)\pl, \qquad R(1-s(\rho))=0\pl,\]
and $\forall\pl x\in \N , y\in\M\pl$,
\begin{align}
\label{recovery}
\lan Jy h_\rho^{1/2}, J_0\Phi(x) h_{\rho_{0}}^{1/2}\ran =\lan J R(y) h_\rho^{1/2}, J_0x h_{\rho_0}^{1/2}\ran\pl. \pl\end{align}
In particular, if $\rho_0=\rho\circ \Phi$ is faithful, then $R$ is unital.

Recall that the modular automorphism group $\al_{t}^{\rho}:\M \to \M$ for a state $\rho$ is given by
\[
\al_t^\rho(x)= \Delta(\rho,\rho)^{-it} x\Delta(\rho,\rho)^{it}\pl.
\]
The rotated Petz map is defined as follows:
\begin{align}
&E_\rho^t(x) \coloneqq \al_{t}^{\rho_\N}\circ E_\rho\circ \al_{-t}^\rho\pl, \qquad R_{\Phi,\rho}^t(x)=\al_{t}^{\rho_0}\circ R_{\Phi,\rho}\circ \al_{-t}^\rho\pl. \label{rotated}
\end{align}

Recall that in the Stinespring dilation $\Phi(\cdot)= V^*\pi(\cdot) V$, $\pi$ can be faithful (c.f. \cite[Theorem 1.41]{pisier20}). By the same argument in the proof of Theorem \ref{data},
it suffices to consider two cases:
\begin{enumerate}
\item[i)] For an inclusion $\iota:\N \to \M$, $R_{\iota,\rho}=E_\rho$ is the generalized conditional expectation
\item[ii)] Consider the projection map $$P:\M\to e\M e\pl,\pl P(x)=exe $$
    for a projection $e\in \M$ and let $\rho$ be a state with $s(\rho)\le e$.
     The recovery map $R_{P,\rho}=\iota_\rho: {s(\rho)\M s(\rho)}\to \M$ is the embedding
      and so is the rotated Petz map $R_{P,\rho}^t=\al^{\rho}_{t}\circ\iota_\rho\circ \al^{\rho}_{-t}=\iota_\rho$.
\end{enumerate}
Let $\Phi:\N\to \M$ be a general normal UCP map given by the composition $\Phi=P\circ \iota$.
Note that by the symmetric role of $\Phi$ and $R_{\Phi,\rho}$ in \eqref{recovery}, the Petz map $R_{\Phi,\rho}=R_{\iota,\rho}\circ R_{P,\rho}=E_\rho\circ \iota_\rho$ is a composition of the Petz map of the above two cases. Similarly for a rotated Petz map,
\[ R_{\Phi,\rho}^t=\al^{\rho_0}_t\circ R_{\Phi,\rho}\circ \al^{\rho}_{-t} =(\al^{\rho_0}_t\circ E_\rho\circ \al^{\rho}_{-t})\circ ( \al^{\rho}_{t}\circ \iota_\rho \circ \al^{\rho}_{-t} )=E_\rho^t\circ \iota_\rho\pl. \]
Since the embedding $\iota_\rho:s(\rho)\M s(\rho)\hookrightarrow \M$ always preserves the $L_1$-norm and (optimized) $f$-divergence on its support (Lemma \ref{rest2}), it suffices to consider the recovery result for $E_\rho^t$.

We now extend the recovery results in Section~\ref{sec:finite} to the general setting. For simplicity, we will mainly focus on faithful cases. The main steps that need adaptation are Lemmas~\ref{3.2}, \ref{3.10}, and \ref{3.19}, which we reproduce here using standard form on Haagerup $L_2$-spaces.

\begin{lemma}
Let $\rho$, $\si$, and $\omega$ be normal states, and let $\ket{\brho}=h_\rho^{1/2}\in L_2(\M)$ be the vector representation of $\rho$. Suppose $\ket{\brho}\in \text{supp}(\Delta(\si,\omega))=s(\si)s(\omega')$. Then for all $t \in \mathbb{R}$,
\[ \bra{\brho}\Delta(\si,\omega)^{-it}x\Delta(\si,\omega)^{it}\ket{\brho}=\rho\circ\al^\si_{t}(x).
\]
Thus $\Delta(\si,\omega)^{-it}x\Delta(\si,\omega)^{it}=\al^\si_{t}(x)$.
\end{lemma}

\begin{proof}
Let $h_\rho$, $h_\si$, and $h_\omega$ be the density operators of $\rho$, $\si$, and $\omega$, respectively. We have
\[\ket{\brho}=\ket{h_\rho^{1/2}}\pl, \qquad \Delta(\si,\omega)^{it}\ket{h_\rho^{1/2}}=\ket{h_\si^{it}h_\rho^{1/2}h_\omega^{-it}}\pl.\]
Then for $x\in \M$,
\begin{align*} \bra{\brho}\Delta(\si,\omega)^{-it}x\Delta(\si,\omega)^{it}\ket{\brho}&= \tr((h_\si^{it}h_\rho^{1/2}h_\omega^{-it})^*xh_\si^{it}h_\rho^{1/2}h_\omega^{-it})
\\ &= \tr((h_\omega^{it}h_\rho^{1/2}h_{\si}^{-it}xh_\si^{it}h_\rho^{1/2}h_\omega^{-it})
\\ &= \tr(h_\rho h_{\si}^{-it}xh_\si^{it})
\\ &= \tr(h_\rho \al_\si^t(x))
\\ &=\rho\circ \al_\si^t(x).\qedhere
\end{align*}
\end{proof}
\medskip

\begin{lemma}
\label{opin}
Let $\rho\in D_+(\M)$ and $\omega_\N\in D_+(\N)$ be faithful. Then \[V_\rho^*\Delta_\M(\si, E_\rho^\dag(\omega_\N))V_\rho= \Delta_\N(\si_\N,\omega_\N)\pl.\]
As a consequence, for all operator anti-monotone functions $f:(0,\infty)\to \R$,
\[\bra{h_{\brho_\N}^{1/2}} f(\Delta_\N(\si_\N,\omega_\N))\ket{h_{\brho_\N}^{1/2}} \le \bra{h_{\brho}^{1/2}} f(\Delta_\M(\si,R_\rho(\omega_\N)))\ket{h_{\brho}^{1/2}}.\]
\end{lemma}

\begin{proof}
 Let $h_{\rho_\N},h_{\omega_\N},h_{\si_\N}$ and $h_{\rho},h_{\omega},h_{\si}$ be the corresponding density operators.
Let $S_{\si_\N,\omega_\N}:L_2(\N)\to L_2(\N)$ (resp. $S_{\si,\omega}$) be the anti-linear operator for the standard form of~$\N$ (resp. $\M$).
We have for $a\in \N$,
\[V_\si S_{\si_\N,\omega_\N} (a h_{\omega_\N}^{1/2})=V_\si(a^* h_{\si_\N}^{1/2})=a^*h_{\si}^{1/2}\pl.\]
On the other hand, for any $a,b\in \N$, 
\[ V_\rho(a b h_{\rho}^{1/2})=ab h_{\rho}^{1/2}=aV_\rho( b h_{\rho}^{1/2})\pl.\]
By the density of $\N h_{\rho_\N}^{1/2}$ in $L_2(\N)$, this implies $aV_\rho=V_\rho a$. 
Then if we choose the $L_2$ vector $\bome=V_\rho h_{\omega_\N}^{1/2}$,
\[S_{\si,\bome} V_\rho(ah_{\omega_\N}^{1/2})=S_{\si,\bome} (a\bome) =a^*h_\si^{1/2}\pl.\]
 Note that $\Delta_\M(\si, \bome)$ only depends on $\omega'\in \M'$ induced by $\bome$. Indeed, for $x\in \M$,
\begin{align*}
\bra{\bome}JxJ\ket{\bome}
&=\bra{h_{\omega_\N}^{1/2}} V_\rho^* JxJ V_\rho \ket{h_{\omega_\N}^{1/2}}
\\ &=\bra{h_{\omega_\N}^{1/2}}J_0 V_\rho^* JxJ V_\rho  J_0\ket{h_{\omega_\N}^{1/2}}
\\ &=\omega_\N\circ E_\rho( x)
\\ &= \tr( xh_{\omega_\N\circ E_\rho}).
\end{align*}
Thus $S_{\si,\bome}^*\bar{S}_{\si,\bome}=\Delta_\M(\si, \omega)$ for $\omega=\omega_\N\circ E_\rho$. Thus for this choice $\bome=V_\rho h_{\omega_\N}^{1/2}$,
\[ S_{\si,\bome} V_\rho=V_\si S_{\si_\N,\omega_\N}\pl, \pl V_\rho^*\Delta_\M(\si,\omega)V_\rho=\Delta_\N(\si_\N, \omega_\N)\pl.\]
The other assertion follows from operator convexity and operator monotonicity of $f$.
\end{proof}

\begin{lemma}\label{4.11}Let $\rho,\si\in D_+(\M)$ and $\omega_\N\in D_+(\N)$ be faithful.
Define the vectors
\begin{align*}
\ket{a_t} & :=J\Delta(\si,\rho)^{-it}V_\rho\Delta(\si_\N,\rho_\N)^{\frac{1}{2}+it}\ket{h_{\rho_\N}^{1/2}},
\\
\ket{b_t} & := \Delta(\si,\rho)^{\frac{1}{2}+it}V_\rho\Delta(\si_\N,\rho_\N)^{-\frac{1}{2}-it}\ket{h_{\rho_\N}^{1/2}},
\\ \ket{c_t} & :=\Delta(\si,E_\rho(\omega_\N))^{1/2+it}V_\rho \Delta(\si_\N,\omega_\N)^{-1/2-it}\ket{h_{\rho_\N}^{1/2}}.
\end{align*}
The following equalities hold for $x\in \M$:
\begin{align*}
\bra{a_t}x\ket{a_t}=\si_\N\circ E_\rho^{t} ( x)\pl, \quad \bra{b_t}x\ket{b_t}=\rho_\N\circ E_\si^{-t} ( x)\pl, \quad
\bra{c_t}x\ket{c_t}=\rho_\N\circ E_\si^{-t} ( x)\pl.
\end{align*}
\end{lemma}
\begin{proof}
For the first one,
\begin{align*}
\ket{a_t}& =J\Delta(\si,\rho)^{-it}V_\rho\Delta(\si_\N,\rho_\N)^{\frac{1}{2}+it}\ket{h_{\rho_\N}^{1/2}}
\\ & = J\Delta(\si,\rho)^{-it}J J V_\rho J_0 J_0\Delta(\si_\N,\rho_\N)^{\frac{1}{2}+it}J_0\ket{h_{\rho_\N}^{1/2}}
\\ & = \Delta(\rho,\si)^{-it} J V_\rho J_0 \Delta(\rho_\N,\si_\N)^{it} \ket{h_{\si_\N}^{1/2}}.
\end{align*}
Then
\begin{align*}\bra{a_t}x\ket{a_t}&=\bra{h_{\si_\N}^{1/2}} \Delta(\rho_\N,\si_\N)^{-it}J_0 V_\rho^* J \Delta(\rho,\si)^{it}x\Delta(\rho,\si)^{-it} J V_\rho J_0 \Delta(\rho_\N,\si_\N)^{it} \ket{h_{\si_\N}^{1/2}}
\\ &=\bra{h_{\si_\N}^{1/2}} \Delta(\rho_\N,\si_\N)^{-it}J_0 V_\rho^* J \al^\rho_{-t}(x) J V_\rho J_0 \Delta(\rho_\N,\si_\N)^{it} \ket{h_{\si_\N}^{1/2}}
\\ &=\bra{h_{\si_\N}^{1/2}} \Delta(\rho_\N,\si_\N)^{-it} E_\rho\circ\al^\rho_{-t}(x) \Delta(\rho_\N,\si_\N)^{it} \ket{h_{\si_\N}^{1/2}}
\\ &=\bra{h_{\si_\N}^{1/2}} (\al^{\rho}_{t}\circ E_\rho\circ\al^{\rho_\N}_{-t})(x) \ket{h_{\si_\N}^{1/2}}
\\ &=\bra{h_{\si_\N}^{1/2}} E_\rho^{t}(x) \ket{h_{\si_\N}^{1/2}}
\\ &= \si_\N\circ E_\rho^{t}(x).
\end{align*}
For the second one, we first show that
\[ \Delta(\si,\rho)^{\frac{1}{2}}V_\rho\Delta(\si_\N,\rho_\N)^{-\frac{1}{2}}= JV_\si J_0.
\]
Indeed, for $a\in \M$
\begin{align*}  JV_\si J_0\ket{h_{\si_\N}^{1/2} a}& =JV_\si \ket{a^*h_{\si_\N}^{1/2} }=J \ket{a^*h_{\si}^{1/2} }=\ket{h_{\si}^{1/2} a}\\
 \Delta(\si,\rho)^{\frac{1}{2}}V_\rho\Delta(\si_\N,\rho_\N)^{-\frac{1}{2}}\ket{h_{\si_\N}^{1/2} a}
& = \Delta(\si,\rho)^{\frac{1}{2}}V_\rho J\Delta(\rho_\N,\si_\N)^{\frac{1}{2}}J\ket{h_{\si_\N}^{1/2} a}
\\ & = \Delta(\si,\rho)^{\frac{1}{2}}V_\rho J\Delta(\rho_\N,\si_\N)^{\frac{1}{2}}\ket{a^*h_{\si_\N}^{1/2}}
\\ & = \Delta(\si,\rho)^{\frac{1}{2}}V_\rho J\ket{h_{\rho_\N}^{1/2}a^*}
\\ & = \Delta(\si,\rho)^{\frac{1}{2}}V_\rho \ket{a h_{\rho_\N}^{1/2}}
\\ & = \Delta(\si,\rho)^{\frac{1}{2}} \ket{a h_{\rho}^{1/2}}
\\ & = \ket{h_{\si}^{1/2} a }.
\end{align*}
Then for $x\in \M$,
\begin{align*}\bra{b_t}x\ket{b_t}&=\bra{h_{\rho_\N}^{1/2}}\Delta(\si_\N,\rho_\N)^{-\frac{1}{2}+it}V_\rho^*\Delta(\si,\rho)^{\frac{1}{2}-it}x\Delta(\si,\rho)^{\frac{1}{2}+it}V_\rho\Delta(\si_\N,\rho_\N)^{-\frac{1}{2}-it}\ket{h_{\rho_\N}^{1/2}}
\\ &=\bra{h_{\rho_\N}^{1/2}}\Delta(\si_\N,\rho_\N)^{it}J_0V_\si^*J\Delta(\si,\rho)^{-it}x\Delta(\si,\rho)^{it}JV_\si J_0\Delta(\si_\N,\rho_\N)^{-it}\ket{h_{\rho_\N}^{1/2}}
\\ &=\bra{h_{\rho_\N}^{1/2}}\al^{\si_\N}_{-t}\circ E_\si\circ \al^{\si}_t(x)\ket{h_{\rho_\N}^{1/2}}
\\ &=\bra{h_{\rho_\N}^{1/2}}E_\si^{-t}(x)\ket{h_{\rho_\N}^{1/2}}
\\ &=\rho_\N\circ E_\si^{-t}(x).
\end{align*}
For the third assertion, note that we have shown in Lemma~\ref{opin} that
\begin{align*} \bra{h_{\omega_\N}^{1/2}}V_\rho^* Jx J V_\rho \ket{h_{\omega_\N}^{1/2}}
=\tr( x h_{\omega_\N\circ E_\rho})=\bra{h_{\omega_\N\circ E_\rho}^{1/2}}JxJ\ket{ h_{\omega_\N\circ E_\rho}^{1/2}}\pl.
\end{align*}
Then we have
\[ V_\rho h_{\omega_\N}^{1/2}= uh_{\omega_\N\circ E_\rho}^{1/2}\pl.\]
for some unitary $u$ in $\M$. For  ease of notation, we write  $\Delta_\N=\Delta(\si_\N,\omega_\N)$
and $\Delta_\M=\Delta(\si,\omega_\N\circ E_\rho)$.
 Then for $a\in \M$
\begin{align*}
\Delta_{\M}^{\frac{1}{2}}V_\rho\Delta_\N^{-\frac{1}{2}}\ket{h_{\si_\N}^{1/2} a}& =\Delta_{\M}^{\frac{1}{2}}V_\rho J\Delta_\N^{\frac{1}{2}}\ket{a^*h_{\si_\N}^{1/2}}= \Delta_{\M}^{\frac{1}{2}}V_\rho \ket{a h_{\omega_\N}^{1/2}}
\\ & = \Delta_{\M}^{\frac{1}{2}} a V_\rho \ket{ h_{\omega_\N}^{1/2}}
=\Delta_{\M}^{\frac{1}{2}} \ket{a u h_{\omega_\N\circ E_\rho}^{1/2}}
= \ket{ h_{\si}^{1/2}a u }= Ju^*V_\si J_0\ket{ h_{\si_\N}^{1/2}a }.
\end{align*}
Thus we have shown that
\[ Ju^*V_\si J_0=\Delta_{\M}^{\frac{1}{2}}V_\rho\Delta_\N^{-\frac{1}{2}},
\]
where $u$ is the unitary from the polar decomposition of $V_\rho h_{\omega_\N}^{1/2}$.
Then
\begin{align*}\bra{c_t}x\ket{c_t}& =\bra{h_{\rho_\N}^{1/2}} \Delta_{\N}^{-1/2+it}V_\rho^* \Delta_{\M}^{1/2-it} x\Delta_{\M}^{1/2+it}V_\rho \Delta_{\N}^{-1/2-it}\ket{h_{\rho_\N}^{1/2}}
\\ & =\bra{h_{\rho_\N}^{1/2}} \Delta_{\N}^{it}J_0V_\si^* J JuJ \Delta_{\M}^{-it} x\Delta_{\M}^{+it}
Ju^*J JV_\si J_0 \Delta_{\N}^{-it}\ket{h_{\rho_\N}^{1/2}}
\\ & =\bra{h_{\rho_\N}^{1/2}} \Delta_{\N}^{it}J_0V_\si^* J JuJ \al^{\si}_t(x)
Ju^*J JV_\si J_0 \Delta_{\N}^{-it}\ket{h_{\rho_\N}^{1/2}}
\\ & =\bra{h_{\rho_\N}^{1/2}} \Delta_{\N}^{it}J_0V_\si^* J \al^{\si}_t(x)
JV_\si J_0 \Delta_{\N}^{-it}\ket{h_{\rho_\N}^{1/2}}
\\ & =\bra{h_{\rho_\N}^{1/2}} \Delta_{\N}^{it} E_\si\circ\al^{\si}_t(x)
 \Delta_{\N}^{-it}\ket{h_{\rho_\N}^{1/2}}
\\ & =\bra{h_{\rho_\N}^{1/2}}\al_{-t}^{\si_\N} \circ E_\si\circ\al^{\si}_t(x)
 \ket{h_{\rho_\N}^{1/2}}
 \\ & =\bra{h_{\rho_\N}^{1/2}}E_{-\si}^t(x)
 \ket{h_{\rho_\N}^{1/2}}
 \\ & =\rho_\N\circ E_\si^{-t}(x).\qedhere
\end{align*}
\end{proof}
\bigskip
Now we can recover the estimate in Lemma \ref{3.2}, \ref{3.10}, and \ref{3.19}.
\begin{lemma}\label{5.14}
Define
\begin{align*}
\ket{w_\la} & \coloneqq (\Delta(\si,\rho)+\la)^{-1}\ket{h_{\rho}^{1/2}}-V_\rho(\Delta(\si_\N,\rho_\N)+\la)^{-1}\ket{h_{\rho_\N}^{\frac12}}\pl,
\\ \ket{u_\la} & \coloneqq (\Delta_\M(\si,R_\rho(\omega_\N))+\la)^{-1}\ket{h_{\rho}^{1/2}}-V_\rho (\Delta_\N(\si_\N,\omega_\N)+\la)^{-1}\ket{h_{\rho_\N}^{1/2}}\pl,\end{align*}
and
\begin{align*}
\ket{w_t} & \coloneqq -\frac{\cosh(\pi t)}{\pi}\left(\int^{\infty}_{0}\la^{1/2+it}\ket{w_\la}\, d\la\right),
\\ \ket{v_t} & \coloneqq \frac{\cosh(\pi t)}{\pi}\Delta_{\M}^{\frac{1}{2}+it}\int_0^\infty \la^{-\frac{1}{2}-it}\ket{w_\la}\, d\la,
\\ \ket{u_t} & \coloneqq \frac{\cosh(\pi t)}{\pi} \Delta_\M(\si,R_\rho(\omega_\N))^{1/2+it} \int_{0}^\infty \la^{-1/2-it} \ket{u_\la}\, d\la.
\end{align*}
Then the following inequalities hold
\begin{align*}
\norm{\si-\si_\N\circ E_\rho^{t}}{1} & \le 2\norm{\ket{w_t}}{2},
\\
\norm{\rho-\rho_\N\circ E_\si^{-t}}{1} & \le 2\norm{\ket{v_t}}{2},
\\ \norm{\rho-\rho_\N\circ E_\si^{-t}}{1} & \le 2\norm{\ket{u_t}}{2}.
\end{align*}
\end{lemma}
\begin{proof}
For ease of notation, we write $\Delta_{\M}:=\Delta(\rho,\si)$ and $\Delta_{\N}:=\Delta(\rho_\N,\si_\N)$.
As in the finite-dimensional case,
\begin{align*}
\ket{w_t}& =\Delta_{\M}^{1/2+it}\ket{h_{\rho}^{1/2}}-V_\rho\Delta_{\N}^{1/2+it}\ket{h_{\rho_\N}^{1/2}}=\Delta_{\M}^{it}\Big( \Delta_{\M}^{1/2}\ket{h_{\rho}^{1/2}}-\Delta_{\M}^{-it}V_\rho\Delta_{\N}^{1/2+it}\ket{h_{\rho_\N}^{1/2}}\Big)\\& =\Delta_{\M}^{it}\Big( \ket{h_{\si}^{1/2}}-J\ket{a_t}\Big),
\\ \ket{v_t}& =\ket{h_{\rho}^{1/2}}-\Delta_{\M}^{\frac{1}{2}+it} V_\rho\Delta_{\N}^{-1/2-it}\ket{h_{\rho_\N}^{1/2}}
=\ket{h_{\rho}^{1/2}}-\ket{b_t},
\\ \ket{u_t}& =\ket{h_{\rho}^{1/2}}-\Delta(\si,\omega_\N\circ E_\rho)^{\frac{1}{2}+it} V_\rho\Delta(\si_\N,\omega_\N )^{-\frac{1}{2}-it} \ket{h_{\rho_\N}^{1/2}}=\ket{h_{\rho}^{1/2}}-\ket{c_t}.
\end{align*}
For the first one, by $J\ket{h_{\si}^{1/2}}=\ket{h_{\si}^{1/2}}$ and Lemma \ref{4.11}
\[2\norm{\ket{w_t}}{2}=2\norm{J\ket{h_{\si}^{1/2}}-J\ket{a_t}}{2}=2\norm{\ket{h_{\si}^{1/2}}-\ket{a_t}}{2}\ge \norm{\si-\si_\N\circ E_\rho^t}{1}\pl,\]
where we have used the inequality in \eqref{2}. This inequality remains valid in Haagerup $L_p$-spaces since its proof in \cite[Lemma 2.2]{CV17} only uses H\"{o}lder's inequality. The other two assertions follow similarly.
\end{proof}

Based on the lemma above,
the rest of the argument is identical to that given for Lemmas~\ref{fd}, \ref{fd2}, and \ref{fd3}, which  estimate the Hilbert-space norm of $\ket{v_t}, \ket{w_t}$, and $\ket{u_t}$, respectively. In particular, the argument of Lemmas~\ref{fd}, \ref{fd2}, and \ref{fd3} implies the integral expression of $\ket{v_t}, \ket{w_t}$, and $\ket{u_t}$ converges absolutely if $Q_f(\rho\Vert\si)$ and $\widetilde{Q}_f(\rho\Vert\si)$ are finite for some regular $f$.

We now state our recovery results for quantum channels on general von Neumann algebras. Recall that we denote $D$ as the standard relative entropy and $\widetilde{D}_\al$ as the $\al$-sandwiched R\'enyi relative entropy. The maps $R_{\Phi,\si}^t$ and $R_{\Phi,\rho}^t$ are the rotated Petz maps defined in~\eqref{rotated}.

\begin{theorem}
\label{thm:recoverability-vNa}
Let $\Phi:\N\to \M$ be a normal unital completely positive map. Let $\rho,\si\in D(\M)$ be two states and denote $\rho_0=\rho\circ \Phi, \si_0=\si\circ \Phi$. Suppose $s(\rho)\le s(\si)$.
For $t\in \R$,
\begin{enumerate}
\item[i)] if $s(\rho)=s(\si)$ and $Q_{x^2}(\rho\|\si)<\infty$,
\[D(\rho\|\si)-D(\rho_0\|\si_0)\ge
\left(\frac{\pi}{8 \cosh(\pi t)} \right)^4 Q_{x^2}(\rho\|\si)^{-1}\norm{\si-\si_0\circ R_{\Phi,\rho}^t}{1}^4\pl.\]
\item[ii)] if $Q_{x^{-1}}(\rho\|\si)<\infty$, then for all $\varepsilon \in (0,1/2)$,
\[D(\rho\|\si)-D(\rho_0\|\si_0)\ge \left(K(Q_{x^{-1}}(\rho\|\si),\varepsilon) \frac{\pi}{2\cosh(\pi t)} \norm{\rho-\rho_0\circ R_{\Phi,\si}^t}{1}\right)^{\frac{1}{1/2-\varepsilon}}.\]
\item[iii)] if $\widetilde{Q}_{\infty}(\rho\|\si)=\inf \{\la| \rho\le \la\si\}<\infty$, then for all $\al\in(1/2,1)$, $\alpha'=\alpha/(\alpha-1)$, and $\varepsilon\in(0,(1-1/|\al'|)/2)$,
\begin{multline*}
\widetilde{D}_\al(\rho\|\si)-\widetilde{D}_\al(\rho_0\|\si_0) \ge \\
|\al'|\log\left(1+\left(K(\al ,\varepsilon ,\widetilde{Q}_{\infty}(\rho\|\si))\frac{\pi}{2\cosh{\pi t}}\norm{\rho-\rho_0\circ R_{\Phi,\si}^t}{1}\right)^{\frac{1}{\frac{1-1/|\al'|}{2}}-\varepsilon}\right).
\end{multline*}
For all $\al>1 $, $\alpha'=\alpha/(\alpha-1)$, and $\varepsilon\in(0,(1-1/|\al'|)/2)$,
\begin{multline*}
\widetilde{D}_\al(\rho\|\si)-\widetilde{D}_\al(\rho_0\|\si_0) \ge \\
\al' \log\left(1+
\frac{1}{\widetilde{Q}_{\infty}(\rho\|\si)^{\frac{1}{\al'}}}
\left(K(\al ,\varepsilon ,\widetilde{Q}_{\infty}(\rho\|\si))\frac{\pi}{2\cosh{\pi t}}\norm{\rho-\rho_0\circ R_{\Phi,\si}^t}{1}\right)^{\frac{1}{\frac{1-1/|\al'|}{2}}-\varepsilon}\right).
\end{multline*}
\end{enumerate}
In the inequalities above, $K(Q_{x^{-1}}(\rho\|\si),\varepsilon)$ and $K(\al ,\varepsilon ,\widetilde{Q}_{\infty}(\rho\|\si))$) are constants defined as in \eqref{eq:K-constant-rel-ent-orig} and \eqref{con}, respectively.
\end{theorem}
\begin{proof}
Note that the assumption $s(\rho)=s(\si)$ is equivalent to $\rho$ and $\si$ being faithful because we can always restrict our considerations to $s(\si)\M s(\si)$, as mentioned in Remark~\ref{4.7}. The faithfulness is needed for Lemma~\ref{opin} where we used the identity $V_\rho^* \Delta(\si,\rho )V_\rho=\Delta(\si_\N,\rho_\N)$ for the estimates in i). For ii) and iii),  we first obtain the faithful cases by the Lemma \ref{opin} and \ref{4.11}. For the general case of $s(\rho)\le s(\si)$, we use the approximation in
Lemma \ref{approximation}. Indeed, take $
\rho_\varepsilon=(1-\varepsilon)\rho+\varepsilon\si$ and $\rho_{0,\varepsilon}=(1-\varepsilon)\rho_0+\varepsilon\si_0$.
Then $s(\rho_\varepsilon)=s(\si)$, $\rho_{0,\varepsilon}=\rho_{\varepsilon}\circ \Phi$ and moreover
\[\lim_{\varepsilon\to 0}D(\rho_\varepsilon\Vert\si)=D(\rho\Vert\si)\pl, \lim_{\varepsilon\to 0}D(\rho_{0,\varepsilon}\Vert\si)=D(\rho_{0}\Vert\si_{0})\pl.\]
Then the estimate follows the faithful cases and $\displaystyle \norm{\rho-\rho_{0}\circ R_{\Phi,\si}^t}{1}=\lim_{\varepsilon\to 0}\norm{\rho_\varepsilon-\rho_{0,\varepsilon}\circ R_{\Phi,\si}^t}{1}$. The argument for iii) is similar.
\end{proof}

We have the following corollary regarding reversibility:

\begin{cor}
\label{cor:reversibility-vNa}
Let $\Phi:\N\to \M$ be a normal unital completely positive map.  Let $\rho,\si\in D(\M)$ be two states, and let $\rho_0\coloneqq\rho\circ \Phi$ and $ \si_0\coloneqq\si\circ \Phi$. Suppose $s(\rho)\le s(\si)$ and $D(\rho\|\si)<\infty$.
The following are equivalent:
\begin{enumerate}
\item[i)] $D(\rho\|\si)=D(\rho_0\|\si_0)<\infty$.
\item[ii)]$\widetilde{D}_\al(\rho\|\si)=\widetilde{D}_\al(\rho_0\|\si_0)<\infty$ for some $\al\in(1/2,1)\cup (1,\infty)$.
\item[iii)]$\widetilde{Q}_f(\rho\|\si)=\widetilde{Q}_f(\rho_0\|\si_0)$ for some  regular operator anti-monotone function $f$.
\item[iv)] $\widetilde{Q}_f(\rho\|\si)=\widetilde{Q}_f(\rho_0\|\si_0)$ for all operator anti-monotone functions $f$.
\item[v)]there exists a normal UCP map $\Phi:\M\to \N$ such that $\rho_0\circ \Phi=\rho$ and $\si_0\circ \Phi=\si$.
\item[vi)]$\rho_0\circ R_{\Phi,\si}^{t}=\rho$ for all $t$.
    \item[vii)]$\si_0\circ R_{\Phi,\rho}^{t}=\si$ for all $t$.
\end{enumerate}
\end{cor}

\begin{proof}We argue for the subalgebra case $\Phi=\iota:\N\hookrightarrow\M$.
vi), vii)$\Rightarrow$ v) is trivial. Note that by the monotonicity $\al\mapsto \widetilde{D}_\al$, $D(\rho\Vert\si)<
\infty$ implies $\widetilde{D}_\al(\rho\Vert\si)<\infty$. Then v)$\Rightarrow$ i)-iv) by data processing inequality. Also iv) $\Rightarrow$ iii) is trivial and iii) $\Rightarrow$ i) by Example \ref{5.5}. i) implies $\ket{w_\la}=0$ for $\la>0$, which by Lemma \ref{5.14} further implies vi) and vii).  ii)$\Rightarrow $ v) is proved in \cite[Theorem 5.1]{Jenvcova2}.
\end{proof}

\section{Conclusion}

In summary, we have established physically meaningful remainder terms for the data-processing inequality for the optimized $f$-divergence, and we have improved upon prior  results like this for the standard $f$-divergence. As a consequence, we have established the first physically meaningful remainder terms for the data-processing inequality for the sandwiched R\'enyi relative entropy. Finally, we generalized all of our results to the von Neumann algebraic setting of the optimized $f$-divergence, by suitably generalizing its definition, its data-processing inequality, and refinements to this setting.

Going foward from here, we consider it to be a great challenge to establish universal remainder terms for the data-processing inequalities of the standard and optimized $f$-divergences, in the sense of \cite{universal}. Such results would significantly extend the domain of applicability of these refined data-processing inequalities.

\textbf{Acknowledgements.} We thank Anna Vershynina for feedback on our paper. This project was initiated during the workshop ``Algebraic and Statistical Ways into Quantum Resource Theories,'' which took place during the summer of 2019 at the Banff International Research Station (BIRS). The authors thank the Banff International Research Station for Mathematical Innovation and Discovery for the kind hospitality during the workshop. MMW acknowledges support from the National Science Foundation under Grant No.~1714215,
from Stanford QFARM, and from AFOSR under grant number FA9550-19-1-03.

\bibliography{fdiv}
\bibliographystyle{alpha}

\appendix

\section{Preliminaries on von Neumann algebras}

\label{app:vNa-basics}

In this appendix, we briefly review some of the von Neumann algebra theory used in Section~\ref{sec:opt-f-div}. We refer to the classic texts \cite{takesaki,takesaki2} for more information on von Neumann algebras and to \cite{petzbook} 
for a similar introduction related to quantum divergences.

\subsection{Spatial derivative and relative modular operator}

Let $\M\subset B(H)$ be a von Neumann algebra. A linear functional $\phi:\M\to \mathbb{C}$ is
\begin{itemize}\item[i)]{\bf normal} if it is weak$^*$-continuous;
\item[ii)]{\bf positive} if $\phi(x^*x)\ge 0\pl, \pl \forall x\in \M $;
\item[iii)]{\bf unital} if $\phi(1)=1$;
\item[iv)] a {\bf state} if $\phi$ is positive and unital.
\end{itemize}
The predual $\M_*$ of $\M$ is the space of all normal linear functionals. We denote by $\M_*^+$ the set of all normal positive linear functionals and by $D(\M)$ the set of all normal states. A positive normal linear functional $\phi$ is {\bf faithful} if $\phi(x^*x)=0$ implies $x=0$. A von Neumann algebra is $\si$-finite if it admits a normal faithful state. For $\phi\in \M_*^+$, its support  $s(\phi)$ is the smallest projection $e\in \M$ such that $\phi(e)=\phi(1)$.
We say that $\pi:\M\to B(H)$ is a $*$-representation if $\pi$ is a normal $*$-homomorphism (not necessarily unital). We say that the vector $\bphi\in H$ implements $\phi\in \M_*^+$ via $\pi$ if for all $x\in \M$,
\[\phi(x)=\bra{\bphi}\pi(x)\ket{\bphi}\pl.\]
We typically use Greek letters $\rho,\si,\phi,\psi$ to denote states and linear functionals, and boldface letters $\brho,\bsi,\bphi,\bpsi$ to denote vectors implementing the corresponding states. Let $G_\phi$ be the Hilbert space completion of $\M$ with respect to the $\phi$-inner product:
\[\lan x,y \ran_\phi=\phi(x^*y)\pl.\]
Let $\bet_\phi(x)$ (resp.~$\bet_\phi$) be the vector corresponding to $x\in \M$ (resp.~identity $1$).  The GNS representation $\pi_\phi:\M \to B(G_\phi)$ is the normal $*$-homomorphism given by
\[\pi_\si(a)\bet_\phi(x)=\bet_\phi(ax)\pl.\]
In particular, $\bet_\phi$ implements $\phi$ via $\pi_\phi$. Letting $\bphi\in H$ be a vector implementing $\phi$ via $\pi:\M\to B(H)$, we can define the isometry $V: G_\phi \to H$ as follows:
\[ V (\pi_\phi(x)\bet_\phi)=\pi(x)\bphi\pl.\] We denote $[\pi(\M)\bphi]$ as the closure of $\pi(\M)\bphi \subset H$ as a subspace, and with slight abuse of notation, also identify it as the projection onto $[\pi(\M)\bphi]$. Thus $G_\phi \cong [\pi(\M)\bphi]$ for all $\bphi$ implementing $\phi$.

Let $\M\subset B(H)$ be a von Neumann algebra acting on $H$, and let
\[
\M' \coloneqq \{x\in B(H)\pl |\pl xa=ax \quad \forall  a\in \M\}
\]
be its commutant. For a vector $\bphi\in H$, we denote by $\phi\in \M_*$ and $\phi'\in (\M')_*$ the corresponding states implemented on $\M$ and $\M'$. The support projections are given by
\[s_{\M}(\bphi):=s(\phi)=[\M'\bphi]\in \M\pl,
\qquad s_{{\M}'}(\bphi):=s(\phi')=[\M\bphi]\in \M'.
\]
Given two vectors $\bphi,\bpsi\in H$, we define the anti-linear operator $S_{\bpsi,\bphi}$ as follows:
\begin{equation}
\label{eq:S-op-def}
S_{\bpsi,\bphi}(a\bphi+\bet)=s(\bphi)a^*\bpsi\pl. \pl a\in \M\pl,
\end{equation}
where $a\bphi\in [\M\bphi], \bet\in [\M\bphi]^\perp$.
Then $S_{\bpsi,\bphi}$ is a closable operator, and the relative modular operator is the positive self-adjoint operator defined as
\begin{equation}
\label{eq:rel-mod-op-def}
\Delta(\bpsi,\bphi) \coloneqq (S_{\bpsi,\bphi})^*\bar{S}_{\bpsi,\bphi}\pl,
\end{equation}
where $\bar{S}_{\bpsi,\bphi}$ is the closure of $S_{\bpsi,\bphi}$.
For $a\bphi\in \M\bphi$,
\begin{align} \bra{a\bphi}\Delta(\bpsi,\bphi)\ket{a\bphi}=\bra{\bpsi}a s(\bphi)a^*\ket{\bpsi}\pl. \label{modular}\end{align}
and the support $\text{supp}(\Delta(\bpsi,\bphi))=s(\psi)s(\phi')$.

We also recall the spatial derivative. Given $\phi\in \M_*^+$, define the lineal of $\phi$ as the subspace
\[ H_\phi=\{ \bxi\in H\pl|\pl \norm{a\bxi}{H}^2\le C\phi(a^*a)\pl \forall \pl a\in \M\pl, \pl \text{for some } C\ge 0 \}\pl.\]
The closure $\overline{H_\phi}=s(\phi)H$. For $\bxi\in H_\phi$, we define the bounded operator $R_\phi(\bxi): G_\phi\to H$ as follows:
\begin{align}\label{R}R_\phi(\bxi)\bet_\phi(x)=x \bxi\pl.\end{align}
Then $R_\phi(\bxi)\pi_\phi(a)=a R_\phi(\bxi)$, which implies $R_\phi(\bxi)R_\phi(\bxi)^*\in \M'$. For a vector $\bpsi\in H$, the spatial derivative $\Delta(\bpsi/\phi)$ is the positive self-adjoint operator on $H_\phi$ defined by
\[\bra{\bxi}\Delta(\bpsi/\phi)\ket{\bxi} \coloneqq \bra{\bpsi}R_\phi(\bxi)R_\phi(\bxi)^*\ket{\bpsi}\pl.\]
We can write $\Delta(\bpsi/\phi)=\Delta(\psi'/\phi)$ because it only depends on $\psi'\in (\M')_*^+$ implemented by $\bpsi$. The connection to the relative modular operator is given by
\[
\Delta(\bpsi, \bphi)=\Delta(\psi/\phi'),
\] where $\psi\in (\M_*)^+$  is implemented by $\bpsi$ and $\phi'\in (\M_*')^+$ implemented by $\bphi$. Indeed, $R_{\phi'}(a\bphi)=aR_{\phi'}(\bphi)$ for $a\in \M$ and $R_{\phi'}(\bphi)R_{\phi'}(\bphi)^*=[\M'\bphi]=s(\bphi)\in \M$.
Then for $a\bphi\in \M\bphi$,
\begin{align*}\bra{a\bphi}\Delta(\psi/\phi')\ket{a\bphi}=
\psi(R_{\phi'}(a\bphi)R_{\phi'}(a\bphi)^*)=\psi(a s(\bphi)a^*)\pl,\end{align*}
which coincides with \eqref{modular}. Thus we verify that $\Delta(\bpsi,\bphi)=\Delta(\psi/\phi')$ for all $\bpsi,\bphi\in H$.\\

The relative modular operator $\Delta(\bpsi,\bphi)$ is independent of vector representations up to isometry.
Let $\phi$ and $\psi$ be two normal states of $\M$.
Let $\pi_1:\M\to B(H_1)$ (resp. $\pi_2:\M\to B(H_2)$) be a representation, and suppose that $\bphi_1,\bpsi_1\in H_1$ (resp. $\bphi_2,\bpsi_2\in H_2$) implement $\phi$ and $\psi$ via $\pi_1$ (resp. $\pi_2$). Define the  partial isometries $V_\phi: H_1\to H_2$ and $V_{\psi}: H_1\to H_2$ as follows:
\begin{align}
V_\phi(\pi_1(a)\bphi_1+\bet) & =\pi_2(a)\bphi_2\pl, \nonumber\\
 V_{\psi}(\pi_1(a)\bpsi_1+\bzeta) & =\pi_2(a)\bpsi_2\pl, \pl \forall \pl a\in \M, \label{Vop}
\end{align}
where $\bet\in [\pi_1(\M)\bphi_1]^\perp$ and $\bzeta\in [\pi_1(\M)\bpsi_1]^\perp$.
Let $S_{\bpsi_1,\bphi_1}$ and $\Delta(\bpsi_1,\bphi_1)$ (resp. $S_{\bpsi_2,\bphi_2}$ and $\Delta(\bpsi_2,\bphi_2)$) be the operators defined in \eqref{eq:S-op-def} and \eqref{eq:rel-mod-op-def} for $\pi_1(\M)$ (resp. $\pi_2(\M)$). Note that $\pi_1(s(\phi))=s(\bphi_1), \pi_2(s(\phi))=s(\bphi_2)$ and $V_\psi^*V_\psi=s(\bpsi_1)\supset \text{Ran}(S_{\bpsi_1,\bphi_1})$. We have
\begin{align*}
&S_{\bpsi_2,\bphi_2}V_{\phi}=V_{\psi}S_{\bpsi_1,\bphi_1}, \qquad \Delta(\bpsi_1,\bphi_1)=V_{\phi}^*\Delta(\bpsi_2,\bphi_2)V_{\phi}.
\end{align*}
\subsection{Standard form of von Neumann algebras}

The theory of the standard form of von Neumann algebras was developed by Araki \cite{Araki74}, Connes \cite{Connes76}, and Haagerup \cite{haagerup76}.
Recall that the standard form $(\M,H,J,P)$ of a von Neumann algebra $\M$ is given by an injective $*$-homomorphism~$\pi:\M \to B(H)$,  an anti-linear isometry $J$ on $H$, and a self-dual cone $P$ such that
\begin{itemize}
\item[i)] $J^2=1$, $J\M J=\M'$,
\item[ii)] $JaJ=a^*$ for $a\in \M\cap \M'$,
\item[iii)] $J\bxi=\bxi$ for $\bxi\in P$,
\item[iv)] $aJaJP=P$ for $a \in \M$.
\end{itemize}
The standard form is unique up to unitary equivalence.
For each normal state $\phi\in \M_*^{+}$, there exists a unique unit vector $\bxi_\phi\in P$ implementing $\phi$. We write the standard form of the relative modular operator as
\[
\Delta(\phi,\psi):=\Delta(\bxi_\phi,\bxi_\psi)\pl.
\]
By the symmetric role of $\M$ and $\M'$, we have
\begin{align}
\label{inv}\Delta(\phi,\psi)=J\Delta(\psi,\phi)^{-1}J.
\end{align}
In particular, the modular operator of $\phi$ is $\Delta(\phi,\phi)$
 and the modular automorphism group $\al_t^{\phi}:\M\to \M\pl$ is as follows:
 \[
 \al_t^{\phi}(x)=\Delta(\phi,\phi)^{-it} x \Delta(\phi,\phi)^{it}.
 \]

When $\M$ is semifinite equipped with a normal faithful semi-finite trace $\tau$, the standard form is basically given by the GNS construction. Define the $\tau$-inner product and $L_2$-norm respectively as
\[ \lan a,b\ran=\tau(a^*b)\pl,\qquad \norm{a}{2}^2=\lan a, a\ran\pl. \]
The $L_2$-space $L_2(\M)$ is a Hilbert space as the norm completion of $\left\{ a\in \M \, | \, \tau(s(|a|))<\infty\right\}$, where $s(|a|)$ is the support of $|a|$. The GNS representation $\pi:\M\to L_2(\M)$ has the following action for all $x\in \M$:
\[
\pi(x) a= xa .
\]
This gives a standard form $(\M,L_2(\M),J ,L_2(\M)^+)$,
where the anti-linear isometry is $J a= a^*$ and $L_2(\M)^+$ is the positive cone in $L_2$.


\subsection{Haagerup $L_p$-spaces}\label{a3}
In this part, we briefly review Haagerup's $L_p$-space \cite{lp} as our tool to Section \ref{section5}. We refer to \cite{terp} and \cite[Appendix]{Jenvcova18} for more details on this topic.

Let $\M\subset B(H)$ be a von Neumann algebra acting on a Hilbert space $H$. Given a distinguished normal faithful state $\omega\in D(\M)$,
we denote by $\al_t:=\al_t^\omega: \M\to \M, t\in \R$ the one parameter modular automorphism group. The crossed product
\[\mathcal{R}=\M\rtimes_\al \mathbb{R}\]
 is the von Neumann algebra acting on $L_2(\mathbb{R},H)$, generated by the operator $\pi(x), x\in \M$, and the operator $\la(s),s\in \mathbb{R}$, defined as follows: for all $\xi\in L_2(\mathbb{R},H)$ and $t\in \R$
\[
\pi(x)(\xi)(t) \coloneqq \al_{-t}(x)\xi(t)\pl,
\qquad
\la(s)(\xi)(t)=\xi(t-s)\pl.
\]
Note that $\pi$ is a normal faithful representation of $\M$ on $H\ten_2 L_2(\mathbb{R})\cong L_2(\mathbb{R},H)$ and $(\pi,\la(s))$ gives a covariant representation such that $\al_t(x)=\la(t)x\la(t)^*, x\in \M , t\in \mathbb{R}$. The dual action $\hat{\al}_t$ of $\mathbb{R}$ on $\R$ is a one-parameter automorpshim group of $\mathbb{R}$ on $\R$, implemented by the unitary representation $\{W(t)\}_{t\in \mathbb{R}}$ on $L_2(\mathbb{R},H)$,
\[\hat{\al}_t(x)=W(t)x W(t)^* ,\]
where
\[W(t)(\xi)(s)=e^{-its} \xi(s), \quad \xi \in L_2(\mathbb{R},H), \quad t,s\in \mathbb{R} \pl.\]
The dual action $\hat{\al}$ satisfies (and is uniquely determined by)
\[\hat{\al}_t(x)=x, \quad \hat{\al}_t(\la(s))=e^{-ist}\la(s)\pl, \quad x\in \M, \quad s,t\in \mathbb{R}\pl,\]
and $\M=\{x\in \mathcal{R}\pl | \pl \hat{\al}_t(x)=x\pl,  \forall t\in \mathbb{R}\}$.
This cross product algebra $\mathcal{R}$ admits a normal faithful semi-finite trace $\tau$ satisfying
\[ \tau\circ \hat{\al}_t=e^{-t}\tau, \quad \forall t\in \mathbb{R}.\]
For $0<p\le \infty$, the Haagerup noncommutative $L_p$-space is then defined as
\[L_p(\M,\omega)=\{ x\in L_0(\mathcal{R},\tau)\pl: \hat{\al}_t=e^{-t/p}x, \quad \forall t\in \mathbb{R}\}\pl.\]
We will suppress `` $\omega$'' in the notation $L_p(\M)$ since the $L_p$-spaces constructed for different states are isomorphic. The positive part is $L_p(\M)_+=L_p(\M)\cap L_0(\mathcal{R})_+.$ For all $\phi\in \M_*^+$, there exists a Radon-Nikodym derivative $h_\phi\in L_1(\mathcal{R},\tau)$ with respect to $\tau$ such that
\[ \tilde{\phi}(x)=\tau(h_\phi x), \quad x\in \mathcal{R}_+\pl, \quad \hat{\al}_t(h_\phi)=e^{-t}h_\phi\pl.\]
where $\tilde{\phi}$ is the dual weight of $\phi$ on $\mathcal{R}$.
This gives a linear bijection
\[\phi\in \M_*^+ \longleftrightarrow h_\phi\in L_1(\M)^+\]
This bijection further extends an identification
$\phi\in \M_* \leftrightarrow h_\phi\in L_1(\M)$
with the property $h_{x  \phi  y}=xh_{  \phi }y, x,y\in \M$.
Moreover, if $\phi=u|\phi|$ is the polar decomposition, $h_\phi=uh_{|\phi|}$. Using this linear bijection, the trace and $L_1$-norm on $L_1(\M)$ is defined as
\[ \tr(h_\phi):=\phi(1)\pl, \pl
\norm{h_\phi}{1}:=\tr(|h_\phi|)=\tr(h_{|\phi|})=|\phi|(1)=\norm{\phi}{\M_*}.
\]
For $a\in \mathcal{R}$, we have the polar decomposition $a=u|a|$ and for $p\in [1,\infty)$
\[a\in L_p(\M)\Longleftrightarrow |a|\in L_p(\M)\Longleftrightarrow |a|^{p}\in L_1(\M)\pl.\]
which leads to the $L_p$-norm, defined as
\[ \norm{a}{L_p(\M)}=\tr(|a|^p)^{1/p}\pl, \qquad \norm{a}{\infty}=\norm{a}{\M}.
\]
For $a\in L_p(\M),b\in L_q(\M)$ with $1/p+1/q=1$, $ab,ba\in L_1(\M)$ and
the trace ``$\tr$'' has the following tracial property:
\[\tr(ab)=\tr(ba)\pl.\]
In particular, the $L_2$-space $L_2(\M)$ is a Hilbert space with inner product $\lan a,b\ran=\tr(a^*b)$. Define the left regular representation
\[ \pi:\M \to B(L_2(\M)), \quad \pi(x)a=xa\pl.\]
and the anti-linear isometry
\[ J:L_2(\M)\to L_2(\M)\pl, \quad Ja=a^*\pl.\]
Identifying $\pi(\M)\cong \M$, the quadruple $(\M, L_2(\M),J, L_2(\M)^+)$ is a standard form of~$\M$. In particular, $JMJ$ acts on $L_2(\M)$ as the right multiplication $JxJa=ax^*$.

\end{document}